\documentclass[11pt]{article}
\usepackage[margin=0.8in]{geometry}
\usepackage{amsmath,amsfonts,amssymb,amsthm,mathtools}

\usepackage{amsthm}

\theoremstyle{plain} 
\newtheorem{theorem}{Theorem}[section]
\newtheorem{proposition}[theorem]{Proposition}
\newtheorem{lemma}[theorem]{Lemma}
\newtheorem{corollary}[theorem]{Corollary}

\theoremstyle{definition} 

\theoremstyle{remark} 
\newtheorem{remark}[theorem]{Remark}
\usepackage{mathrsfs}
\usepackage{hyperref}
\usepackage{tikz}
\usepackage{cite}
\usetikzlibrary{arrows}
\usetikzlibrary{arrows.meta}
\usetikzlibrary{cd}          
\usetikzlibrary{matrix}
\usetikzlibrary{decorations.pathmorphing}
\usetikzlibrary{calc}

\usepackage{tikz-cd}
\DeclareMathOperator{\Ad}{Ad}
\DeclareMathOperator{\ad}{ad}

\usepackage{tikz}
\usetikzlibrary{arrows.meta}
\usetikzlibrary{shapes.geometric}
\usetikzlibrary{calc}
\usetikzlibrary{fit}
\usetikzlibrary{positioning}
\usetikzlibrary{decorations.pathmorphing}
\usetikzlibrary{backgrounds}
\usepackage{comment}
\usepackage{enumerate}

\newcommand{\g}{\mathfrak{g}}
\newcommand{\Phat}{\widehat{\mathsf P}}

\title{Covariant Euler-Poincar\'e Reduction for Field Theories\\ on Centered Semidirect Products}

\author{
Miguel Ángel Berbel\thanks{Department of Applied Mathematics, ICAI School of Engineering, Comillas Pontifical University, Madrid,
Spain. {\tt\small \{maberbel,~arabella\}@comillas.edu}}
, Leonardo Colombo\thanks{Centre for Automation and Robotics (CSIC-UPM), Ctra. M300 Campo Real, Km 0,200, Arganda del Rey - 28500 Madrid, Spain, {\tt\small leonardo.colombo@csic.es}}
, Álvaro Rodríguez Abella$^*$}

\date{}

\begin{document}

\maketitle

\begin{abstract}

Let $G$ be a Lie group equipped with commuting left and right linear representations on a vector space $V$, and let $G\Join V$ be the associated centered semidirect product. We develop a reduction scheme for covariant field theories on principal bundles with structure group $G\Join V$ from both Lagrangian and Hamiltonian perspectives, obtaining the reduced Euler-Poincar\'e and Lie-Poisson field equations decomposed into their $\mathfrak g$- and $V$-components, where $\mathfrak g$ is the Lie algebra of $G$. We also derive the reduced multisymplectic form formula, as well as the Noether and energy-momentum balance laws. Our reduction results are compared with those for principal bundles with standard semidirect-product structure groups, and the role of the centered variables is investigated: they retain the original two-sided geometric data and separate the left and right contributions via the heart and diamond operators. The framework is then applied to a two-sided matrix microstructural model and to the second-order jet group on $\mathbb R^r$, $G_r^2$. For the matrix model, the conjugation representation decomposes $\operatorname{Mat}(3)$ into scalar, vector, and symmetric-traceless sectors, clarifying its precise relation to director and alignment-tensor variables without identifying it with standard nematodynamics. For the one-dimensional jet group on a $(1+1)$-dimensional
Minkowski base, we obtain a semilinear symmetric-hyperbolic
system and an exact traveling conversion front with explicit
energy-momentum exchange between the linear- and quadratic-jet
sectors. Its hyperbolic-target interpretation yields a conserved
positive energy for the linearized perturbations.
A nonlinear numerical experiment illustrates the reduced
evolution and reconstruction of the group-valued field,
with consistency assessed through mesh refinement and an
independent wave-map computation.
\end{abstract}

\section{Introduction}

The Euler-Poincar\'e reduction framework in mechanics and continuum dynamics provides a variational method for deriving reduced dynamical equations for systems on Lie groups whose Lagrangians are invariant by the left (or right) action of such Lie group. The Euler-Poincar\'e theory with advected parameters, developed by
Holm, Marsden, and Ratiu, relates broken symmetry to semidirect-product
Lie-Poisson structures and yields the associated diamond couplings
and Kelvin-Noether circulation theorem~\cite{HoMaRa1998}. The present paper instead reduces by the full centered structure group;
its vector-space component is part of the group-valued configuration,
rather than an advected parameter introduced by symmetry breaking. More recently, Holm and collaborators have emphasized a
composition-of-maps viewpoint in which iterated left or right
Lie group actions generate nested semidirect product dynamics
~\cite[\S 9.4]{Holm2025}; see also~\cite{HolmHuStreet2023}. This viewpoint raises the
corresponding question of how these reduction mechanisms
extend to covariant field theories.

The corresponding covariant extension is naturally formulated using
jet bundles and multisymplectic geometry. Beginning with the symplectic framework of Kijowski and Tulczyjew, the Lagrangian and Hamiltonian formulations of first-order field theories were systematically developed on jet and multimomentum bundles, where covariant variational principles yield the Euler--Lagrange and the Hamilton--De Donder--Weyl equations~\cite{KiTu1979,GoIsMa1998,GoIsMaMo2004,DeLeonMartinDiegoSantamaria2004,RomanRoy2009}. Marsden, Patrick, and Shkoller established the multisymplectic form
formula and related it to variational integrators and nonlinear PDEs
~\cite{MaPaSh1998}, while Marsden, Pekarsky, Shkoller, and West
extended this variational and multisymplectic framework to continuum
mechanics~\cite{MaPeShWe2001}. Within this setting, covariant Euler--Poincar\'e and Lagrange--Poincar\'e reduction for $G$-invariant field theories was developed in~\cite{CaRaSh2000,CaGaRa2001,CaRa2003,CaMa2008,
GayBalmazRatiu2010}: the quotient $(J^1P)/G$ is identified with the affine bundle $C(P)$ of principal connections, while the reduced variables and equations are expressed through the associated adjoint bundle and constrained variations. The Lagrange--Poincar\'e field equations of Ellis, Gay--Balmaz, Holm, and Ratiu further established the corresponding reconstruction conditions and the Kelvin--Noether theorem~\cite{EllisGayBalmazHolmRatiu2011}. More recently, covariant reduction by stages has provided a framework
for successive reductions, as well as the corresponding
reconstruction and Noether results~\cite{BerbelCastrillon2024}. Together, these developments provide the covariant geometric setting
needed to formulate reduced continuum models with internal degrees
of freedom. More recent developments extend covariant reduction to fiberwise gauge
symmetries modeled by Lie group bundles~\cite{CastrillonRodriguezGauge2024}.

A central class of such models is provided by continuum theories with internal microstructure. Classical micropolar and micromorphic models introduce independent microrotation and microdeformation fields in addition to the macroscopic deformation~\cite{Eringen1999,Capriz1989}, whereas complex fluids and liquid-crystal theories involve directors or more general order parameters~\cite{deGennesProst1993,Virga1994,Holm2002ComplexFluids,GBRa2009,GBRa2012}. In geometric-reduction formulations, these internal variables are represented as advected quantities or as sections of associated bundles, and their coupling to the macroscopic momentum gives rise to diamond and, in affine theories, cocycle terms. These terms encode forces and torques produced by broken symmetry; by contrast, couple stresses arise from the dependence of the reduced Lagrangian on spatial gradients of the internal variables and enter through the corresponding variational derivatives and covariant-divergence terms. Complementary geometric and Hamiltonian treatments of nonlinear elasticity clarify the material and convective representations of stress and momentum in solids, rods, and plates~\cite{Simo1984,Simo1988}. The composition-of-maps perspective emphasized by Holm and collaborators \cite{Holm2025} suggests a related but distinct structural question: how should covariant reduction be organized when the relevant variables carry simultaneous commuting left and right actions?

The algebraic structure naturally suggested by this question is the \emph{centered} semidirect product $G\Join V$, introduced by Colombo and Jacobs~\cite{CoJa2015}, in which $G$ acts on $V$ through commuting left and right representations. Although $G\Join V$ can be rewritten as an ordinary semidirect product---indeed, the change of variables $(g,v)\mapsto(g,v\cdot g^{-1})$ yields the induced representation $\rho(g)v=g\cdot v\cdot g^{-1}$---, the centered presentation keeps the original left and right actions explicit. In the split reduced equations, their combined contribution is organized by the heart and diamond operators, which encode the coupling between the $\mathfrak{g}$- and $V$-components. This two-sided structure arises naturally in an isotropy group of the second-order jet groupoid of a manifold, whose multiplication is induced by composition of diffeomorphisms and the chain rule~\cite{CoJa2015,KMS1993}. 

The purpose of this paper is to extend the Euler--Poincaré reduction theory on centered semidirect products from finite-dimensional mechanics to first-order covariant field theory. Let $P\to X$ be a principal bundle with structure group $G\Join V$, and let
$\mathcal L:J^1P\rightarrow\Lambda^{n+1}T^*X$
be a $G\Join V$-invariant Lagrangian density. The standard identification
$(J^1P)/(G\Join V)\simeq C(P)$
realizes the reduced bundle as the affine bundle of principal connections on $P\to X$. After choosing a reference connection, the reduced field may be represented by an $\operatorname{Ad}(P)$-valued one-form and, in a local trivialization, decomposed into $\mathfrak g$- and $V$-components. This makes it possible to incorporate the finite-dimensional centered reduction of Colombo and Jacobs~\cite{CoJa2015} into the covariant Euler--Poincar\'e framework on principal bundles~\cite{CaRaSh2000,CaRa2003,CaMa2008} and into the multisymplectic and covariant Hamiltonian formalisms of geometric field theory~\cite{KiTu1979,MaPaSh1998,FoRo2005}. The contribution is to make the reduced equations,
reconstruction conditions, and covariant currents explicit
in the original two-sided variables. For jet groups, this
description retains the distinction between the action on
the output and the actions on the input slots of the
quadratic jet.

\paragraph{Summary of the paper and main contributions.}

In Section~\ref{sec:centered-field-theory}, the centered semi-direct product is introduced, and the centered reduction framework for covariant field theories on principal bundles here developed is compared to the classical reduction framework by semidirect products. In particular, we prove naturality under the explicit isomorphism $G\Join V\simeq G\ltimes_\rho V$. This identifies precisely what
is inherited from standard covariant reduction and what the centered
presentation contributes: it retains the original two-sided variables and
resolves the effective semidirect action into its left and right momentum
transfer terms. 

Section~\ref{sec:covariantreduction} constitutes the theoretical core of the paper. Specifically, we achieve the following goals:
\begin{enumerate}
\item We formulate the covariant Euler--Poincar\'e field reduction theory in the centered
variables, thus deriving the constrained variations and reduced field equations in split
$\mathfrak g$- and $V$-form. 

\item Starting from the covariant Poincar\'e--Cartan form, we obtain the reduced multisymplectic form formula and its expression in centered variables. 

\item For variations satisfying the linearized reduced equations, this yields a closed covariant symplectic current whose $\mathfrak g$- and $V$-components display explicitly the two-sided coupling. Under the appropriate invariance and background assumptions, we also formulate the associated covariant Noether and energy-momentum balance laws.

\item We develop the Hamiltonian counterpart of the centered reduction theory on the multimomentum bundle. In particular, we formulate the centered Lie--Poisson field equations and show that, for hyperregular reduced Lagrangian densities, they are equivalent to the Euler--Poincar\'e field equations \emph{via} the covariant Legendre transformation.
\end{enumerate}

Section~\ref{sec:applications} provides two examples that illustrate the usefulness of the centered reduction framework:

\begin{enumerate}
\item
We first particularize the general construction to field theories on principal bundles whose structure group is a matrix centered product. This example extends the finite-dimensional setting in~\cite{CoJa2015} and makes the reduced equations, reconstruction conditions,
and covariant currents explicit in matrix variables.

\item 
For the standard $\mathrm{SO}(3)$ action on $\mathbb{R}^3$, we prove that no
independent nontrivial commuting right action exists, so the director
representation itself does not support the centered structure required by
the theory.  We therefore pass to
$\mathrm{SO}(3)\Join\operatorname{Mat}(3)$ and use the equivalent semidirect
coordinate $B=AR^{-1}$, on which $\mathrm{SO}(3)$ acts by conjugation.  The
invariant decomposition
$\operatorname{Mat}(3)
=
\mathbb{R}I
\oplus
\mathfrak{so}(3)
\oplus
\operatorname{Sym}_0(3)$ produces scalar, vector, and alignment-tensor sectors.  This establishes
a precise connection with variables occurring in micropolar and
liquid-crystal theories while also identifying the additional
broken-symmetry and order-parameter structure required to recover
standard director or $Q$-tensor models.
\end{enumerate}

Section~\ref{sec:higher-order-gauge} provides the main application of the centered reduction framework developed in this work. Specifically, we return to the geometric origin of the centered presentation
by considering the second-order jet group
$G_r^2\simeq\operatorname{GL}(r)\Join
(S^2(\mathbb R^r)^*\otimes\mathbb R^r)$.  Its multiplication is obtained
directly from the second-order chain rule.  Besides computing the general
field equations, reconstruction conditions, and multisymplectic current, we
specialize to $G_1^2$ on a $1+1$-dimensional Minkowski base.  The resulting
semilinear symmetric-hyperbolic system has an integrable traveling-wave reduction on a specified
zero-integration-constant branch and an exact heteroclinic--pulse solution which converts energy and
momentum pointwise between the linear- and quadratic-jet sectors while
preserving their totals. The associated hyperbolic-plane formulation yields a positive conserved energy for the linearized perturbations.
A nonlinear computation then illustrates the evolution of
compatible reduced data and reconstruction of the group-valued
field, with observed fourth-order self-convergence and agreement
with an independent wave-map solution under mesh refinement. 

Lastly, Section~\ref{sec:conclusions} summarizes the results and discusses directions for further work.


\section{The Centered Semidirect Product In Field Theory}
\label{sec:centered-field-theory}

In this section we collect the geometric ingredients for the covariant
Euler--Poincar\'e reduction on centered semidirect products. Specifically, we recall commuting left/right actions of a Lie group on a vector space, the centered semidirect product,
and the covariant reduction framework on principal bundles (see,  e.g.,~\cite{CoJa2015,CaRaSh2000,CaGaRa2001}). We also provide a comparison between our centered approach and classical reduction by semi-direct products.

\subsection{Preliminaries on Centered Semidirect Products}\label{subsec:centered}

Let $G$ be a Lie group with unit element $e\in G$ and Lie algebra $\mathfrak g$, and let $V$ be a
finite-dimensional real vector space.
Assume that $G$ acts on $V$ through smooth linear left and right
actions,
\[
G\times V\to V,\quad (g,v)\mapsto g\cdot v,
\qquad
V\times G\to V,\quad (v,g)\mapsto v\cdot g,
\]
and that these actions \emph{commute}:
\begin{equation*}
(g\cdot v)\cdot h = g\cdot (v\cdot h),
\qquad g,h\in G,~v\in V.
\end{equation*}
The induced infinitesimal actions are denoted by
\[
\xi\cdot v := \left.\frac{d}{d\epsilon}\right|_{\epsilon=0}\exp(\epsilon\xi)\cdot v,
\qquad
v\cdot \xi := \left.\frac{d}{d\epsilon}\right|_{\epsilon=0}v\cdot \exp(\epsilon\xi),\qquad \xi\in\mathfrak g,~v\in V.
\]

Under the previous conditions, the \emph{centered semidirect product}
$G\Join V$ is the product manifold $G\times V$ endowed with the following multiplication:
\begin{equation}\label{eq:centered-product}
(g_1,v_1)\cdot(g_2,v_2)
=
\big(g_1g_2,\ g_1\cdot v_2 + v_1\cdot g_2\big),
\end{equation}
for each $(g_1,v_1),(g_2,v_2)\in G\times V$. Commutativity of the left and right actions makes
\eqref{eq:centered-product} associative. Its identity is $(e,0)$ and
$(g,v)^{-1}= \bigl(g^{-1},-g^{-1}\cdot v\cdot g^{-1}\bigr)$. Since the multiplication and inversion maps are smooth, this defines a
Lie group structure on $G\times V$, denoted by $G\Join V$
~\cite{CoJa2015}.
The corresponding Lie algebra $\mathfrak g\Join V$ is the vector space $\mathfrak g\oplus V$
with the Lie bracket given by
\begin{equation}\label{eq:centered-bracket}
[(\xi_1,v_1),(\xi_2,v_2)]
=
\left([\xi_1,\xi_2],\
\xi_1\cdot v_2 + v_1\cdot \xi_2
-\xi_2\cdot v_1 - v_2\cdot \xi_1
\right),
\end{equation}
for each $(\xi_1,v_1),(\xi_2,v_2)\in\mathfrak g\oplus V$. Note that we use the same symbol $[\cdot,\cdot]$ for both the Lie bracket on $\mathfrak g$ and $\mathfrak g\Join V$ as the distinction will be clarified by the context. Similarly, $\langle\cdot,\cdot\rangle$ will denote the natural pairings on $V^*\times V$ and $\mathfrak g^*\times\mathfrak g$.
Following~\cite{CoJa2015}, the \emph{heart} and \emph{diamond} operators are defined by
\begin{align}
\heartsuit:\mathfrak g\times V^*\to V^*,\qquad\langle \xi\heartsuit \alpha,\ v\rangle
&:=
\langle \alpha,\ \xi\cdot v - v\cdot \xi\rangle,
\label{eq:heart-def}\\[2pt]
\diamondsuit:V\times V^*\to\mathfrak g^*,\qquad\langle v\diamondsuit \alpha,\ \xi\rangle
&:=
\langle \alpha,\ v\cdot \xi - \xi\cdot v\rangle,
\label{eq:diamond-def}
\end{align}
for each $\xi\in\mathfrak g$, $v\in V$ and $\alpha\in V^*$. Equivalently, by introducing the \emph{primal centered commutator} as
\begin{equation}\label{eq:primal-centered-commutator}
\spadesuit:V\times\mathfrak g\to V,\qquad v\spadesuit \xi := v\cdot \xi - \xi\cdot v,
\end{equation}
we can write 
\begin{equation} \label{eq:heart-diamond-spades}
    \langle \alpha,\ v\spadesuit\xi\rangle=\langle v\diamondsuit\alpha,\ \xi\rangle=-\langle \xi\heartsuit\alpha,\ v\rangle.
\end{equation}

In the same vein, for each $\xi\in\mathfrak g$, we adopt the coadjoint convention given by
\begin{equation*}
{\rm ad}_\xi^*:\mathfrak g^*\to\mathfrak g^*,\qquad\langle\ad_\xi^*(\mu),\eta\rangle=-\langle\mu,[\xi,\eta]\rangle,
\end{equation*}
where $\mu\in\mathfrak g^*$ and $\eta\in\mathfrak g$. This is the negative of the convention used in~\cite{CoJa2015}; coadjoint formulas from that reference are therefore translated accordingly.

\subsection{Centered Semidirect Reduction Framework on Principal Bundles and its Relation to Standard Semidirect Products} 
\label{subsec:reduction-principal-bundles-centered}

Let $X$ be an oriented $(n+1)$-dimensional manifold with volume form ${\rm vol}_X$, and
$P\to X$ be a (right) principal $(G\Join V)$-bundle.
The adjoint bundle,
\[
\Ad(P):=(P\times(\mathfrak g\Join V))/(G\Join V)\to X,
\]
is endowed with fiberwise Lie algebra structure induced by \eqref{eq:centered-bracket}.



In the covariant reduction framework on principal bundles \cite{CaGaRa2001}, the reduced field is a section of the bundle of connections $\mathcal C(P)=(J^1P)/(G\Join V)\to X$, hence it defines a principal connection. Since the bundle of connections is an affine bundle modeled on $T^*X\otimes\Ad(P)\to X$, by fixing a principal connection $\mathcal A\in\Gamma(\mathcal C(P))$, we may identify
\begin{equation}\label{eq:identificationCP}
\mathcal C(P)\overset{\mathcal A}{\simeq} T^*X\otimes\Ad(P),\qquad\sigma_x\mapsto\sigma_x-\mathcal A(x).    
\end{equation}
Therefore, relative to $\mathcal A$, reduced fields are represented by
$\Ad(P)$-valued one-forms $\xi\in\Omega^1(X,\Ad(P))$. Note that the choice of reference connection changes the affine coordinates but not the
reduced connection, i.e., it does not contain additional dynamical data.

Let $VP:=\ker(T\pi)\subset TP$ denote the vertical bundle of
$\pi:P\to X$. A vertical variation of a section $s\in\Gamma(P)$ is a
section $\delta s\in\Gamma(s^*VP)$. In the following, we consider only compactly supported variations. Vertical variations of sections of $P\to X$ induce \emph{constrained} variations of the reduced field as follows (see, for instance,~\cite{CaRaSh2000,CaGaRa2001})
\begin{equation}\label{eq:delta-xi-CGPR}
\delta^{\mathcal A}\xi
:=(\delta\xi)^v=
\nabla^{\mathcal A}\eta - [\xi,\eta],
\end{equation}
where $\eta\in\Gamma(\Ad(P))$ is compactly supported, $\nabla^{\mathcal A}$ is the covariant derivative on the adjoint bundle induced by $\mathcal A$ and $(\cdot)^v$ denotes the vertical variation on the affine bundle
$\mathcal C(P)\to X$, expressed on the affine coordinates determined by
$\mathcal A$.

After fixing a (possibly local) trivialization\footnote{In general, a trivialization of a principal bundle is only local, i.e., $P|_\mathcal{U}\simeq \mathcal{U}\times(G\Join V)$ for some open $\mathcal{U}\subset X$. From now on, for brevity, we omit this notation while taking into account that this may be the case.} $P\simeq X\times(G\Join V)$, the adjoint bundle is given by
$\Ad(P)\simeq X\times(\mathfrak g\Join V)$. Therefore, the principal connection and the reduced field may be expressed as $\mathcal A=(\mathcal A^{\mathfrak g},\mathcal A^V)$ and $\xi=(\xi^{\mathfrak g},\xi^V)$ for some (possibly local) forms $\mathcal A^{\mathfrak g},\xi^{\mathfrak g}\in\Omega^1(X,\mathfrak g)$ and $\mathcal A^V,\xi^V\in\Omega^1(X,V)$. Note that $\mathcal A^{\mathfrak g}$ defines a principal connection on the trivial bundle $X\times G\to X$. The induced covariant derivative on the adjoint bundle $X\times\mathfrak g\to X$ is denoted by $\nabla^{\mathcal A^{\mathfrak g}}$. On this trivialization, \eqref{eq:delta-xi-CGPR} splits into two parts:
\begin{align}
\delta^{\mathcal A}\xi^{\mathfrak g}
&=
\nabla^{\mathcal A^{\mathfrak g}}\eta^{\mathfrak g} - [\xi^{\mathfrak g},\eta^{\mathfrak g}],
\label{eq:delta-sigma}\\[2pt]
\delta^{\mathcal A}\xi^V
&=
{\rm d}\eta^V-\eta^V\spadesuit\mathcal A^{\mathfrak g}+\mathcal A^V\spadesuit\eta^{\mathfrak g}-\xi^V\spadesuit\eta^{\mathfrak g}+\eta^V\spadesuit\xi^{\mathfrak g},
\label{eq:delta-nu}
\end{align}
where $\eta=(\eta^{\mathfrak g},\eta^V)$ for some (possibly local) functions $\eta^{\mathfrak g}\in C^\infty(X,\mathfrak g)$ and $\eta^V\in C^\infty(X,V)$. Here, ${\rm d}:C^\infty(X,V)\to\Omega^1(X,V)$ denotes the exterior
derivative applied component-wise on the chosen trivialization. Note that we regard $\eta^{\mathfrak g}$ and $\eta^V$ as functions or sections of a trivial bundle depending on the context.


\begin{remark}
\label{rem:reference-connection-sign}
The affine formula~\eqref{eq:identificationCP} is intrinsic, but its local signs
must be distinguished from those of a conventional gauge potential. On the trivialization $P\simeq X\times(G\Join V)$, let
$A,K\in\Omega^1(X,\mathfrak g\Join V)$ be the usual local connection one-forms representing the reference and reduced
connections $\mathcal A,\sigma\in\Gamma(\mathcal C(P))$, respectively.
For a right principal action the corresponding horizontal splittings have
local vertical parts $-A$ and $-K$.  Consequently, the local representative of
the affine difference $\sigma-\mathcal A$ is $\xi=A-K$.
Thus, the reconstruction condition
$\operatorname{Curv}(\mathcal A+\xi)=0$ (see~\eqref{eq:reconstructioncompatibility} below), where $+$ denotes affine addition in
$\mathcal C(P)$, reads in conventional local connection forms
\begin{equation}
F_{A-\xi}
=\mathrm d(A-\xi)
+\frac12[A-\xi\wedge A-\xi]=0,
\label{eq:local-reconstruction-sign}
\end{equation}
where $[K\wedge K](u,v)=[K(u),K(v)]-[K(v),K(u)]$ for (local) fields $u,v\in\mathfrak X(X)$.

In particular, for the zero reference connection and the right logarithmic
derivative $\Xi=(\mathrm ds)s^{-1}$, the reduced connection form is
$K=-\Xi$, whence~\eqref{eq:local-reconstruction-sign} becomes
\begin{equation}
\mathrm d\Xi-\frac12[\Xi\wedge\Xi]=0.
\label{eq:right-MC-sign-convention}
\end{equation}

\end{remark}\medskip

Let $\ell:\mathcal C(P)\to\mathbb R$ be a reduced Lagrangian.
On the affine coordinates determined by $\mathcal A$, define the reduced
momenta as the fiber derivatives
\[
\mathsf P_{\mathfrak g}
:=
\frac{\delta\ell}{\delta\xi^{\mathfrak g}}:T^*X\otimes(\mathfrak g\Join V)\to TX\otimes\mathfrak g^*,
\qquad
\mathsf P_V
:=
\frac{\delta\ell}{\delta\xi^V}:T^*X\otimes(\mathfrak g\Join V)\to TX\otimes V^*.
\]
With respect to ${\rm vol}_X$, we denote by
$\operatorname{div}^{\mathcal A^{\mathfrak g}}$ and
$\operatorname{div}$ the negative formal adjoints of
$\nabla^{\mathcal A^{\mathfrak g}}$ and ${\rm d}$, respectively.

As we will see in Theorem~\ref{thm:covariant-reduction-centered} below, the split variational principle~\eqref{eq:delta-sigma}--\eqref{eq:delta-nu} yields the Euler--Poincar\'e field equations on a centered semidirect product:
\begin{align}\label{eq:centeredeqsg}
& \operatorname{div}^{\mathcal A^{\mathfrak g}}\left(\mathsf P_{\mathfrak g}\right)=\ad^{*}_{\xi^{\mathfrak g}}\left(\mathsf P_{\mathfrak g}\right)+\mathcal A^V\diamondsuit\mathsf P_V
-\xi^V\diamondsuit\mathsf P_V,\\[2pt]\label{eq:centeredeqsV}
& \operatorname{div}\left(\mathsf P_V\right)=\mathcal A^{\mathfrak g}\heartsuit\mathsf P_V-\xi^{\mathfrak g}\heartsuit\mathsf P_V.
\end{align}
The centered coupling terms in the previous equations
arise precisely from the mixed pieces in \eqref{eq:delta-nu} and are encoded by
the heart and diamond operators \eqref{eq:heart-def}, \eqref{eq:diamond-def},
with implicit contraction over the spacetime index.



\paragraph{Relation to standard semidirect products.}

Let us clarify the distinction between the reduction theory developed in this paper and the Euler--Poincar\'e reduction on principal bundles modeled by standard semi-direct products. The centered presentation does not define an inequivalent reduction theory, but its advantage is geometric and computationally organizational. Specifically, under the conditions of~\S\ref{subsec:centered}, the map
\[\rho:G\to\operatorname{GL}(V),\qquad\rho(g)v:=g\cdot v\cdot g^{-1}.
\]
defines a representation of $G$. Indeed, the commutativity of the left and right actions ensures that $\rho(g_1g_2)=\rho(g_1)\rho(g_2)$ for each $g_1,g_2\in G$. By differentiating at the identity, we obtain the induced representation of $\mathfrak g$,
\begin{align*}
\rho_*:\mathfrak g\to\operatorname{End}(V),\qquad \rho_*(\xi)v=\xi\cdot v-v\cdot\xi,
\end{align*}

\begin{proposition}
\label{prop:centered-standard-equivalence}
Under the previous conditions, the following statements hold true:
\begin{enumerate}
    \item The map
    \begin{equation*}
    \Phi:G\Join V\longrightarrow G\ltimes_\rho V,
    \qquad
    \Phi(g,v):=(g,v\cdot g^{-1}),
    \end{equation*}
    where $\ltimes_\rho$ denotes the semi-direct product given by the representation $\rho$, is a Lie group isomorphism, whose tangent map at the identity element is the identity,
    \begin{equation}
    F:=T_{(e,0)}\Phi:
    \mathfrak g\Join V\longrightarrow\mathfrak g\ltimes_{\rho_*}V,
    \qquad
    F(\xi,v)=(\xi,v).
    \label{eq:centered-standard-Lie-isomorphism}
    \end{equation}
    \item Let $P_\Phi$ denote the same total space as $P$, regarded as a principal
    $(G\ltimes_\rho V)$-bundle with the right action
    $p\cdot_\Phi(g,b):=p\cdot\Phi^{-1}(g,b)$ for each $p\in P$ and $(g,b)\in G\ltimes_{\rho}V$. Then the map $F$ induces natural
    isomorphisms of adjoint and connection bundles,
    \begin{equation}
    \operatorname{Ad}_{G\Join V}(P)
    \simeq
    \operatorname{Ad}_{G\ltimes_\rho V}(P_\Phi),
    \qquad
    \mathcal C_{G\Join V}(P)
    \simeq
    \mathcal C_{G\ltimes_\rho V}(P_\Phi).
    \label{eq:centered-standard-bundle-isomorphisms}
    \end{equation}
    \item For each $\xi\in\Omega^1(X,\Ad_{G\Join V}(P))$ and $\eta\in\Gamma(\Ad_{G\Join V}(P))$, we have
    \begin{equation}
    F\big(\nabla^{\mathcal A}\eta-[\xi,\eta]\big)
    =\nabla^{\mathcal A'}\eta'-[\xi',\eta'],
    \label{eq:centered-standard-variation-intertwining}
    \end{equation}
    where $\xi'=F(\xi)$, $\eta'=F(\eta)$, and
    $\mathcal A'=F(\mathcal A)$.
    \item For the reduced Lagrangian on the semidirect product,
    \[\ell':T^*X\otimes\Ad_{G\ltimes_\rho V}(P_\Phi)\to\mathbb R,\qquad\ell'(\xi'):=\ell(F^{-1}\xi'),
    \]
    the momenta satisfy
    \begin{equation}
    \frac{\delta\ell'}{\delta\xi'}=(F^{-1})^*\left(\frac{\delta\ell}{\delta\xi}\right).
    \label{eq:centered-standard-momentum-intertwining}
    \end{equation}
\end{enumerate}
\end{proposition}

\begin{proof}
It is readily seen that the inverse of $\Phi$ reads $\Phi^{-1}(g,b)=(g,b\cdot g)$ for each $(g,b)\in G\ltimes_\rho V$, whence $\Phi$ is a diffeomorphism. Using commutativity of the left and right actions, we have
\begin{align*}
\Phi\big((g_1,v_1)(g_2,v_2)\big)
&=\big(g_1g_2,
(g_1\cdot v_2+v_1\cdot g_2)\cdot(g_1g_2)^{-1}\big)\\
&=\big(g_1g_2,
v_1\cdot g_1^{-1}
+\rho(g_1)(v_2\cdot g_2^{-1})\big),
\end{align*}
for each $(g_1,v_1),(g_2,v_2)\in G\Join V$, which is the multiplication in $G\ltimes_\rho V$, thus making $\Phi$ a Lie group isomorphism. By differentiating at the
identity, we obtain~\eqref{eq:centered-standard-Lie-isomorphism}; in particular,
the two Lie brackets coincide on the underlying vector space
$\mathfrak g\oplus V$. Relabelling the principal action by a group
isomorphism gives the bundle maps in
\eqref{eq:centered-standard-bundle-isomorphisms}.  Naturality of the adjoint
connection, bracket, covariant divergence, and curvature under $F$ proves
\eqref{eq:centered-standard-variation-intertwining}.  Finally, the chain rule for
$\ell'=\ell\circ F^{-1}$ gives
\eqref{eq:centered-standard-momentum-intertwining}.
\end{proof}

Under the isomorphisms~\eqref{eq:centered-standard-Lie-isomorphism} and \eqref{eq:centered-standard-bundle-isomorphisms}, reference connections, constrained variations, covariant
momenta, Euler--Poincar\'e equations, and reconstruction conditions are
intertwined. Consequently, the centered equations
\eqref{eq:centeredeqsg}--\eqref{eq:centeredeqsV} are exactly the component
form of standard covariant Euler--Poincar\'e field equations for
a $(G\ltimes_\rho V)$-invariant Lagrangian density.

In the centered
coordinate $v$, the two actions entering the original construction remain
separate, whereas the semidirect coordinate
$b=v\cdot g^{-1}$ combines them into the effective representation $\rho$.
The heart and diamond operators resolve the dual of
$\rho_*(\xi)v=\xi\cdot v-v\cdot\xi$ into the contributions of the original
left and right actions.  This is particularly useful for tensor
representations, as we will illustrate on Section~\ref{sec:higher-order-gauge}: for $G_r^2$, the centered coordinate is the raw Hessian
$T=D^2\varphi(0)$ and the two actions record separately the transformation of
its output and of its two input slots, while
$B=T\cdot g^{-1}$ is a normalized semidirect coordinate.  Thus the centered
formalism preserves the variables supplied directly by the chain rule and
makes their individual momentum-transfer terms visible.  No claim of
inequivalent dynamics or of a new class of Lie groups is intended.

\begin{remark}
Semidirect-product symmetry has appeared previously in covariant
field-theoretic reduction in settings different from the one
considered here. Castrill\'on L\'opez and Ratiu~\cite{CaRa2003}
discuss the semidirect-product case within covariant
Lagrange--Poincar\'e reduction on principal bundles. Castrill\'on
L\'opez, Chac\'on, and Garc\'ia~\cite{CaChGa2013} consider an affine
principal bundle $Q\times_X E\to X$ with structure group $H\ltimes V$,
constructed from an $H$-principal bundle $Q\to X$ and the associated
vector bundle $E\to X$, and reduce the corresponding Lagrangian theory
by the subgroup $H$. The Hamiltonian counterpart of this affine-principal-bundle reduction was developed recently in~\cite{BerbelCastrillonAffine2025}. By contrast, here
$P\to X$ is a principal $(G\Join V)$-bundle and the reduction is
performed by the full centered structure group; the split
$\mathfrak g$- and $V$-equations reflect the two commuting actions
defining $G\Join V$.
\end{remark}\medskip

\section{Covariant Reduction on Principal Bundles Modeled by Centered Semidirect Products}\label{sec:covariantreduction}

This section constitutes the theoretical core of the paper. After recalling the geometry of the reduced configuration bundle, we compute the split Euler--Poincar\'e field equations on centered semidirect products. Next, we derive the reduced multisymplectic form formula, demonstrating how the centered couplings enter the symplectic current, and establish the corresponding covariant balance laws for energy and momentum. Finally, we complete the geometric picture by formulating the
covariant Lie--Poisson equations on the reduced multimomentum bundle
associated with the centered semidirect-product symmetry. Our treatment combines the covariant Euler--Poincar\'e framework on
principal bundles~\cite{CaGaRa2001} with the multisymplectic framework
for first-order field theories
~\cite{MaPaSh1998,GoIsMa1998,GoIsMaMo2004}, in a form adapted to the
special features of the centered product $G \Join V$. When the reduced variational principle and field equations are
expressed in the split $\mathfrak g$- and $V$-variables associated
with the centered presentation, the two actions give rise to the
heart and diamond coupling terms.

\medskip

Let $G$ be a Lie group with Lie algebra $\mathfrak g$, and let $V$ be
a vector space endowed with commuting left and right linear actions
of $G$. We denote the corresponding centered semidirect product by
$G\Join V$. Let $X$ be an oriented $(n+1)$-dimensional manifold with volume form ${\rm vol}_X$, and
$P\to X$ be a (right) principal $(G\Join V)$-bundle, and denote the associated adjoint bundle by $\Ad(P)\to X$. Recall that the bundle of connections $\mathcal C(P)=(J^1 P)/(G\Join V)\to X$ is an affine bundle modeled on $T^{*}X\otimes \Ad(P)$ (see, for instance,~\cite{CaGaRa2001}). Lastly, let $P\simeq X\times (G\Join V)$ be a (possibly local) trivialization of $P\to X$, which induces the trivialization $\Ad(P)\simeq X\times(\mathfrak g\Join V)$ of the adjoint bundle. Recall that, if $P\to X$ admits a global section, then it is globally trivializable.

\subsection{Euler-Poincar\'e Field Equations}
\label{sec:eulerpoincare}

Recall that the action of $G\Join V$ on $P$ naturally lifts to the jet bundle $J^{1}P\to X$. Let
\[
\mathcal L=L\,{\rm vol}_X:J^{1}P \longrightarrow \Lambda^{n+1}T^*X
\]
be a Lagrangian density, where $L:J^{1}P\to\mathbb R$, that is invariant under the natural right action of
$G\Join V$ on $J^{1}P$. The invariance of $L$ defines the reduced Lagrangian function
\begin{equation}\label{eq:reducedlagrangian}
\ell:\ \mathcal C(P)\longrightarrow\mathbb R, 
\qquad
[j^{1}_{x}s]\mapsto\ell([j^{1}_{x}s])=L(j^{1}_{x}s).  
\end{equation}
Upon choosing a principal connection $\mathcal A\in\Gamma(\mathcal C(P))$, we may equivalently regard $\ell:T^*X\otimes\Ad(P)\to\mathbb R$ thanks to \eqref{eq:identificationCP}. The reduced variable associated with a local section
$s\in\Gamma(P)$ is denoted by $\xi=[j^1s]-\mathcal A\in\Omega^1(X,\Ad(P))$.

Let us now formulate the covariant reduced variational principle and derive the covariant Euler--Poincar\'e equations on a principal bundle modeled on a centered semidirect product. 
The constrained variations \eqref{eq:delta-xi-CGPR} entering the reduced variational principle are induced by variations of the original section, i.e., compactly supported vertical variations $\delta s\in\Gamma(s^*VP)$ of $s\in\Gamma(P)$, and may be expressed in terms of the
covariant derivative induced by $\mathcal A$, which is given by $\nabla^{\mathcal A}\eta={\rm d}\eta + [\mathcal A,\eta]$ for each $\eta\in\Gamma(\Ad(P))$.

\medskip

To connect with local formulas, a (possibly local) section can be written as $s=(g,v)\in\Gamma(P)$ for some (possibly local) $g\in C^\infty(X,G)$ and $v\in C^\infty(X,V)$.
Analogously, the principal connection and the reduced section read $\mathcal A=(\mathcal A^{\mathfrak g},\mathcal A^V)$ and $\xi=(\xi^{\mathfrak g},\xi^V)$, respectively, for some (possibly local) $\mathcal A^{\mathfrak g},\xi^{\mathfrak g}\in\Omega^1(X,\mathfrak g)$ and $\mathcal A^V,\xi^V\in\Omega^1(X,V)$. Recall that $\mathcal A^{\mathfrak g}$ defines a (local)
connection form on the trivial principal bundle $X\times G\to X$ whose covariant derivative on the adjoint bundle $X\times\mathfrak g\to X$ is denoted by $\nabla^{\mathcal A^{\mathfrak g}}$. 
By writing $\eta=(\eta^{\mathfrak g},\eta^V)$, it is easy to check that
\begin{equation}\label{eq:cov-der-A}
\nabla^{\mathcal A}\eta=\left(\nabla^{\mathcal A^{\mathfrak g}}\eta^{\mathfrak g},~{\rm d}\eta^V-\eta^V\spadesuit\mathcal A^{\mathfrak g}+\mathcal A^V\spadesuit\eta^{\mathfrak g}\right),
\end{equation}
where \eqref{eq:centered-bracket} and \eqref{eq:primal-centered-commutator} have been taken into account. From this and \eqref{eq:centered-bracket}, it is readily seen that the constrained variation~\eqref{eq:delta-xi-CGPR} splits into \eqref{eq:delta-sigma} and \eqref{eq:delta-nu}.

In the (local) trivialization $P\simeq X\times(G\Join V)$, the partial derivatives of the reduced Lagrangian, which are typically known as the \emph{momenta}, \eqref{eq:reducedlagrangian} are given by
\begin{align*}
& {\mathsf P}_{\mathfrak g}:=\frac{\delta\ell}{\delta\xi^{\mathfrak g}}:T^*X\otimes(\mathfrak g\Join V)\to TX\otimes\mathfrak g^*, && \left\langle\frac{\delta\ell}{\delta\xi^{\mathfrak g}}
(\xi^{\mathfrak g},\xi^V),\dot\xi^{\mathfrak g}\right\rangle=\left.\frac{d}{dt}\right|_{t=0}\ell(\xi^{\mathfrak g}+t\dot\xi^{\mathfrak g},\xi^V),\\
& {\mathsf P}_V :=\frac{\delta\ell}{\delta\xi^V}:T^*X\otimes(\mathfrak g\Join V)\to TX\otimes V^*, && \left\langle\frac{\delta\ell}{\delta\xi^V}
(\xi^{\mathfrak g},\xi^V),\dot\xi^V\right\rangle=\left.\frac{d}{dt}\right|_{t=0}\ell(\xi^{\mathfrak g},\xi^V+t\dot\xi^V),
\end{align*}
for each $(\xi^{\mathfrak g},\xi^V)\in T_x^*X\otimes(\mathfrak g\Join V)$, where $\dot\xi^{\mathfrak g}\in T_x^*X\otimes\mathfrak g$, and
$\dot\xi^V\in T_x^*X\otimes V$. 

Lastly, with respect to ${\rm vol}_X$, we denote by
$\operatorname{div}^{\mathcal A^{\mathfrak g}}:
\Gamma(TX\otimes\mathfrak g^*)\to C^\infty(X,\mathfrak g^*)$
and
$\operatorname{div}:
\Gamma(TX\otimes V^*)\to C^\infty(X,V^*)$
the negative formal adjoints of
$\nabla^{\mathcal A^{\mathfrak g}}:
C^\infty(X,\mathfrak g)\to\Omega^1(X,\mathfrak g)$
and
${\rm d}:C^\infty(X,V)\to\Omega^1(X,V)$,
respectively.

We are ready to present the Euler--Poincar\'e reduction theorem for principal bundles on centered semidirect products.



\begin{theorem}
\label{thm:covariant-reduction-centered}
Let $X$ be an oriented manifold with volume form ${\rm vol}_X$, let
$P\to X$ be a right principal $(G\Join V)$-bundle, and let $P\simeq\mathcal U\times(G\Join V)$ be a (possibly local) trivialization of $P\to X$, where $\mathcal U\subset X$ is an
open set.  Let
$\mathcal L=L\,{\rm vol}_X:J^1P\to\Lambda^{n+1}T^*X$ be a
$(G\Join V)$-invariant Lagrangian density and let
$\ell:\mathcal C(P)\to\mathbb R$ be its reduced Lagrangian~\eqref{eq:reducedlagrangian}.  Given a
principal connection $\mathcal A$ on $P$, regard $\ell:T^*X\otimes\Ad(P)\to\mathbb R$ thanks to~\eqref{eq:identificationCP}.  Let
$s=(g,v)\in\Gamma_{\mathcal U}(P)$ be a (local) section and
$\xi=(\xi^{\mathfrak g},\xi^V)=[j^1s]-\mathcal A\in
\Omega^1(U,\operatorname{Ad}(P)|_U)$ its (local) reduced field. 

\medskip

Then the following statements are equivalent:

\begin{enumerate}[(i)]
\item The section $s$ is critical for the Hamilton principle, i.e.,
\[
\delta\mathcal S[s]=\delta\int_U L(j^1 s)\,{\rm vol}_X=0,
\]
for compactly supported vertical variations
$\delta s\in\Gamma_c(s^*VP)$.

\item The section $s$ satisfies the Euler--Lagrange equations for $L$.

\item The reduced section $\xi$ is critical for the reduced variational principle, i.e.,
\[
\delta\mathcal S_{\mathrm{red}}[\xi]=\delta\int_U \ell(\xi)\,{\rm vol}_X=0,
\]
for compactly supported variations of the form \eqref{eq:delta-xi-CGPR}. 

\item The reduced section $\xi$ satisfies the Euler--Poincar\'e field equations \cite{CaRaSh2000,CaGaRa2001}. In the trivialization, we obtain the \emph{Euler--Poincar\'e equations on the centered semidirect product} \eqref{eq:centeredeqsg} and \eqref{eq:centeredeqsV}. 
\end{enumerate}
\end{theorem}

\begin{proof}
The equivalence between $(i)$ and $(ii)$ is classical in geometric field theory~\cite{GoIsMa1998}. Similarly, the equivalence between $(i)$ and $(iii)$ comes from the standard Euler--Poincar\'e reduction procedure (cf.~\cite{CaGaRa2001,CaRaSh2000}). The expression for the variations on the trivialization is a straightforward computation using \eqref{eq:centered-bracket}, \eqref{eq:primal-centered-commutator} and \eqref{eq:cov-der-A}.
Lastly, the equivalence between $(iii)$ and $(iv)$ comes from substituting \eqref{eq:delta-sigma} and \eqref{eq:delta-nu} into the reduced variational principle; namely,
\begin{align*}
0 & =\delta\int_U \ell(\xi^{\mathfrak g},\xi^V)\,{\rm vol}_X\\
& =\int_U\left(
\Big\langle {\mathsf P}_{\mathfrak g},~\nabla^{\mathcal A^{\mathfrak g}}\eta^{\mathfrak g} - [\xi^{\mathfrak g},\eta^{\mathfrak g}]\Big\rangle
+
\Big\langle {\mathsf P}_V,~{\rm d}\eta^V-\eta^V\spadesuit\mathcal A^{\mathfrak g}+\mathcal A^V\spadesuit\eta^{\mathfrak g}-\xi^V\spadesuit\eta^{\mathfrak g}+\eta^V\spadesuit\xi^{\mathfrak g}
\Big\rangle
\right)\mathrm{vol}_X .
\\
& =\int_U\bigg(
\Big\langle-\operatorname{div}^{\mathcal A_{\mathfrak g}}({\mathsf P}_{\mathfrak g})+\ad_{\xi^{\mathfrak g}}^*({\mathsf P}_{\mathfrak g})+\mathcal A^V\diamondsuit{\mathsf P}_V-\xi^V\diamondsuit{\mathsf P}_V,~\eta^{\mathfrak g}\Big\rangle\\
& \qquad+
\Big\langle-\operatorname{div}({\mathsf P}_V)+\mathcal A^{\mathfrak g}\heartsuit{\mathsf P}_V-\xi^{\mathfrak g}\heartsuit{\mathsf P}_V,~\eta^V\Big\rangle\bigg)\mathrm{vol}_X,
\end{align*}
where \eqref{eq:heart-def}, \eqref{eq:diamond-def} and \eqref{eq:primal-centered-commutator}, as well as covariant integration by parts, have been used (note that boundary terms vanish as variations are compactly supported), and terms have been rearranged. Since $\eta^{\mathfrak g}$ and $\eta^V$ are arbitrary compactly
supported sections, the fundamental lemma of the calculus of
variations yields the equations in (iv). Conversely, if these
equations hold, the last integral vanishes for every admissible
variation, which proves the equivalence between (iii) and (iv).
\end{proof}



Equivalently, setting
\[
\mathsf P:=\frac{\delta\ell}{\delta\xi}:T^*X\otimes\Ad(P)\to TX\otimes\Ad(P)^*,\qquad \left\langle\frac{\delta\ell}{\delta\xi}
(\xi),\dot\xi\right\rangle=\left.\frac{d}{dt}\right|_{t=0}\ell(\xi+t\dot\xi),
\]
for each $\xi,\dot\xi\in T_x^*X\otimes\Ad(P)$, the reduced equations admit the intrinsic form
\begin{equation}\label{eq:intrinsic-EP}
\operatorname{div}^{\mathcal A}(\mathsf P)
=
\operatorname{ad}^*_\xi(\mathsf P).
\end{equation}
Its local decomposition with respect to
$\mathfrak g\Join V$ is precisely
\eqref{eq:centeredeqsg}--\eqref{eq:centeredeqsV}.

The theorem above is formulated for a reduced field
$\xi=[j^1s]-\mathcal A$ already induced by a section $s$. Conversely,
an arbitrary (local) reduced field
$\xi\in\Omega^1(X,\Ad(P))$ is locally of this form precisely when the
connection $[j^1 s]=\mathcal A+\xi$ is flat, that is,
\begin{align}\label{eq:reconstructioncompatibility}
\operatorname{Curv}(\mathcal A+\xi)=0.
\end{align}
Global reconstruction may additionally depend on the topology and
holonomy of the principal bundle; see
\cite{CaRaSh2000,CaGaRa2001}.

\begin{remark}
If the principal bundle is trivial and the reference connection is
chosen to be the trivial connection, $\mathcal A=0$, then the reduced equations boil down to
\begin{align*}
& \operatorname{div}({\mathsf P}_{\mathfrak g})=
\ad^{*}_{\xi^{\mathfrak g}}({\mathsf P}_{\mathfrak g})-
\xi^V\diamondsuit{\mathsf P}_V,\\[2pt]
& \operatorname{div}({\mathsf P}_V)=
-\xi^{\mathfrak g}\heartsuit{\mathsf P}_V.
\end{align*}

\end{remark}\medskip

\begin{remark}
For a fixed local section $s$, one may choose the reference connection
to be $\mathcal A=[j^1s]$ and keep this reference fixed when taking
variations. At the section under consideration, one then has $\xi=0$,
and the reduced equations become
\begin{align*}
& \operatorname{div}^{\mathcal A^{\mathfrak g}}({\mathsf P}_{\mathfrak g})=\mathcal A^V\diamondsuit{\mathsf P}_V,\\[2pt]
& \operatorname{div}({\mathsf P}_V)=\mathcal A^{\mathfrak g}\heartsuit{\mathsf P}_V,
\end{align*}
which are the local decomposition of $\operatorname{div}^{\mathcal A}(\mathsf P)=0$ with respect to
$\mathfrak g\Join V$.

\end{remark}\medskip

\subsection{Reduced Multisymplectic Form Formula}\label{sec:multisymplecticformula}

It is standard to associate the Poincar\'e--Cartan geometry with a
first-order Lagrangian field theory and to obtain the corresponding
multisymplectic form formula; see, for example,
\cite{MaPaSh1998,MaPeShWe2001,GoIsMa1998}. Here the underlying multisymplectic
identity is standard, while its expression in reduced
Euler--Poincar\'e variables reflects the centered semidirect-product
geometry. In particular, the reduced symplectic current displays
explicitly the heart and diamond couplings.

We retain the setting and notation of Subsection~\ref{sec:eulerpoincare}, with
$\mathcal L=L\,{\rm vol}_X$. 
Let $\Theta_{\mathcal L}\in\Omega^{n+1}(J^1P)$ be the Poincar\'e--Cartan $(n+1)$-form on $J^1P$ associated with $\mathcal L$, and set $\Omega_{\mathcal L}:=-{\rm d}\Theta_{\mathcal L}$, the associated Poincar\'e--Cartan $(n+2)$-form. The form $\Omega_{\mathcal L}$ is closed and, for $\dim (X)\geq 2$, it is multisymplectic whenever the Lagrangian is regular.

As stated in Theorem \ref{thm:covariant-reduction-centered}, a section $s:X\to P$ is a critical point of
$\mathcal S[s]=\int_X\mathcal L(j^1s)$ if and only if it satisfies the
Euler--Lagrange equations. In multisymplectic form, these can be
expressed as
\begin{equation}\label{eq:EL-multisymplectic}
(j^1s)^*(\iota_Z\Omega_{\mathcal L})=0
\quad
\text{for all $\pi_{X,J^1P}$-vertical vector fields $Z$ on $J^1P$}.
\end{equation}

Let $s$ be a solution and let $U,W$ be vertical first variations along
$s$, that is, vertical variations satisfying the linearized
Euler--Lagrange equations along $s$. Define the associated symplectic
current by
\begin{equation}\label{eq:symplectic-current-cov}
\omega(U,W)
:=
(j^1s)^*
\left(
\iota_{j^1W}\iota_{j^1U}\Omega_{\mathcal L}
\right)
\in\Omega^n(X).
\end{equation}
By the multisymplectic form formula
\cite{MaPaSh1998,MaPeShWe2001,GoIsMa1998}, this current is closed:
\begin{equation}\label{eq:MSFF}
\mathrm d\,\omega(U,W)=0.
\end{equation}

We now express the symplectic current in the reduced variables of
Subsection~\ref{sec:eulerpoincare}. We keep the reference connection
$\mathcal A$ fixed and use $\mathsf P$ for the reduced momentum
evaluated at $\xi=[j^1s]-\mathcal A$,
\[
\mathsf P
=
\frac{\delta\ell}{\delta\xi}(\xi)
\in\Gamma\!\left(TX\otimes\Ad(P)^*\right).
\]
Using the fixed volume form ${\rm vol}_X$, we denote by
$\widehat{\mathsf P}
\in\Omega^n(X,\Ad(P)^*)$
the corresponding momentum $n$-form, defined by
\begin{equation}\label{eq:momentum-n-form}
\langle\widehat{\mathsf P},\eta\rangle
=
\iota_{\langle\mathsf P,\eta\rangle}{\rm vol}_X,
\qquad
\eta\in\Gamma(\Ad(P)).
\end{equation}
Thus $\widehat{\mathsf P}$ is simply the differential-form
representation of $\mathsf P$ determined by ${\rm vol}_X$.
Locally, if $\mathsf P=\mathsf P^\mu\partial_\mu$, then $\widehat{\mathsf P}
=\mathsf P^\mu\,{\rm d}^n x_\mu$, where
${\rm d}^n x_\mu:=\iota_{\partial_\mu}{\rm vol}_X.$

For arbitrary vertical variations $U,W$ along $s$, let
$\eta_U,\eta_W\in\Gamma(\Ad(P))$ denote their corresponding reduced
variation parameters. By \eqref{eq:delta-xi-CGPR},
\begin{equation}\label{eq:linearized-constraint}
\delta_U^{\mathcal A}\xi
=
\nabla^{\mathcal A}\eta_U-[\xi,\eta_U],
\qquad
\delta_W^{\mathcal A}\xi
=
\nabla^{\mathcal A}\eta_W-[\xi,\eta_W].
\end{equation}
Let
\[
\mathbb H_\ell(\xi)
:=
D_\xi\!\left(\frac{\delta\ell}{\delta\xi}\right)
\]
denote the fiber Hessian of $\ell$. The induced momentum variation is
\begin{equation}\label{eq:momentum-linearization}
\delta_U\mathsf P
=
\mathbb H_\ell(\xi)\big[\delta_U^{\mathcal A}\xi\big],
\end{equation}
with the analogous expression for $\delta_W\mathsf P$.

For a vertical variation $U$ along $s$, define the presymplectic
potential current by
\begin{equation}\label{eq:reduced-potential-current}
\theta_s(U)
:=
(j^1s)^*
\left(
\iota_{j^1U}\Theta_{\mathcal L}
\right)
\in\Omega^n(X).
\end{equation}
Its reduced expression is
\begin{equation}\label{eq:reduced-potential-current-local}
\theta_s(U)
=
\langle\widehat{\mathsf P},\eta_U\rangle
=
\iota_{\langle\mathsf P,\eta_U\rangle}{\rm vol}_X.
\end{equation}
where the pairing is taken in the adjoint-bundle factors, so that
$\langle\mathsf P,\eta_U\rangle$ is a vector field on $X$.

Let $s_{\varepsilon,\lambda}$ be a two-parameter variation of
$s=s_{0,0}$ whose infinitesimal variations are $U$ and $W$, with
commuting parameter derivatives. Since
$\Omega_{\mathcal L}=-\mathrm d\Theta_{\mathcal L}$,
\begin{equation}\label{eq:omega-from-theta}
\omega(U,W)
=
\delta_W\theta_s(U)-\delta_U\theta_s(W).
\end{equation}
The reduced variation parameters satisfy the Maurer--Cartan identity
\begin{equation}\label{eq:variation-MC}
\delta_U\eta_W-\delta_W\eta_U
=
[\eta_U,\eta_W].
\end{equation}
Indeed, locally $\eta$ is the right-trivialized variation associated
with the principal $(G\Join V)$-action, and the sign in
\eqref{eq:variation-MC} is consistent with the constrained variation
$\delta^{\mathcal A}\xi
=
\nabla^{\mathcal A}\eta-[\xi,\eta]$
used in \eqref{eq:delta-xi-CGPR}.

\begin{lemma}\label{lem:reduced-symplectic-current}
Let $s:X\to P$ be a section and let $U,W$ be vertical variations along
$s$, with reduced variation parameters
$\eta_U,\eta_W\in\Gamma(\Ad(P))$. Then the symplectic current
\eqref{eq:symplectic-current-cov} admits the reduced expression
\begin{equation}\label{eq:reduced-symplectic-current}
\omega(U,W)
=
\langle\delta_W\Phat,\eta_U\rangle
-
\langle\delta_U\Phat,\eta_W\rangle
-
\langle\Phat,[\eta_U,\eta_W]\rangle.
\end{equation}
\end{lemma}

\begin{proof}
Since ${\rm vol}_X$ is fixed, differentiating
\eqref{eq:reduced-potential-current-local} and using
\eqref{eq:omega-from-theta}, we obtain
\begin{equation*}
\omega(U,W)
=
\langle\delta_W\Phat,\eta_U\rangle
-
\langle\delta_U\Phat,\eta_W\rangle
+
\langle
\Phat,\delta_W\eta_U-\delta_U\eta_W
\rangle.
\end{equation*}
Then, the result follows from the Maurer--Cartan identity \eqref{eq:variation-MC}.
\end{proof}

We next consider the linearized reduced equations
(cf.~\cite{CaGaRa2001}). Since the reference connection $\mathcal A$ is kept fixed, linearization of \eqref{eq:intrinsic-EP} along $U$ gives
\begin{equation}\label{eq:linearized-EP}
\operatorname{div}^{\mathcal A}(\delta_U\mathsf P)
=
\operatorname{ad}^*_\xi(\delta_U\mathsf P)
+
\operatorname{ad}^*_{\delta_U^{\mathcal A}\xi}\mathsf P ,
\end{equation}
together with
\begin{equation}\label{eq:linearized-constraint}
\delta_U^{\mathcal A}\xi
=
\nabla^{\mathcal A}\eta_U-[\xi,\eta_U],
\end{equation}
and
\begin{equation}\label{eq:momentum-linearization}
\delta_U\mathsf P
=
\mathbb H_\ell(\xi)
\big[\delta_U^{\mathcal A}\xi\big].
\end{equation}
The same relations hold for $W$. We refer to
\eqref{eq:linearized-constraint},
\eqref{eq:momentum-linearization}, and
\eqref{eq:linearized-EP}
as the linearized reduced Euler--Poincar\'e system.

\begin{proposition}\label{prop:reduced-MSFF}
Let $s:X\to P$ be a solution of the Euler--Lagrange equations and let
$\xi=[j^1s]-\mathcal A$ be the corresponding reduced solution of
\eqref{eq:intrinsic-EP}. Let $U,W$ be vertical first variations
along $s$, with reduced variation parameters
$\eta_U,\eta_W\in\Gamma(\Ad(P))$. Then the reduced symplectic current
\eqref{eq:reduced-symplectic-current} is closed:
\begin{equation}\label{eq:reduced-MSFF}
\mathrm d\,\omega(U,W)=0.
\end{equation}
\end{proposition}

\begin{proof}
Since ${\rm vol}_X$ is fixed, $\delta_U\Phat$ and $\delta_W\Phat$
are the $n$-form representatives of $\delta_U\mathsf P$ and
$\delta_W\mathsf P$, respectively. Starting from
\eqref{eq:reduced-symplectic-current} and using
\eqref{eq:momentum-n-form}, together with the compatibility of
$\nabla^{\mathcal A}$ with the duality pairing and the Lie bracket,
we obtain
\begin{align*}
\mathrm d\,\omega(U,W)
={}&
\Big(
\langle
\operatorname{div}^{\mathcal A}(\delta_W\mathsf P),\eta_U
\rangle
+
\langle
\delta_W\mathsf P,\nabla^{\mathcal A}\eta_U
\rangle
\\
&\quad-
\langle
\operatorname{div}^{\mathcal A}(\delta_U\mathsf P),\eta_W
\rangle
-
\langle
\delta_U\mathsf P,\nabla^{\mathcal A}\eta_W
\rangle
\\
&\quad-
\langle
\operatorname{div}^{\mathcal A}\mathsf P,[\eta_U,\eta_W]
\rangle
-
\langle
\mathsf P,\nabla^{\mathcal A}[\eta_U,\eta_W]
\rangle
\Big){\rm vol}_X.
\end{align*}
Using \eqref{eq:intrinsic-EP} and \eqref{eq:linearized-EP},
together with
\[
\nabla^{\mathcal A}\eta_U
=
\delta_U^{\mathcal A}\xi+[\xi,\eta_U],
\qquad
\nabla^{\mathcal A}\eta_W
=
\delta_W^{\mathcal A}\xi+[\xi,\eta_W],
\]
the coadjoint and bracket terms cancel by the convention
$\langle\operatorname{ad}^*_\zeta\alpha,\eta\rangle
=-\langle\alpha,[\zeta,\eta]\rangle
$
and the Jacobi identity. Hence
\begin{equation}\label{eq:MSFF-Hessian-step}
\mathrm d\,\omega(U,W)
=
\Big(
\langle
\delta_W\mathsf P,\delta_U^{\mathcal A}\xi
\rangle
-
\langle
\delta_U\mathsf P,\delta_W^{\mathcal A}\xi
\rangle
\Big){\rm vol}_X.
\end{equation}
Using \eqref{eq:momentum-linearization}, the right-hand side becomes
\[
\langle
\mathbb H_\ell(\xi)[\delta_W^{\mathcal A}\xi],
\delta_U^{\mathcal A}\xi
\rangle
-
\langle
\mathbb H_\ell(\xi)[\delta_U^{\mathcal A}\xi],
\delta_W^{\mathcal A}\xi
\rangle,
\]
which vanishes by symmetry of the fiber Hessian of $\ell$.
Therefore $\mathrm d\,\omega(U,W)=0$.
\end{proof}

\begin{corollary}\label{cor:centered-MS-current}
In a local trivialization of $P$, write
\[
\Phat=(\Phat_{\mathfrak g},\Phat_V),
\qquad
\eta_U=(\eta_U^{\mathfrak g},\eta_U^V),
\qquad
\eta_W=(\eta_W^{\mathfrak g},\eta_W^V).
\]
Then
\begin{align}
\omega(U,W)
={}&
\langle
\delta_W\Phat_{\mathfrak g},\eta_U^{\mathfrak g}
\rangle
-
\langle
\delta_U\Phat_{\mathfrak g},\eta_W^{\mathfrak g}
\rangle
+
\langle\delta_W\Phat_V,\eta_U^V\rangle
-
\langle\delta_U\Phat_V,\eta_W^V\rangle
\nonumber\\
&+
\langle
\operatorname{ad}^*_{\eta_U^{\mathfrak g}}\Phat_{\mathfrak g},
\eta_W^{\mathfrak g}
\rangle
-
\langle
\eta_U^{\mathfrak g}\heartsuit\Phat_V,\eta_W^V
\rangle
-
\langle
\eta_U^V\diamondsuit\Phat_V,\eta_W^{\mathfrak g}
\rangle.
\label{eq:centered-MS-current}
\end{align}
Here the coadjoint, heart, and diamond operators act on the dual
factors, leaving the differential-form factor unchanged.

Under the hypotheses of Proposition~\ref{prop:reduced-MSFF},
this current is closed. Consequently, for every compact domain
$\Sigma\subset X$ with smooth boundary,
\begin{equation}\label{eq:centered-MS-flux}
\int_{\partial\Sigma}\omega(U,W)=0.
\end{equation}
\end{corollary}
\begin{proof}
The $V$-component of the centered bracket is
\[
[\eta_U,\eta_W]_V
=
\eta_U^V\spadesuit\eta_W^{\mathfrak g}
-
\eta_W^V\spadesuit\eta_U^{\mathfrak g}.
\]
Therefore,
\begin{align*}
-\langle\Phat_V,[\eta_U,\eta_W]_V\rangle
={}&
-\langle
\Phat_V,\eta_U^V\spadesuit\eta_W^{\mathfrak g}
\rangle
+
\langle
\Phat_V,\eta_W^V\spadesuit\eta_U^{\mathfrak g}
\rangle
\\
={}&
-\langle
\eta_U^V\diamondsuit\Phat_V,\eta_W^{\mathfrak g}
\rangle
-
\langle
\eta_U^{\mathfrak g}\heartsuit\Phat_V,\eta_W^V
\rangle,
\end{align*}
where the last equality follows from \eqref{eq:heart-diamond-spades}. For the $\mathfrak g$-component, the coadjoint
convention gives
\[
-\langle
\Phat_{\mathfrak g},
[\eta_U^{\mathfrak g},\eta_W^{\mathfrak g}]
\rangle
=
\langle
\operatorname{ad}^*_{\eta_U^{\mathfrak g}}\Phat_{\mathfrak g},
\eta_W^{\mathfrak g}
\rangle.
\]
Substituting these identities into
\eqref{eq:reduced-symplectic-current} yields
\eqref{eq:centered-MS-current}. The conservation law follows from
Proposition~\ref{prop:reduced-MSFF}, and the flux identity follows
from Stokes' theorem.
\end{proof}
\begin{remark}
Formula \eqref{eq:centered-MS-current} separates the canonical
$\mathfrak g$- and $V$-contributions from the terms induced by the
centered Lie algebra structure. The coadjoint term is the usual
Lie-algebraic contribution from the $\mathfrak g$-sector, whereas the
heart and diamond terms encode the two-sided coupling between the
$\mathfrak g$- and $V$-components. Thus the multisymplectic current
does not decompose into two independently conserved sectorial
currents: the centered coupling is already present in the current
itself.
\end{remark}\medskip

\subsection{Noether Theorem and Energy-Momentum Balance Laws}\label{sec:noether}

We now consider the covariant balance laws associated with the reduced
Euler-Poincar\'e equations. The covariant Noether theorem and its
Poincar\'e-Cartan formulation are standard in first-order field theory; see,
for instance, \cite{GoIsMa1998,DeLeonMartinDiegoSantamaria2004}. Conservation
laws in covariant Euler-Poincar\'e reduction on principal bundles were studied
in \cite{CaGaRa2001}. Here we make explicit how these balance laws split for
the centered semidirect product $G\Join V$. For Noether conservation laws associated with particle-relabeling symmetries in Euler--Poincar\'e fluid models, see also
\cite{CoHo2013}.

Let $s:X\to P$ be a solution of the Euler--Lagrange equations and let
$\xi=[j^1s]-\mathcal A\in\Omega^1(X,\Ad(P))$ be the associated reduced field.
As above, write
$\mathsf P=\delta\ell/\delta\xi\in\Gamma(TX\otimes\Ad(P)^*)$
for the reduced covariant momentum. Then $\xi$ and $\mathsf P$ satisfy the
intrinsic Euler--Poincar\'e equation \eqref{eq:intrinsic-EP}.

For $\eta\in\Gamma(\Ad(P))$, define the reduced momentum current by
\begin{equation}\label{eq:reduced-momentum-current}
J_\eta
:=
\langle\Phat,\eta\rangle
=
\iota_{\langle\mathsf P,\eta\rangle}{\rm vol}_X
\in\Omega^n(X).
\end{equation}
Thus $J_\eta$ is obtained by pairing the momentum $n$-form
$\Phat$ with the reduced variation parameter $\eta$. 
Locally, if
$\Phat=\mathsf P^\mu\,{\rm d}^n x_\mu$, then
\[
J_\eta
=
\langle\mathsf P^\mu,\eta\rangle\,{\rm d}^n x_\mu.
\]

\begin{proposition}\label{prop:covariant-balance-laws}
Under the assumptions above, the following statements hold.

\emph{(i)} For every $\eta\in\Gamma(\Ad(P))$,
\begin{equation}\label{eq:reduced-momentum-balance}
\mathrm d J_\eta
=
\left\langle
\mathsf P,
\delta_\eta^{\mathcal A}\xi
\right\rangle
{\rm vol}_X,
\qquad
\delta_\eta^{\mathcal A}\xi
:=
\nabla^{\mathcal A}\eta-[\xi,\eta].
\end{equation}

\emph{(ii)} Let $\chi\in\mathfrak g\Join V$ and consider the infinitesimal
variation generated by the fixed right action of $G\Join V$,
$s_\varepsilon=s\cdot\exp(\varepsilon\chi)$. If $\eta_\chi$ denotes its
reduced variation parameter, then
\begin{equation}\label{eq:fixed-group-transport}
\delta_{\eta_\chi}^{\mathcal A}\xi=0,
\qquad
\mathrm dJ_{\eta_\chi}=0.
\end{equation}
In a local trivialization, write
$\eta_\chi=(\eta_\chi^{\mathfrak g},\eta_\chi^V)$, $\mathsf P=(\mathsf P_{\mathfrak g},\mathsf P_V)$ and $\Phat=(\Phat_{\mathfrak g},\Phat_V)$, and define
\begin{equation}\label{eq:centered-sector-currents}
J_\chi^{\mathfrak g}
:=
\langle
\Phat_{\mathfrak g},
\eta_\chi^{\mathfrak g}
\rangle,
\qquad
J_\chi^V
:=
\langle
\Phat_V,
\eta_\chi^V
\rangle.
\end{equation}
Then $J_{\eta_\chi}=J_\chi^{\mathfrak g}+J_\chi^V$ and
\begin{equation}\label{eq:centered-sector-balance}
\mathrm dJ_\chi^{\mathfrak g}
=
\mathcal E_\chi\,{\rm vol}_X,
\qquad
\mathrm dJ_\chi^V
=
-\mathcal E_\chi\,{\rm vol}_X,
\end{equation}
where
\begin{equation}\label{eq:centered-exchange-rate}
\mathcal E_\chi
:=
\left\langle
(\mathcal A^V-\xi^V)\diamondsuit\mathsf P_V,
\eta_\chi^{\mathfrak g}
\right\rangle
=
-
\left\langle
\eta_\chi^{\mathfrak g}\heartsuit\mathsf P_V,
\mathcal A^V-\xi^V
\right\rangle.
\end{equation}
Thus the two sectorial currents need not be separately conserved, while their
sum is closed.

\emph{(iii)} Let $\widetilde Z\in\mathfrak X(P)$ be a projectable vector field
over $Z\in\mathfrak X(X)$ whose first prolongation is an infinitesimal
symmetry of $\mathcal L$. Along $s$, let
$U_{\widetilde Z}:=\widetilde Z\circ s-Ts(Z)\in\Gamma(s^*VP)$
and let $\eta_{\widetilde Z}\in\Gamma(\Ad(P))$ be its reduced variation
parameter. The covariant Noether current is
\begin{equation}\label{eq:reduced-Noether-current}
\mathcal J_{\widetilde Z}
:=
(j^1s)^*\!\left(
\iota_{j^1\widetilde Z}\Theta_{\mathcal L}
\right)
=
J_{\eta_{\widetilde Z}}
+
\ell\,\iota_Z{\rm vol}_X.
\end{equation}
and it satisfies
\begin{equation}\label{eq:covariant-Noether-conservation}
\mathrm d\mathcal J_{\widetilde Z}=0.
\end{equation}

In particular, let $Z^{H_{\mathcal A}}$ be the horizontal lift of
$Z\in\mathfrak X(X)$ with respect to the reference connection $\mathcal A$.
Its vertical characteristic has reduced parameter
\begin{equation}\label{eq:horizontal-lift-parameter}
\eta_{Z^{H_{\mathcal A}}}
=
-\iota_Z\xi.
\end{equation}
Hence, whenever $Z^{H_{\mathcal A}}$ is an infinitesimal symmetry of
$\mathcal L$,
\begin{equation}\label{eq:energy-momentum-current}
\mathcal J_{Z^{H_{\mathcal A}}}
= -\iota_{\mathsf T_{\mathcal A}(Z)}{\rm vol}_X.\end{equation}
where $T_{\mathcal A}$ is the energy–momentum tensor associated with the $\mathcal A$-horizontal lift,
\begin{equation}\label{eq:canonical-energy-momentum}
\mathsf T_{\mathcal A}(Z)
=
\left\langle
\mathsf P,\iota_Z\xi
\right\rangle
-\ell Z,
\qquad Z\in\mathfrak X(X),
\end{equation}
and 
\begin{equation}\label{eq:energy-momentum-conservation}
\operatorname{div}\bigl(\mathsf T_{\mathcal A}(Z)\bigr)=0.
\end{equation}
\end{proposition}

\begin{proof}
For \emph{(i)}, the definition of the covariant divergence and the
compatibility of $\nabla^{\mathcal A}$ with the duality pairing give
\begin{align*}
\mathrm dJ_\eta
={}&
\left(
\left\langle
\operatorname{div}^{\mathcal A}\mathsf P,
\eta
\right\rangle
+
\left\langle
\mathsf P,
\nabla^{\mathcal A}\eta
\right\rangle
\right){\rm vol}_X.
\end{align*}
Using \eqref{eq:intrinsic-EP} and the coadjoint convention
$\langle\operatorname{ad}^*_\xi\mathsf P,\eta\rangle
=-\langle\mathsf P,[\xi,\eta]\rangle$, we obtain
\[
\mathrm dJ_\eta
=
\left\langle
\mathsf P,
\nabla^{\mathcal A}\eta-[\xi,\eta]
\right\rangle{\rm vol}_X,
\]
which proves \eqref{eq:reduced-momentum-balance}.

For \emph{(ii)}, the reduced connection is unchanged by the fixed right
$G\Join V$-action:
$[j^1s_\varepsilon]=[j^1s]$. Since $\mathcal A$ is fixed, the induced
variation of $\xi=[j^1s]-\mathcal A$ therefore vanishes. By
\eqref{eq:delta-xi-CGPR}, this is precisely the first identity in
\eqref{eq:fixed-group-transport}, and the second follows from
\eqref{eq:reduced-momentum-balance}.

It remains to compute the two centered components. From
\eqref{eq:delta-sigma} and
$\delta_{\eta_\chi}^{\mathcal A}\xi=0$,
\[
\nabla^{\mathcal A^{\mathfrak g}}\eta_\chi^{\mathfrak g}
=
[\xi^{\mathfrak g},\eta_\chi^{\mathfrak g}].
\]
Using \eqref{eq:centeredeqsg},
\begin{align*}
\mathrm dJ_\chi^{\mathfrak g}
={}&
\Big(
\left\langle
\operatorname{div}^{\mathcal A^{\mathfrak g}}
\mathsf P_{\mathfrak g},
\eta_\chi^{\mathfrak g}
\right\rangle
+
\left\langle
\mathsf P_{\mathfrak g},
\nabla^{\mathcal A^{\mathfrak g}}
\eta_\chi^{\mathfrak g}
\right\rangle
\Big){\rm vol}_X
=
\left\langle
(\mathcal A^V-\xi^V)\diamondsuit\mathsf P_V,
\eta_\chi^{\mathfrak g}
\right\rangle{\rm vol}_X,
\end{align*}
since the two terms containing
$[\xi^{\mathfrak g},\eta_\chi^{\mathfrak g}]$
cancel by the coadjoint convention.

Similarly, \eqref{eq:delta-nu} with
$\delta_{\eta_\chi}^{\mathcal A}\xi=0$ gives
\[
{\rm d}\eta_\chi^V
=
\eta_\chi^V\spadesuit
(\mathcal A^{\mathfrak g}-\xi^{\mathfrak g})
-
(\mathcal A^V-\xi^V)\spadesuit\eta_\chi^{\mathfrak g}.
\]
Together with \eqref{eq:centeredeqsV}, this yields
\begin{align*}
\mathrm dJ_\chi^V
=
\Big(
\left\langle
\operatorname{div}\mathsf P_V,
\eta_\chi^V
\right\rangle
+
\left\langle
\mathsf P_V,
{\rm d}\eta_\chi^V
\right\rangle
\Big){\rm vol}_X
=
-\left\langle
(\mathcal A^V-\xi^V)\diamondsuit\mathsf P_V,
\eta_\chi^{\mathfrak g}
\right\rangle{\rm vol}_X,
\end{align*}
where we used
$\langle\alpha,v\spadesuit\zeta\rangle
=\langle v\diamondsuit\alpha,\zeta\rangle
=-\langle\zeta\heartsuit\alpha,v\rangle$.
This proves \eqref{eq:centered-sector-balance}.

For \emph{(iii)}, the covariant Noether theorem gives
$\mathrm d\mathcal J_{\widetilde Z}=0$ on every Euler--Lagrange solution
whenever $j^1\widetilde Z$ is an infinitesimal symmetry of the Lagrangian
density \cite{GoIsMa1998,DeLeonMartinDiegoSantamaria2004}.
To obtain the reduced expression of the Noether current, from the
definition of $U_{\widetilde Z}$ we have
\[
\widetilde Z\circ s
=
U_{\widetilde Z}+Ts(Z).
\]
Accordingly, along the holonomic jet $j^1s$,
\[
j^1\widetilde Z
=
j^1U_{\widetilde Z}+T(j^1s)(Z).
\]
Moreover,
\[
(j^1s)^*\Theta_{\mathcal L}
=
\mathcal L(j^1s)
=
\ell(\xi)\,{\rm vol}_X.
\]
Hence, using naturality of the contraction and
\eqref{eq:reduced-potential-current-local},
\begin{align*}
\mathcal J_{\widetilde Z}
=
(j^1s)^*\!\left(
\iota_{j^1U_{\widetilde Z}}\Theta_{\mathcal L}
\right)
+
\iota_Z\!\left((j^1s)^*\Theta_{\mathcal L}\right)
=
J_{\eta_{\widetilde Z}}
+
\ell\,\iota_Z{\rm vol}_X,
\end{align*}
which proves \eqref{eq:reduced-Noether-current}.

Finally, by the definition of the connection represented by
$\sigma=[j^1s]$, the vector $Ts(Z)$ is the $\sigma$-horizontal lift of
$Z$ along $s$, whereas $Z^{H_{\mathcal A}}\circ s$ is the
$\mathcal A$-horizontal lift. Under the affine identification
\eqref{eq:identificationCP}, the vertical difference
$Ts(Z)-Z^{H_{\mathcal A}}\circ s$ corresponds to
$\iota_Z(\sigma-\mathcal A)=\iota_Z\xi$. Therefore the vertical
characteristic
$Z^{H_{\mathcal A}}\circ s-Ts(Z)$ corresponds to
$-\iota_Z\xi$, proving \eqref{eq:horizontal-lift-parameter}.
Substitution into \eqref{eq:reduced-Noether-current} gives
\eqref{eq:energy-momentum-current} and
\eqref{eq:canonical-energy-momentum}. If
$Z=\partial_\nu$ and $Z^{H_{\mathcal A}}$ is an infinitesimal symmetry,
then \eqref{eq:covariant-Noether-conservation} gives
\eqref{eq:energy-momentum-conservation}.
\end{proof}
\begin{remark}
In a volume coordinate chart, the energy–momentum tensor has the local expression 
\begin{equation}\label{eq:canonical-energy-momentum-local}
\mathsf T_{\mathcal A}{}^\mu{}_{\nu}
:=
\left\langle
\mathsf P^\mu,
\xi_\nu
\right\rangle
-\delta^\mu{}_{\nu}\ell.
\end{equation}
If the horizontal lifts of the translations
$\partial_\nu$ are infinitesimal symmetries of $\mathcal L$, then, from \eqref{eq:energy-momentum-conservation},
\begin{equation}\label{eq:energy-momentum-conservation-local}
\partial_\mu
\mathsf T_{\mathcal A}{}^\mu{}_{\nu}
=0.
\end{equation}
\end{remark}

\begin{remark}\label{rem:noether-generator-splitting}
The splitting of the sectorial currents in
\eqref{eq:centered-sector-currents} should not be confused with splitting
the fixed infinitesimal generator itself. Indeed, for
$\chi=(\chi^{\mathfrak g},\chi^V)\in\mathfrak g\Join V$, set
\[
\chi_{\mathfrak g}:=(\chi^{\mathfrak g},0),
\qquad
\chi_V:=(0,\chi^V).
\]
Since the Lagrangian is invariant under the full $G\Join V$-action,
each of these generators separately defines an infinitesimal symmetry.
Hence, by Proposition~\ref{prop:covariant-balance-laws} \emph{(ii)},
\[
\mathrm dJ_{\eta_{\chi_{\mathfrak g}}}=0,
\qquad
\mathrm dJ_{\eta_{\chi_V}}=0.
\]
By linearity,
\[
\eta_\chi
=
\eta_{\chi_{\mathfrak g}}+\eta_{\chi_V},
\qquad
J_{\eta_\chi}
=
J_{\eta_{\chi_{\mathfrak g}}}
+
J_{\eta_{\chi_V}}.
\]
These separately conserved currents should not, however, be identified
with $J_\chi^{\mathfrak g}$ and $J_\chi^V$ in
\eqref{eq:centered-sector-currents}. In general,
$\eta_{\chi_{\mathfrak g}}$ and $\eta_{\chi_V}$ need not coincide with
$(\eta_\chi^{\mathfrak g},0)$ and $(0,\eta_\chi^V)$, respectively,
because the centered semidirect-product adjoint action mixes the two
components. Thus splitting the symmetry generator before forming the
reduced current is different from projecting the reduced current into
its $\mathfrak g$- and $V$-sectorial contributions. The latter satisfy
the exchange law \eqref{eq:centered-sector-balance} and therefore need
not be separately conserved.
\end{remark}\medskip

The currents in Proposition~\ref{prop:covariant-balance-laws} and the
symplectic current \eqref{eq:symplectic-current-cov} arise from the same
Poincar\'e--Cartan geometry but play different roles. A Noether current is
obtained by contracting $\Theta_{\mathcal L}$ with one infinitesimal symmetry,
whereas the symplectic current is obtained by contracting
$\Omega_{\mathcal L}$ with two first variations. In particular, a spacetime symmetry is not itself a vertical first variation, although its vertical characteristic defines one along a solution. Nevertheless, the associated Noether current should not in general be identified with the symplectic current obtained by taking this first variation as one of the arguments in the multisymplectic form formula.


\subsection{Lie-Poisson Field Equations}
\label{sec:LiePoisson}

We now describe the Hamiltonian counterpart of the covariant
Euler--Poincar\'e equations on the centered semidirect product $G\Join V$. We shall first review the geometric formalism of covariant Hamiltonian field theory, as explained in \cite{GoIsMa1998}. We then particularize the full-structure-group Lie--Poisson reduction
of~\cite{BerbelCastrillon2023,CaMa2003} to centered semidirect-product
groups. A recent extension of the covariant Hamiltonian framework to
Hamiltonian field theories invariant under a subgroup of the structure
group is developed in~\cite{BerbelCastrillonSubgroup2026}. Finally, we
show that, for hyperregular theories, the reduced Lagrangian and
Hamiltonian formulations are equivalent.

\medskip

We briefly recall the Hamiltonian setting for first-order field theories. Thus, throughout the first part of this subsection, \(P\to X\) denotes an arbitrary fiber bundle over an oriented manifold \(X\) of dimension \(n+1\). The affine dual jet bundle is \[ J^1P^*=\operatorname{Aff}\bigl(J^1P,\Lambda^{n+1}T^*X\bigr), \] that is, for \(y\in P_x\), $x\in X$, the fiber \(J^1_yP^*\) consists of affine maps $J^1_yP\longrightarrow \Lambda^{n+1}_xT^*X$. The use of affine morphisms gives \(J^1P^*\) one additional affine coordinate with respect to \(J^1P\). Indeed, $\mathrm{rank}( J^1 P^*) = \mathrm{rank}( J^1 P ) + 1$. This bundle is a multisymplectic space as it carries a canonical \((n+1)\)-form \(\Theta\), and therefore a canonical multisymplectic \((n+2)\)-form $ \Omega=-{\rm d}\Theta$.

The corresponding polysymplectic momentum bundle is 
\[ \Pi = TX\otimes_P V^*P\otimes_P \Lambda^{n+1}T^*X . \] It has the same fiber rank as \(J^1P\to P\) and there is a natural affine projection \[ J^1P^*\rightarrow \Pi , \qquad \varphi\mapsto \hat{\varphi}, \] where \(\hat{\varphi}\) denotes the linear part of the affine map \(\varphi\in J^1_yP^*\). This projection is an affine bundle modeled on the rank-one bundle \(\Lambda^{n+1}T^*X\to X\). While \(J^1P^*\) carries the canonical multisymplectic form, \(\Pi\) is the natural bundle of polymomenta.

A Hamiltonian system may be described by a  section \(\delta:\Pi\to J^1P^*\) of the affine projection above. Pulling back the canonical forms on \(J^1P^*\) gives \[ \Theta_\delta=\delta^*\Theta, \qquad \Omega_\delta=-{\rm d}\Theta_\delta . \] 
Then, a section \(p:X\to\Pi\) is said to be a solution of the Hamiltonian system if 
\begin{equation} \label{eq:hamilton}
 p^*\bigl(i_Y\Omega_\delta\bigr)=0   
\end{equation}
for every vertical vector field \(Y\) with respect to \(\Pi\to X\), in analogy with equation \eqref{eq:EL-multisymplectic}.

In order to express the Hamiltonian section \(\delta\) as a Hamiltonian density on $\Pi$, one must choose a reference section of the affine bundle \(J^1P^*\to\Pi\). A connection $\mathcal A$ on the fiber bundle $P\to X$, provides such a reference section $\delta_{\mathcal{A}}:\Pi\to J^1P^*$, as detailed in \cite{GoIsMa1998}. 
We use the Hamiltonian density
$$\mathcal H:=\delta_{\mathcal A}-\delta,$$ so that
$\delta=\delta_{\mathcal A}-\mathcal H$.
For a hyperregular Lagrangian density $\mathcal L$, this
convention agrees with the Legendre Hamiltonian density
$\mathcal H(p)=\langle p,v-\mathcal A\rangle-\mathcal L(v)$,
where $F\mathcal L:J^1P\to\Pi$ is the fibre derivative
of $\mathcal L$ and $v=(F\mathcal L)^{-1}(p)$.
Indeed, the affine Legendre transform is the map
$j\mapsto\mathcal L(v)+\langle p,j-v\rangle$,
whereas the reference section is
$j\mapsto\langle p,j-\mathcal A\rangle$.
Thus the affine Legendre section minus the reference section equals
$-\mathcal H(p)$, consistently with
$\delta=\delta_{\mathcal A}-\mathcal H$.
In local volume coordinates, and for a zero reference connection,
the pulled-back canonical form is
$\Theta_\delta=p_a^\mu\,dq^a\wedge d^n x_\mu-H\,\mathrm{vol}_X$.
With this sign, the Hamiltonian system may equivalently be specified
by the triple $(\Pi,\mathcal A,\mathcal H)$.

Alternatively, there is an equivalent bracket formulation for these Hamiltonian systems. We choose the covariant bracket convention in which the bracket
with the Legendre Hamiltonian appears on the left-hand side of the
following evolution identity.
Fix an oriented volume form $\mathrm{vol}_X$ on $X$. A section $p$ of $\Pi\to X$ is a solution of $(\Pi,\mathcal{A},\mathcal{H})$, with $\mathcal{H}=H\mathrm{vol}_X$, if and only if for any horizontal Poisson $n$-form $F$ the identity
\begin{equation} \label{eq:hamilton-bracket-equations}
\{F,H\}\mathrm{vol}_X\circ p={\rm d}(p^*F)-({\rm d}^hF)\circ p
\end{equation}
holds, where ${\rm d}^hF$ is the horizontal differential of $F$ with respect to the extension on $\Pi$ obtained from the connection $\mathcal{A}$ and any Riemannian connection on $X$ as explained in \cite[Proposition~5.1]{CaMa2003}.

\medskip

We now return to the principal $(G\Join V)$-bundle setting.
The reduced polymomentum bundle is
\[
\bar\Pi:=\Pi/(G\Join V)
=TX\otimes\Ad(P)^*\otimes\Lambda^{n+1}T^*X
\simeq\Ad(P)^*\otimes\Lambda^nT^*X,
\]
where the last identification is the canonical contraction
$v\otimes\alpha\otimes\nu\mapsto\alpha\otimes\iota_v\nu$.
We denote a reduced multimomentum section by
$\widehat{\mathsf P}\in\Gamma(\bar\Pi)$ and use the same symbol
for its representation as an $\Ad(P)^*$-valued $n$-form.
Using the fixed volume form, we write
\[
\widehat{\mathsf P}=\mathsf P\otimes{\rm vol}_X,
\qquad
\mathsf P\in\Gamma(TX\otimes\Ad(P)^*).
\]
In the differential-form representation, this reads
$\widehat{\mathsf P}=\iota_{\mathsf P}{\rm vol}_X$.
We use the same symbol $h$ for a function on $\bar\Pi$ and its
representative on $TX\otimes\Ad(P)^*$ under this volume-dependent
identification. All derivatives $\delta h/\delta\mathsf P$ below
refer to the latter representative.

The reduced observables are the projections of
$(G\Join V)$-invariant horizontal Poisson $n$-forms on $\Pi$.
For $\zeta\in\Gamma(\Ad(P))$, let $\theta_\zeta$ be the horizontal
$n$-form on $\bar\Pi$ characterized by
\[
\widehat{\mathsf P}^{\,*}\theta_\zeta
=\langle\widehat{\mathsf P},\zeta\rangle
\]
for every section $\widehat{\mathsf P}$.
The reduced observables have the form~\cite{CaMa2003}
\begin{equation}\label{eq:reduced-Poisson-forms}
f=\theta_\zeta+\pi_{X,\bar\Pi}^*\omega+Z,
\end{equation}
where $\zeta\in\Gamma(\Ad(P))$, $\omega\in\Omega^n(X)$,
and $Z$ is a closed horizontal $n$-form on $\bar\Pi$.

For a function $h$ on $\bar\Pi$ and a reduced Poisson $n$-form
$f$ as in \eqref{eq:reduced-Poisson-forms}, the Lie--Poisson bracket,
expressed at $\widehat{\mathsf P}=\mathsf P\otimes{\rm vol}_X$, is
\begin{equation}\label{eq:LP-bracket}
\{f,h\}_{\mathrm{LP}}(\widehat{\mathsf P})
=\left\langle
\mathsf P,
\left[\zeta,\frac{\delta h}{\delta\mathsf P}(\mathsf P)\right]
\right\rangle,
\end{equation}
where
\[
\frac{\delta h}{\delta\mathsf P}:
TX\otimes\Ad(P)^*
\longrightarrow T^*X\otimes\Ad(P)
\]
is the fibre derivative of the representative $h$, defined by
\[
\frac{\delta h}{\delta \mathsf P}(\mathsf P)(\mathsf Q)
=\left.\frac{d}{d\varepsilon}\right|_{\varepsilon=0}
h(\mathsf P+\varepsilon\mathsf Q),
\qquad
\mathsf P,\mathsf Q\in T_xX\otimes\Ad(P)_x^*.
\]

Let $\mathcal A$ be a principal connection on $P\to X$, and let
$\mathcal H=H\,{\rm vol}_X$ be a $(G\Join V)$-invariant
Hamiltonian density on $\Pi$. Its reduced density is written as
$
\mathsf h(\widehat{\mathsf P})=h(\mathsf P)\,{\rm vol}_X.
$
For any solution $p:X\to\Pi$ of $(\Pi,\mathcal A,\mathcal H)$,
let $\widehat{\mathsf P}=\kappa\circ p$, where
$\kappa:\Pi\to\bar\Pi$ is the quotient projection.
For every reduced horizontal Poisson $n$-form $f$, one has
\begin{equation}\label{eq:cov-LP-functionals}
\widehat{\mathsf P}^{\,*}\{f,h\}_{\mathrm{LP}}\,{\rm vol}_X
={\rm d}\bigl(\widehat{\mathsf P}^{\,*}f\bigr)
-\widehat{\mathsf P}^{\,*}({\rm d}^hf),
\end{equation}
where ${\rm d}^h$ is the horizontal differential on $\bar\Pi$
defined in \cite{CaMa2003}.
This condition is equivalent to the intrinsic covariant Lie--Poisson equations
\begin{equation}\label{eq:cov-LP-intrinsic}
\operatorname{div}^{\mathcal A}\mathsf P
=\operatorname{ad}^*_{\delta h/\delta\mathsf P}\mathsf P,
\end{equation}
where the divergence is taken with respect to $\mathcal A$.

The sign can be checked on the linear observable $f=\theta_\zeta$.
By our convention for the coadjoint action,
$\{\theta_\zeta,h\}_{\mathrm{LP}}
=\langle\operatorname{ad}^*_{\delta h/\delta\mathsf P}
\mathsf P,\zeta\rangle$.
The right-hand side of the reduced evolution identity is
$\langle\operatorname{div}^{\mathcal A}\mathsf P,\zeta\rangle
\mathrm{vol}_X$.
Since $\zeta$ is arbitrary, this gives the displayed
Lie-Poisson equation.

As discussed in Subsection \ref{subsec:reduction-principal-bundles-centered}, after choosing a  trivialization $P\simeq X\times(G\Join V)$, the adjoint bundle is identified with
$\Ad(P)\simeq X\times(\mathfrak g\Join V)$. Consequently, $\Ad(P)^*\simeq X\times(\mathfrak g^*\oplus V^*)$. Thus, the reduced polymomentum bundle is
\[
    \bar\Pi \simeq TX\otimes(\mathfrak g^*\oplus V^*)\otimes \Lambda^{n+1}T^*X .
\]

In this trivialization, the representative $\mathsf P$ is written as
\[
\mathsf P=(\mathsf P_{\mathfrak g},\mathsf P_V),
\qquad
\mathsf P_{\mathfrak g}
\in\Gamma(TX\otimes\mathfrak g^*),
\qquad
\mathsf P_V
\in\Gamma(TX\otimes V^*).
\]

Similarly, in this trivialization, sections $\zeta = (\zeta^{\mathfrak{g}},\zeta^V) \in \Gamma(\Ad(P))$ appearing in the description of reduced horizontal Poisson $n$-forms split as $\zeta^{\mathfrak{g}}\in C^\infty(X,\mathfrak g)$ and $\zeta^V\in C^\infty(X, V)$.
Accordingly, the bracket in equation \eqref{eq:LP-bracket} decomposes into two components, which in terms of the heart and diamond operators read as
\begin{align}
\{f,h\}_{\mathrm{LP}}(\widehat{\mathsf P})
&=\left\langle\mathsf P,
\left[\zeta,\frac{\delta h}{\delta\mathsf P}\right]
\right\rangle\\
&=\left\langle\mathsf P_{\mathfrak g},
\left[\zeta^{\mathfrak g},\frac{\delta h}{\delta\mathsf P_{\mathfrak g}}\right]
\right\rangle
+\left\langle\mathsf P_V,
\zeta^V\spadesuit\frac{\delta h}{\delta\mathsf P_{\mathfrak g}}
-\frac{\delta h}{\delta\mathsf P_V}\spadesuit\zeta^{\mathfrak g}
\right\rangle\\
&=-\left\langle\operatorname{ad}^*_{\zeta^{\mathfrak g}}\mathsf P_{\mathfrak g},
\frac{\delta h}{\delta\mathsf P_{\mathfrak g}}\right\rangle
+\left\langle\zeta^{\mathfrak g}\heartsuit\mathsf P_V,
\frac{\delta h}{\delta\mathsf P_V}\right\rangle
+\left\langle\zeta^V\diamondsuit\mathsf P_V,
\frac{\delta h}{\delta\mathsf P_{\mathfrak g}}\right\rangle.
\label{eq:LP-heart-diamond}
\end{align}

Equivalently, with our coadjoint convention, the centered coadjoint action is
\[
\operatorname{ad}^{*}_{(\chi^{\mathfrak g},\chi^V)}(\mathsf P_{\mathfrak g},\mathsf P_V)
=
\left(
\operatorname{ad}^{*}_{\chi^{\mathfrak g}}\mathsf P_{\mathfrak g}-\chi^V\diamondsuit \mathsf P_V,
-\chi^{\mathfrak g}\heartsuit \mathsf P_V
\right).
\] Thus, the intrinsic equation \eqref{eq:cov-LP-intrinsic} has the expression 
\begin{align}
& \operatorname{div}^{\mathcal A^{\mathfrak g}}({\mathsf P}_{\mathfrak g})
-\mathcal A^V\diamondsuit{\mathsf P}_V
=
\operatorname{ad}^{*}_{\frac{\delta h}{\delta \mathsf P_{\mathfrak g}}}
({\mathsf P}_{\mathfrak g})
-
\frac{\delta h}{\delta \mathsf P_V}\diamondsuit{\mathsf P}_V,
\label{eq:cov-LP-g}\\[2pt]
& \operatorname{div}({\mathsf P}_V)
-\mathcal A^{\mathfrak g}\heartsuit{\mathsf P}_V
=
-\frac{\delta h}{\delta \mathsf P_{\mathfrak g}}\heartsuit{\mathsf P}_V.
\label{eq:cov-LP-v}
\end{align}

As will follow from the hyperregular Legendre transformation below,
\[
\frac{\delta h}{\delta\mathsf P_{\mathfrak g}}
=
\xi^{\mathfrak g},
\qquad
\frac{\delta h}{\delta\mathsf P_V}
=
\xi^V,
\]
so that equations
\eqref{eq:cov-LP-g}-\eqref{eq:cov-LP-v}
become precisely the Euler--Poincar\'e equations
\eqref{eq:centeredeqsg}-\eqref{eq:centeredeqsV},
after rearranging the terms involving the reference connection.

\paragraph{Hyperregular Lagrangians and Legendre transformation.}

Let $\mathcal L=L\,{\rm vol}_X:J^1P\rightarrow\Lambda^{n+1}T^*X$
be a hyperregular $(G\Join V)$-invariant first-order Lagrangian density,
and let $\ell:\mathcal C(P)\rightarrow\mathbb R$
be the corresponding reduced Lagrangian function defined in
\eqref{eq:reducedlagrangian}. Fix the reference connection
$\mathcal A\in\Gamma(\mathcal C(P))$.
Using the affine identification \eqref{eq:identificationCP}, we regard
the reduced Lagrangian as a function $\ell:T^*X\otimes\Ad(P)\rightarrow\mathbb R$, with reduced variable
\[
\xi=\sigma-\mathcal A\in T^*X\otimes\Ad(P).
\]
Its reduced Legendre transformation is
\begin{equation}
\label{eq:reduced-Legendre}
\mathbb F\ell:
T^*X\otimes\Ad(P)
\longrightarrow
TX\otimes\Ad(P)^*,
\qquad
\xi\mapsto
\mathsf P=\frac{\delta\ell}{\delta\xi}.
\end{equation}
The corresponding density-valued Legendre transformation is
\[
\widehat{\mathbb F\ell}(\xi)
=\frac{\delta\ell}{\delta\xi}\otimes{\rm vol}_X
=\widehat{\mathsf P}\in\bar\Pi.
\]
Hyperregularity means that the reduced Legendre transformation
\eqref{eq:reduced-Legendre} is a fiberwise diffeomorphism.

For a section $s$, set $\xi=[j^1s]-\mathcal A$ and
$p=\mathbb F\mathcal L\circ j^1s$. The reduced momenta satisfy
\[
\widehat{\mathsf P}=\kappa\circ p
=\widehat{\mathbb F\ell}\circ\xi,
\qquad
\mathsf P=\mathbb F\ell\circ\xi.
\]
Similarly, using the fixed volume form ${\rm vol}_X$, the reduced
Hamiltonian function is defined by
\begin{equation}
\label{eq:reduced-Hamiltonian-Legendre}
h(\mathsf P)
=
\langle\mathsf P,\xi\rangle-\ell(\xi),
\qquad
\xi=(\mathbb F\ell)^{-1}(\mathsf P).
\end{equation}
The usual Legendre duality relation gives
\begin{equation}
\label{eq:Legendre-duality}
\frac{\delta h}{\delta\mathsf P}
=
\xi.
\end{equation}
Consequently, the inverse reduced Legendre transformation takes
values in the vector bundle modelling the affine bundle
$\mathcal C(P)$:
\[
\mathbb Fh:
TX\otimes\Ad(P)^*
\rightarrow
T^*X\otimes\Ad(P),
\qquad
\mathsf P\mapsto
\frac{\delta h}{\delta\mathsf P}.
\]

\begin{proposition}
\label{prop:cov-LP}
Let
$\mathcal L=L\,{\rm vol}_X:J^1P\to\Lambda^{n+1}T^*X$
be a $(G\Join V)$-invariant first-order Lagrangian density with
hyperregular reduced Lagrangian $\ell$, and let $h$ be the associated
reduced Hamiltonian defined by the covariant Legendre transformation
\eqref{eq:reduced-Hamiltonian-Legendre}.
Then the covariant Euler--Poincar\'e equations
\eqref{eq:centeredeqsg}-\eqref{eq:centeredeqsV}
are equivalent to the covariant Lie--Poisson equations
\eqref{eq:cov-LP-g}-\eqref{eq:cov-LP-v}.
In particular, covariant reduction on $G\Join V$ admits mutually
consistent Lagrangian (Euler--Poincar\'e) and Hamiltonian
(Lie--Poisson) formulations.
\end{proposition}

\begin{proof}
By hyperregularity and the Legendre duality relation
\eqref{eq:Legendre-duality},
\[
\frac{\delta h}{\delta\mathsf P_{\mathfrak g}}
=
\xi^{\mathfrak g},
\qquad
\frac{\delta h}{\delta\mathsf P_V}
=
\xi^V.
\]
Substituting these identities into
\eqref{eq:cov-LP-g}-\eqref{eq:cov-LP-v} gives
\[
\operatorname{div}^{\mathcal A^{\mathfrak g}}
(\mathsf P_{\mathfrak g})
-
\mathcal A^V\diamondsuit\mathsf P_V
=
\operatorname{ad}^{*}_{\xi^{\mathfrak g}}
(\mathsf P_{\mathfrak g})
-
\xi^V\diamondsuit\mathsf P_V,
\]
and
\[
\operatorname{div}(\mathsf P_V)
-
\mathcal A^{\mathfrak g}\heartsuit\mathsf P_V
=
-\xi^{\mathfrak g}\heartsuit\mathsf P_V.
\]
Rearranging the terms involving the reference connection
$\mathcal A$ yields exactly the Euler--Poincar\'e equations
\eqref{eq:centeredeqsg}-\eqref{eq:centeredeqsV}.
Conversely, the same identities together with the invertibility of
the reduced Legendre transformation \eqref{eq:reduced-Legendre}
transform the Euler--Poincar\'e equations into
\eqref{eq:cov-LP-g}-\eqref{eq:cov-LP-v}.

This is the specialization to the centered semidirect product
$G\Join V$ of the hyperregular Euler--Poincar\'e/Lie--Poisson
equivalence established in \cite[Theorem~6.7]{CaMa2003}.
\end{proof}

This equivalence concerns the reduced differential equations.
A reduced Hamiltonian solution reconstructs locally to an unreduced
solution only when the connection
$\mathcal A+\delta h/\delta\mathsf P$ is flat, where the addition
is the affine addition in the bundle of connections.
Thus the Hamiltonian reconstruction condition is
$\operatorname{Curv}(\mathcal A+\delta h/\delta\mathsf P)=0$.
As on the Lagrangian side \eqref{eq:reconstructioncompatibility}, global reconstruction also requires the
appropriate topological and holonomy conditions.

\begin{remark}
The reduction considered in Sections~\ref{sec:eulerpoincare},~\ref{sec:multisymplecticformula} and~\ref{sec:LiePoisson} is performed
with respect to the natural action of the fixed structure group
$G\Join V$ on the principal bundle $P\to X$. In Section~\ref{sec:noether}, we additionally consider
projectable infinitesimal symmetries in order to formulate the
covariant Noether and energy--momentum currents.

Neither type of symmetry should be confused with gauge invariance
under $x$-dependent fiberwise transformations. The latter requires
the framework of Lie group bundles and generalized principal
connections developed in~\cite{CastrillonRodriguezGeneralized2023,CastrillonRodriguezGauge2024}.
Extending centered semidirect-product reduction to that setting lies
beyond the scope of the present work.
\end{remark}\medskip

\section{Applications of the Centered Reduction Framework}\label{sec:applications}

This section illustrates the general framework of Section~\ref{sec:covariantreduction} through two complementary examples. We first consider the matrix centered product $GL(r)\Join\operatorname{Mat}(r)$, where the heart and diamond couplings can be written explicitly. We then study $\mathrm{SO}(3)\Join\operatorname{Mat}(3)$ and identify its scalar, vector, and symmetric-traceless sectors, clarifying their relation with microstructural and liquid-crystal variables.

\subsection{Matrix-valued Theory on $\mathrm{GL}(r)\Join\mathrm{Mat}(r)$}

Let us illustrate the covariant Euler--Poincar\'e theory on centered
semidirect products by considering a matrix-valued
field theory, which directly extends
the finite-dimensional model in~\cite{CoJa2015} to the covariant (field-theoretic)
setting. 

Let $G=\mathrm{GL}(r)$, which linearly acts on the vector space of all $r\times r$ matrices $V=\mathrm{Mat}(r)$ by left and right
multiplication. These actions clearly commute, yielding the centered semidirect product~\eqref{eq:centered-product}
\[
G \Join V = \mathrm{GL}(r)\Join \mathrm{Mat}(r).
\]

Recall that $\mathfrak{g}=\mathfrak{gl}(r)$ consists of the vector space $\mathrm{Mat}(r)$ endowed with the matrix commutator. Hence, we may identify $\mathfrak g^*$ and $V^*$ with $\mathrm{Mat}(r)$ with the dual pairing given by the trace, i.e., \(\langle v,w\rangle=\operatorname{tr}(v^\top w)\) for each \(v,w\in\mathrm{Mat}(r)\). Under this identification, the heart and diamond operators were computed in~\cite[\S4.1]{CoJa2015}; namely, $A\heartsuit v = A^\top v-vA^\top $ and $v\diamondsuit w = v^\top w - w v^\top$, for each $A\in\mathfrak{gl}(r)$ and $v,w\in\mathrm{Mat}(r)$. Similarly, the coadjoint action of $\mathfrak{gl}(r)$ reads $\operatorname{ad}^*_A(\alpha) = \alpha A^\top - A^\top \alpha$, for each $A,\alpha\in\mathfrak{gl}(r)$.

Let $(X,\gamma)$ be an oriented, $(n+1)$-dimensional semi-Riemannian manifold with volume form $\mathrm{vol}_X$. Typically, $X$ is a spacetime, i.e., $\gamma$ is Lorentzian metric. We work on the trivial right principal bundle
$P=X\times(\operatorname{GL}(r)\Join 
\operatorname{Mat}(r))\to X$, with its fixed product trivialization. In addition, let
\begin{align*}
\mathbb M:\mathfrak{gl}(r)\to\mathfrak{gl}(r)^*=\mathrm{Mat}(r),\qquad\mathbb N:\mathrm{Mat}(r)\to\mathrm{Mat}(r)^*=\mathrm{Mat}(r),
\end{align*}
be symmetric and positive definite constitutive isomorphisms.

For a section $s=(g,u)$, we write its right logarithmic derivative as
$(ds)s^{-1}=(A,v)$, where $A=A_\mu dx^\mu$ and $v=v_\mu dx^\mu$.
We use the metric volume form $\mathrm{vol}_X$ and local coordinates
$(x^\mu)$, with $\gamma=\gamma_{\mu\nu}dx^\mu\otimes dx^\nu$. With the previous notations, the reduced Lagrangian $\ell:T^*X\otimes(\mathfrak{gl}(r)\Join\mathrm{Mat}(r))\to\mathbb R$ is defined as
\begin{align*}
\ell(A,v)=\frac{1}{2}\gamma^{\mu\nu}\left(\langle\mathbb M\,A_\mu,A_\nu\rangle+\langle\mathbb N\,v_\mu,v_\nu\rangle\right),
\end{align*}
where we write $A=A_\mu\,dx^\mu\in T^*X\otimes\mathfrak{gl}(r)$ and $v=v_\mu\, dx^\mu\in T^*X\otimes\mathrm{Mat}(r)$. This defines a right $(\mathrm{GL}(r)\Join\mathrm{Mat}(r))$-invariant Lagrangian density as $\mathcal L(j^1s)=\ell((\mathrm{d}s)s^{-1})\,\mathrm{vol}_X$, for each $s\in\Gamma(P\to X)$. 

The reference connection only determines the affine coordinates
used to describe this fixed model. Indeed, for an arbitrary reference
connection $\mathcal A$, the reduced variable is
$\xi=(ds)s^{-1}-\mathcal A$, and the same Lagrangian is represented by
$\ell_{\mathcal A}(\xi)=\ell(\mathcal A+\xi)$ in the chosen trivialization.
If $\mathcal A'=\mathcal A+\beta$, then $\xi'=\xi-\beta$ and
$\ell_{\mathcal A'}(\xi')=\ell_{\mathcal A}(\xi)$.
Thus changing the reference does not change the underlying model,
provided the reduced Lagrangian is transformed accordingly.
For the explicit calculations below, we choose the zero reference
connection, so that $\xi=(A,v)$.

The reduced momenta $(\alpha,\varpi)\in
\Gamma\!\left(
TX\otimes
\bigl(\mathfrak{gl}(r)\Join\operatorname{Mat}(r)\bigr)^*
\right)
\simeq
\Gamma\!\left(
TX\otimes
\bigl(\mathfrak{gl}(r)^*\oplus\operatorname{Mat}(r)^*\bigr)
\right)$ corresponding to a reduced field $(A,v)\in\Omega^1(X,\mathfrak{gl}(r)\Join\mathrm{Mat}(r))$ read $\alpha=\alpha^\mu\,\partial_\mu$ and $\varpi=\varpi^\mu\,\partial_\mu$, where
\begin{align}\label{eq:matrixcomponents}
& \alpha^\mu=\frac{\delta\ell}{\delta A_\mu}=\gamma^{\mu\nu}\,\mathbb M\, A_\nu,\qquad\varpi^\mu=\frac{\delta\ell}{\delta v_\mu}=\gamma^{\mu\nu}\,\mathbb N\,v_{\nu}.
\end{align}
With the zero reference connection, the general reduced
Euler-Poincar\'e equations specialize to
\[
\begin{aligned}
\operatorname{div}(\alpha)
&=\alpha^\mu A_\mu^\top-A_\mu^\top\alpha^\mu
-v_\mu^\top\varpi^\mu+\varpi^\mu v_\mu^\top,\\
\operatorname{div}(\varpi)
&=-A_\mu^\top\varpi^\mu+\varpi^\mu A_\mu^\top.
\end{aligned}
\] Here the divergence is taken with respect to the metric volume form:
for a matrix-valued vector field $Z=Z^\mu\partial_\mu$,
$\operatorname{div}Z=\upsilon^{-1}\partial_\mu(\upsilon Z^\mu)$,
where $\upsilon=\sqrt{|\det(\gamma_{\alpha\beta})|}$.

By substituting~\eqref{eq:matrixcomponents}, we obtain a coupled system of first-order
PDEs for the reduced field $(A,v)$. As explained in Section~\ref{sec:eulerpoincare}, reconstruction of a solution $(A,v)$
as $(ds)s^{-1}$ additionally requires flatness. In the present
trivialization, this amounts to
$\partial_\mu A_\nu-\partial_\nu A_\mu-[A_\mu,A_\nu]=0$
and
$\partial_\mu v_\nu-\partial_\nu v_\mu-[A_\mu,v_\nu]
+[A_\nu,v_\mu]=0$.
On a simply connected domain, these conditions allow reconstruction
of $s$ after prescribing its value at one point.

For this system we can write explicitly the reduced multisymplectic current of 
Section~\ref{sec:multisymplecticformula}.
Let $s$ be a solution with reduced field $(A,v)$, and let $U,W$
be vertical first variations along $s$, with reduced parameters $\eta_U=(B,b),$ and
$\eta_W=(C,c)$. With the trivial reference connection, the constrained variations given by equation \eqref{eq:linearized-constraint} are
\begin{align*}
\delta_U A_\mu
&=\partial_\mu B-[A_\mu,B],\\
\delta_U v_\mu
&=\partial_\mu b-[A_\mu,b]+[B,v_\mu],
\end{align*}
and similarly for $W$, replacing $(B,b)$ by $(C,c)$.
Keeping $\gamma$, $\mathbb M$, and $\mathbb N$ fixed under
variation, the constitutive relations \eqref{eq:matrixcomponents}
give
\[
\delta_U\alpha^\mu
=\gamma^{\mu\nu}\mathbb M(\delta_U A_\nu),
\qquad
\delta_U\varpi^\mu
=\gamma^{\mu\nu}\mathbb N(\delta_U v_\nu),
\]
with analogous expressions for $W$.
These are the specialization of the Hessian relation
\eqref{eq:momentum-linearization}. 

The centered bracket of the variation parameters is 
$[\eta_U,\eta_W] =
\bigl([B,C],[B,c]-[C,b]\bigr)$.
Thus, identifying $\mathsf P^\mu=(\alpha^\mu,\varpi^\mu)$,
equation \eqref{eq:reduced-symplectic-current} yields
\begin{align}
\omega^\mu(U,W)
={}&
\langle\delta_W\alpha^\mu,B\rangle
-\langle\delta_U\alpha^\mu,C\rangle
\nonumber
+\langle\delta_W\varpi^\mu,b\rangle
-\langle\delta_U\varpi^\mu,c\rangle
\nonumber\\
&-
\langle\alpha^\mu,[B,C]\rangle
-
\langle\varpi^\mu,[B,c]-[C,b]\rangle,
\label{eq:matrix-MS-current}
\end{align}
where all pairings are the trace pairings introduced above.
The last two terms are precisely the coadjoint, heart, and
diamond contributions in \eqref{eq:centered-MS-current}:
\begin{equation*}
-\langle\alpha^\mu,[B,C]\rangle
-\langle\varpi^\mu,[B,c]-[C,b]\rangle
=
\langle\operatorname{ad}_B^*\alpha^\mu,C\rangle
-\langle B\heartsuit\varpi^\mu,c\rangle
-\langle b\diamondsuit\varpi^\mu,C\rangle .
\end{equation*}
In particular, although the quadratic Lagrangian has no mixed
constitutive terms, the current contains mixed contributions
from the centered Lie algebra structure.

By Proposition~\ref{prop:reduced-MSFF}, the current
\eqref{eq:matrix-MS-current} is closed. With
$\upsilon=\sqrt{|\det(\gamma_{\alpha\beta})|}$ as above,
its local conservation law is
\begin{equation}
\partial_\mu\bigl(\upsilon\,\omega^\mu(U,W)\bigr)=0.
\label{eq:matrix-MS-conservation}
\end{equation} In volume coordinates, this reduces to
$\partial_\mu\omega^\mu(U,W)=0$.
For this quadratic Lagrangian, the Hessian cancellation in
\eqref{eq:MSFF-Hessian-step} follows directly from the symmetry
of $\gamma^{\mu\nu}$ and the self-adjointness of
$\mathbb M$ and $\mathbb N$.

We next specialize the balance laws of
Section~\ref{sec:noether}, retaining the trivial reference
connection. Let $s=(g,u)$ denote a solution of the unreduced system, and consider a fixed infinitesimal generator
$\chi=(\chi^\g,\chi^V)\in\mathfrak{gl}(r)\Join\mathrm{Mat}(r)$.
The variation $s_\varepsilon=s\cdot\exp(\varepsilon\chi)$
has reduced parameter
$\eta_\chi=\operatorname{Ad}_{(g,u)}\chi=(B_\chi,b_\chi)$,
where
\begin{equation}\label{eq:matrix-Noether-parameter}
B_\chi=g\chi^\g g^{-1},
\qquad
b_\chi=g\chi^V g^{-1}+[ug^{-1},B_\chi].
\end{equation}

In the current formulas below, we use
$d^n x_\mu:=\iota_{\partial_\mu}\mathrm{vol}_X$. The associated Noether current is
$J_{\eta_\chi}=J_\chi^{\mathfrak g}+J_\chi^V$, with
\begin{equation}\label{eq:matrix-sector-currents}
J_\chi^{\mathfrak g}
=
\langle\alpha^\mu,B_\chi\rangle\,\mathrm d^n x_\mu,
\qquad
J_\chi^V
=
\langle\varpi^\mu,b_\chi\rangle\,\mathrm d^n x_\mu.
\end{equation}
By \eqref{eq:fixed-group-transport},
$\mathrm dJ_{\eta_\chi}=0$.
However, the sectorial currents obey the exchange laws
\begin{equation}\label{eq:matrix-sector-balance}
\mathrm dJ_\chi^{\mathfrak g}
=
\mathcal E_\chi\,{\rm vol}_X,
\qquad
\mathrm dJ_\chi^V
=
-\mathcal E_\chi\,{\rm vol}_X,
\end{equation}
where \eqref{eq:centered-exchange-rate} becomes
\begin{equation}
\mathcal E_\chi
=
-\langle v_\mu\diamondsuit\varpi^\mu,B_\chi\rangle
=
\langle\varpi^\mu,[B_\chi,v_\mu]\rangle
=
\gamma^{\mu\nu}
\langle\mathbb N v_\nu,[B_\chi,v_\mu]\rangle.
\label{eq:matrix-exchange-rate}
\end{equation}
Thus the centered coupling produces an explicit exchange
between the two sectorial currents while preserving their sum.

Formula \eqref{eq:matrix-Noether-parameter} also illustrates
Remark~\ref{rem:noether-generator-splitting}: even when $\chi^V=0$, the reduced parameter generally has
$b_\chi=[ug^{-1},B_\chi]\ne0$.
Hence the conserved current generated by $(\chi^g,0)$ should not
be identified with the projected sectorial contribution $J_\chi^g$.

Finally, the energy-momentum tensor
\eqref{eq:canonical-energy-momentum} is
\begin{align}
\mathsf T^\mu{}_\nu
&=
\langle\alpha^\mu,A_\nu\rangle
+
\langle\varpi^\mu,v_\nu\rangle
-\delta^\mu{}_\nu\ell=
\gamma^{\mu\lambda}
\left(
\langle\mathbb M A_\lambda,A_\nu\rangle
+
\langle\mathbb N v_\lambda,v_\nu\rangle
\right)
-\delta^\mu{}_\nu\ell.
\label{eq:matrix-energy-momentum}
\end{align}
If $\gamma^{\mu\nu}$, $\mathbb M$, and $\mathbb N$ are
independent of $x$,
the trivial horizontal lifts of spacetime translations
are infinitesimal symmetries. Their Noether currents and
the corresponding conservation laws are therefore
\[
\mathcal J_{(\partial_\nu)^{H}}
=
-\mathsf T^\mu{}_\nu\,\mathrm d^n x_\mu,
\qquad
\partial_\mu\mathsf T^\mu{}_\nu=0.
\]

The configuration variable $u(x)\in\operatorname{Mat}(r)$ provides
a matrix-valued internal descriptor, whereas $v_\mu$ is the
vector-space component of the reduced jet.
The field $g(x)$ may be interpreted as an internal frame or
microdeformation. Interpreting it as the gradient of a macroscopic
deformation requires additional compatibility conditions that are
not imposed here.

The centered semidirect product structure provides a geometric
framework for describing an internal matrix variable subject to
commuting left and right transformations, which may be interpreted
as changes of the surrounding frame and internal rearrangements,
respectively. The covariant Euler-Poincar\'e equations describe
the coupled evolution of the frame and microstructural degrees
of freedom. The sectorial currents $J_\chi^g$ and $J_\chi^V$
exhibit the associated momentum exchange while their sum is
conserved; the energy-momentum tensor $T^\mu{}_\nu$ gives the
translation balance laws under the stated invariance assumptions;
and the multisymplectic current $\omega^\mu$ expresses the
conservation law for pairs of first variations.
These distinct balance laws originate from the same underlying
variational and Poincar\'e-Cartan geometry.
In this sense, the present field theory can be viewed as a geometric
continuum analogue of the finite-dimensional centered semidirect
product dynamics in Colombo-Jacobs~\cite{CoJa2015}.

\subsection{Two-sided Matrix Microstructure and its Relation to
Micropolar and Liquid Crystal Variables}
\label{subsec:micropolar-LC}

Micropolar continua and liquid crystals provide an important point of
comparison for centered reduction, but the comparison has to be made at the
level of the underlying symmetry.  In a micropolar medium the internal
orientation is described by a microrotation, whereas in a uniaxial nematic the
physical order parameter is an unsigned director, hence a point of
$\mathbb{RP}^{2}\simeq \mathrm{SO}(3)/\mathrm{O}(2)$ (with the standard embedding
$\mathrm{O}(2)\subset \mathrm{SO}(3)$).  More generally, alignment tensors
belong to $\mathrm{SO}(3)$-orbits in the space of symmetric traceless matrices.  These
models are therefore naturally treated by reduction with broken symmetry,
advected order parameters, or affine Euler--Poincar\'e reduction
\cite{HoMaRa1998,GBRa2009,GBRa2012,GayBalmazTronci2010,
GayBalmazRatiuTronci2013}; see also the classical continuum treatments
\cite{Capriz1989,Eringen1999,deGennesProst1993,Virga1994}.  The purpose of the
present subsection is more specific: we determine the two-sided matrix model
which is covered by Theorem~\ref{thm:covariant-reduction-centered},
and then identify precisely its vector and alignment-tensor sectors.

This distinction is structural.  The reduction theory of Gay-Balmaz and
Tronci treats directors and alignment tensors as order parameters on
$\mathrm{SO}(3)$-orbits, and the comparison of Gay-Balmaz, Ratiu, and Tronci relates
the Ericksen-Leslie and micropolar descriptions only after the appropriate
order-parameter and compatibility data have been specified
\cite{GayBalmazTronci2010,GBRa2012,GayBalmazRatiuTronci2013}.  Those results do
not identify an arbitrary matrix semidirect-product model with nematodynamics.
Here we instead start from the actual centered action, compute its effective
representation, and only then compare its invariant sectors with the standard
liquid-crystal variables.

We begin with an obstruction which rules out a direct centered construction on
the director space.

\begin{proposition}
\label{prop:director-centered-obstruction}
Let the left action of $\mathrm{SO}(3)$ on $\mathbb R^{3}$ be the standard
representation.  Every smooth linear right action of $\mathrm{SO}(3)$ on $\mathbb R^3$
which commutes with this left action is trivial.  Consequently,
$\mathrm{SO}(3)\Join\mathbb R^3$ cannot carry two independent nontrivial actions with
the standard director representation as its left action.
\end{proposition}

\begin{proof}
The commutant of the standard real representation of $\mathrm{SO}(3)$ consists of the
scalar endomorphisms.  Hence a commuting right action has the form
$v\cdot Q=\lambda(Q)v$.  The right-action identity implies that $\lambda$ is a
smooth character of $\mathrm{SO}(3)$ with values in $\mathbb R^{\times}$.  Since
$\mathrm{SO}(3)$ is connected and compact, its image is a connected compact subgroup of
$\mathbb R^{\times}$, and is therefore $\{1\}$.  Thus $v\cdot Q=v$ for every
$Q\in \mathrm{SO}(3)$.
\end{proof}

A natural matrix space on which both actions remain explicit is
$V=\operatorname{Mat}(3)$.  We consider
$H_{\rm mat}:=\mathrm{SO}(3)\Join\operatorname{Mat}(3)$ with
\begin{equation}
R\cdot A:=RA,
\qquad
A\cdot Q:=AQ,
\qquad
R,Q\in \mathrm{SO}(3),\quad A\in\operatorname{Mat}(3).
\label{eq:SO3-Mat3-actions}
\end{equation}
The actions commute by associativity.  As recalled in the Introduction, the
centered product is isomorphic to an ordinary semidirect product.  Here the
isomorphism and the induced representation are
\begin{equation}
\Phi:H_{\rm mat}\rightarrow
\mathrm{SO}(3)\ltimes_{\rho}\operatorname{Mat}(3),
\qquad
\Phi(R,A)=(R,B),\quad B:=AR^{-1},
\qquad
\rho(R)B=RBR^{-1}.
\label{eq:matrix-centered-semidir-isomorphism}
\end{equation}
Indeed, if $(R_i,A_i)$ corresponds to $(R_i,B_i)$, then
\begin{equation}
\Phi\big((R_1,A_1)(R_2,A_2)\big)
=
\big(R_1R_2,B_1+R_1B_2R_1^{-1}\big).
\label{eq:matrix-semidir-product}
\end{equation}
Thus $B$, rather than a fixed column of $A$, is the matrix variable on which
the effective $\mathrm{SO}(3)$-representation acts.

For the explicit field equations, take the trivial principal
$H_{\rm mat}$-bundle and the zero reference connection.  Let
$s=(R,A):X\to H_{\rm mat}$ and use the right logarithmic derivative, which is
the convention underlying the constrained variation
\eqref{eq:delta-xi-CGPR}:
\begin{equation}
\Xi_\mu=(\Omega_\mu,W_\mu)
:=(\partial_\mu s)s^{-1}.
\label{eq:matrix-right-reduced-jet}
\end{equation}
A direct use of the inverse and multiplication in
\eqref{eq:centered-product} gives
\begin{align}
& \Omega_\mu=(\partial_\mu R)R^{-1}\in\mathfrak{so}(3),
\nonumber\\
& W_\mu=(\partial_\mu A)R^{-1}-\Omega_\mu B
 =\partial_\mu B+B\Omega_\mu-\Omega_\mu B
 \in\operatorname{Mat}(3).
\label{eq:matrix-reduced-components}
\end{align}
In particular, $W_\mu=\partial_\mu B-\rho_*(\Omega_\mu)B$.  This formula also
provides a direct consistency check with the general theory.  If
$\eta=(\Sigma,Z)=\delta s\,s^{-1}$, then
\begin{align}
& \delta\Omega_\mu=\partial_\mu\Sigma-[\Omega_\mu,\Sigma],
\nonumber\\
& \delta W_\mu=\partial_\mu Z+Z\spadesuit\Omega_\mu
-W_\mu\spadesuit\Sigma,
\label{eq:matrix-constrained-variations}
\end{align}
which is exactly \eqref{eq:delta-sigma}-\eqref{eq:delta-nu} for
$\mathcal A=0$.

We identify $\mathfrak{so}(3)^*$ with $\mathfrak{so}(3)$ and
$\operatorname{Mat}(3)^*$ with $\operatorname{Mat}(3)$ by the Frobenius
pairing.  Since the first entry of the diamond operator belongs to
$\operatorname{Mat}(3)$ while its value belongs to
$\mathfrak{so}(3)^*$, the orthogonal projection onto skew matrices cannot be
omitted.  The operators \eqref{eq:heart-def}-\eqref{eq:diamond-def} become
\begin{align}
& \Omega\heartsuit N=\Omega^{\mathsf T}N-N\Omega^{\mathsf T}
=N\Omega-\Omega N,
\label{eq:matrix-heart-so3}\\
& W\diamondsuit N=\operatorname{skew}\!\left(W^{\mathsf T}N-NW^{\mathsf T}\right),
\qquad
\operatorname{skew}(C):=\frac12(C-C^{\mathsf T}),
\label{eq:matrix-diamond-so3}
\end{align}
and, with the coadjoint convention fixed in Section~2,
$\operatorname{ad}^*_{\Omega}M=[\Omega,M]$.  These formulas follow directly
from
\[
\langle W\diamondsuit N,\Omega\rangle
=\operatorname{tr}\!\left(N^{\mathsf T}(W\Omega-\Omega W)\right),
\qquad
\langle\Omega\heartsuit N,W\rangle
=\operatorname{tr}\!\left(N^{\mathsf T}(\Omega W-W\Omega)\right).
\]

Let $(X,\gamma)$ be an oriented semi-Riemannian manifold,
with $\mathrm{vol}_X$ its metric volume form.  Choose symmetric
nondegenerate constitutive isomorphisms $\mathbb J:\mathfrak{so}(3)\rightarrow\mathfrak{so}(3)^*$ and $\mathbb K:\operatorname{Mat}(3)\rightarrow
\operatorname{Mat}(3)^*$. The quadratic reduced Lagrangian
\begin{equation}
\ell(\Omega,W)
=\frac12\gamma^{\mu\nu}
\left(
\langle\mathbb J\Omega_\mu,\Omega_\nu\rangle
+\langle\mathbb K W_\mu,W_\nu\rangle
\right)
\label{eq:matrix-microstructure-lagrangian}
\end{equation}
defines the right $H_{\rm mat}$-invariant first-order density
$\mathcal L(j^1s)=\ell((\mathrm ds)s^{-1})\,{\rm vol}_X$ in the displayed
trivialization.  Its
covariant momenta are
\begin{equation}
M^\mu:=\frac{\delta\ell}{\delta\Omega_\mu}
=\gamma^{\mu\nu}\mathbb J\Omega_\nu,
\qquad
N^\mu:=\frac{\delta\ell}{\delta W_\mu}
=\gamma^{\mu\nu}\mathbb K W_\nu.
\label{eq:matrix-microstructure-momenta}
\end{equation}
The reduced Legendre map is fiberwise invertible whenever $\mathbb J$ and
$\mathbb K$ are invertible.  Hence this example also falls under
Proposition~\ref{prop:cov-LP}.

\begin{proposition}
\label{prop:centered-matrix-microstructure}
The Euler-Poincar\'e equations associated with
\eqref{eq:matrix-microstructure-lagrangian} are
\begin{align}
& \operatorname{div}M= [\Omega_\mu,M^\mu]
-\operatorname{skew}\!\left(
W_\mu^{\mathsf T}N^\mu-N^\mu W_\mu^{\mathsf T}
\right),
\label{eq:matrix-microstructure-EP-g}\\
& \operatorname{div}N=[\Omega_\mu,N^\mu].
\label{eq:matrix-microstructure-EP-V}
\end{align}
Here and below the repeated spacetime index is contracted, and the divergence
is the one determined by the volume form, as in Section~2.
\end{proposition}

\begin{proof}
Substitution of \eqref{eq:matrix-heart-so3}-
\eqref{eq:matrix-diamond-so3} into
\eqref{eq:centeredeqsg}-\eqref{eq:centeredeqsV}, with
$\mathcal A=0$, gives \eqref{eq:matrix-microstructure-EP-g}-
\eqref{eq:matrix-microstructure-EP-V}.
\end{proof}

Equations \eqref{eq:matrix-microstructure-EP-g}-
\eqref{eq:matrix-microstructure-EP-V} are the dynamical equations.  A pair
$(\Omega,W)$ represents the reduced jet of an $H_{\rm mat}$-valued field only
if it also satisfies the reconstruction equations.  With the right
logarithmic convention \eqref{eq:matrix-right-reduced-jet}, these are the
Maurer-Cartan equations
\begin{align}
& 0=\mathcal F^{\Omega}_{\mu\nu}:=\partial_\mu\Omega_\nu-\partial_\nu\Omega_\mu
-[\Omega_\mu,\Omega_\nu],
\label{eq:matrix-reconstruction-g}\\
& 0=\mathcal F^{W}_{\mu\nu}:=\partial_\mu W_\nu-\partial_\nu W_\mu
-\rho_*(\Omega_\mu)W_\nu
+\rho_*(\Omega_\nu)W_\mu.
\label{eq:matrix-reconstruction-V}
\end{align}
On a simply connected domain they are also sufficient for local
reconstruction, up to the choice of one value of $s$.

The relation with liquid-crystal variables is most transparent in the
semidirect coordinate $B$ of
\eqref{eq:matrix-centered-semidir-isomorphism}.  The $\mathrm{SO}(3)$-module
$\operatorname{Mat}(3)$ has the invariant decomposition
\begin{equation}
\operatorname{Mat}(3)
=\mathbb R I\oplus\mathfrak{so}(3)\oplus\operatorname{Sym}_0(3).
\label{eq:matrix-so3-decomposition}
\end{equation}
Accordingly, write
\begin{equation}
B=\beta I+\widehat n+Q,
\qquad
\beta=\frac13\operatorname{tr}B,\qquad
\widehat n=\operatorname{skew}B,\qquad
Q=\operatorname{sym}B-\frac13\operatorname{tr}(B)I.
\label{eq:matrix-order-parameters}
\end{equation}
Under $\rho(R)B=RBR^{-1}$, the three components transform as
$\beta\mapsto\beta,
\, n\mapsto Rn,
\, Q\mapsto RQR^{\mathsf T}$. Thus $n$ transforms as a vector and $Q$ as a symmetric traceless tensor.
For $n\ne0$, the line $[n]=[n/\|n\|]\in\mathbb{RP}^2$
defines a projective director.
However, invariance under $n\mapsto-n$ alone does not remove the
amplitude $\|n\|$.
A director-only theory requires fixing the norm and identifying the
two lifts, or imposing the corresponding scale independence on the
constitutive law and its jet dependence.
These additional conditions are not part of the matrix model considered here.

If $\Omega_\mu=\widehat\omega_\mu$, equation
\eqref{eq:matrix-reduced-components} decomposes exactly as
\begin{equation}
W_\mu
=(\partial_\mu\beta)I
+\widehat{\mathcal D_\mu n}
+\mathcal D_\mu Q,
\qquad
\mathcal D_\mu n:=\partial_\mu n-\omega_\mu\times n,
\qquad
\mathcal D_\mu Q:=\partial_\mu Q-[\Omega_\mu,Q].
\label{eq:matrix-order-covariant-gradients}
\end{equation} 

Consequently, if $\mathbb K$ preserves the decomposition
\eqref{eq:matrix-so3-decomposition}, the matrix part of
\eqref{eq:matrix-microstructure-lagrangian} separates into scalar, vector, and
alignment-tensor channels.  In the isotropic case it has the form
\begin{equation}
\ell_{\rm int}
=\frac12\gamma^{\mu\nu}
\left(
\kappa_0\,\partial_\mu\beta\,\partial_\nu\beta
+\kappa_1\,\mathcal D_\mu n\cdot\mathcal D_\nu n
+\kappa_2\,\operatorname{tr}
\big((\mathcal D_\mu Q)^{\mathsf T}\mathcal D_\nu Q\big)
\right).
\label{eq:matrix-order-sector-lagrangian}
\end{equation}
The constants $\kappa_i$ include the normalization factors of the
Frobenius pairing. More explicitly, if $K$ acts by the scalars
$\lambda_0$, $\lambda_1$, and $\lambda_2$ on
$\mathbb RI$, $\mathfrak{so}(3)$, and $\operatorname{Sym}_0(3)$,
respectively, then
$\kappa_0=3\lambda_0$, $\kappa_1=2\lambda_1$,
and $\kappa_2=\lambda_2$. The last two terms are quadratic covariant-gradient densities for
an unconstrained vector and a symmetric traceless matrix.
They provide a comparison with director and alignment-tensor gradient
energies, without imposing the order-parameter constraints of those theories.
Indices on covariant derivatives are raised with $\gamma$;
the divergence includes the metric volume density and the base
Levi--Civita connection. For constant $\kappa_i$, the projected
equations are
\begin{equation}
\operatorname{div}(\kappa_0(d\beta)^\sharp)=0,
\qquad
\mathcal D_\mu(\kappa_1\mathcal D^\mu n)=0,
\qquad
\mathcal D_\mu(\kappa_2\mathcal D^\mu Q)=0.
\label{eq:matrix-order-sector-equations}
\end{equation}
These equations describe the covariant dynamics of the unconstrained
vector and tensor channels. The general rotational equation contains
an internal torque, whose value depends on the constitutive map;
in particular, this torque vanishes for the isotropic constitutive
choice just described.

\begin{remark}
\label{rem:scope-liquid-crystal-comparison}
Equation \eqref{eq:matrix-order-covariant-gradients} is an exact representation
statement and \eqref{eq:matrix-order-sector-lagrangian} gives reduced
densities covered by the centered theory.  It does not, by itself, identify the
full model with Ericksen-Leslie, Eringen, or Oseen-Frank dynamics.  The unit
and projective constraints on a director and the bulk potential of a
$Q$-tensor depend on the order parameter itself, not only on the reduced jet
$(\Omega,W)$.  They therefore require the broken-symmetry or advected-parameter
extensions developed in
\cite{GBRa2009,GBRa2012,GayBalmazTronci2010,
GayBalmazRatiuTronci2013}.  In particular, the column map $A\mapsto Ad_0$ is
not equivariant for the effective representation $B\mapsto RBR^{-1}$, and
$W_\mu d_0$ is not the covariant derivative of $Ad_0$.  The invariant
decomposition \eqref{eq:matrix-so3-decomposition}, rather than a fixed column
of $A$, is what yields the mathematically consistent vector and alignment
sectors.
\end{remark}\medskip

\begin{remark}
\label{rem:advantages-matrix-formulation}
The three summands of $\operatorname{Mat}(3)$ under conjugation are
pairwise non-isomorphic irreducible real $\operatorname{SO}(3)$-modules,
each occurring with multiplicity one. Consequently, an
$\operatorname{SO}(3)$-equivariant linear constitutive map preserves
each summand and acts on it by a scalar. Such a map cannot couple
distinct summands. Off-diagonal constitutive couplings require
additional anisotropic data or an enlarged representation with
equivalent summands.

There is also a useful dynamical consequence. If $K$ is self-adjoint
and equivariant, invariance of its quadratic form gives
$\langle KW,[\Sigma,W]\rangle=0$ for every
$\Sigma\in\mathfrak{so}(3)$.
Polarization and contraction with the symmetric tensor
$\gamma^{\mu\nu}$ therefore yield
$W_\mu\diamondsuit N^\mu=0$.
Thus the internal torque in the rotational equation vanishes for
this constitutive choice. If $J$ is isotropic as well, the contracted
term $[\Omega_\mu,M^\mu]$ also vanishes.
The general reduced equations retain these terms because anisotropic
constitutive maps need not have these cancellations.

The representation decomposition and covariant derivatives remain
useful in both cases. Unit-director constraints, projective
identification, bulk potentials, and anisotropic Oseen--Frank moduli
are additional constitutive or broken-symmetry data; they are not
imposed by the centered product itself.
\end{remark}\medskip

Finally, the specialization of the reduced multisymplectic current is also
completely explicit.  For first variations $\mathsf U,\mathsf V$ with reduced
parameters
$\eta_{\mathsf U}=(\Sigma_{\mathsf U},Z_{\mathsf U})$ and
$\eta_{\mathsf V}=(\Sigma_{\mathsf V},Z_{\mathsf V})$, Corollary~
\ref{cor:centered-MS-current} gives
\begin{align}
\omega^\mu_{\rm mat}(\mathsf U,\mathsf V)
={}&
\langle\delta_{\mathsf V}M^\mu,\Sigma_{\mathsf U}\rangle
-\langle\delta_{\mathsf U}M^\mu,\Sigma_{\mathsf V}\rangle
+\langle\delta_{\mathsf V}N^\mu,Z_{\mathsf U}\rangle
-\langle\delta_{\mathsf U}N^\mu,Z_{\mathsf V}\rangle
\nonumber\\
&+\langle\operatorname{ad}^*_{\Sigma_{\mathsf U}}M^\mu,
\Sigma_{\mathsf V}\rangle
-\langle\Sigma_{\mathsf U}\heartsuit N^\mu,Z_{\mathsf V}\rangle
-\langle Z_{\mathsf U}\diamondsuit N^\mu,
\Sigma_{\mathsf V}\rangle.
\label{eq:matrix-multisymplectic-current}
\end{align}
The last two terms are the centered coupling contributions to the
symplectic current. For first variations satisfying the linearized
field equations, the current is closed. In coordinates with
$\mathrm{vol}_X=\upsilon(x)\,dx^0\wedge\cdots\wedge dx^n$,
this reads $\partial_\mu(\upsilon\,\omega_{\mathrm{mat}}^\mu)=0$;
in volume coordinates it reduces to
$\partial_\mu\omega_{\mathrm{mat}}^\mu=0$.
The corresponding boundary-flux identity is the one established
in Section~\ref{sec:multisymplecticformula}.

The Noether and energy-momentum constructions of Section~\ref{sec:noether}
specialize in the same way. For a fixed generator
$\chi=(\chi^g,\chi^V)$, set
$\Sigma_\chi=R\chi^gR^{-1}$ and
$Z_\chi=R\chi^VR^{-1}+[B,\Sigma_\chi]$.
The associated current is
$J_\chi=(\langle M^\mu,\Sigma_\chi\rangle
+\langle N^\mu,Z_\chi\rangle)d^n x_\mu$,
where $d^n x_\mu=\iota_{\partial_\mu}\mathrm{vol}_X$,
and satisfies $dJ_\chi=0$.
Its projected sectorial currents have exchange density
$-\langle W_\mu\diamondsuit N^\mu,\Sigma_\chi\rangle$,
which vanishes for the equivariant constitutive choice discussed above.
The energy--momentum tensor is
$T^\mu{}_{\nu}=\langle M^\mu,\Omega_\nu\rangle
+\langle N^\mu,W_\nu\rangle-\delta^\mu{}_{\nu}\ell$.
For constant metric coefficients and spatially and temporally
homogeneous constitutive maps, translations give
$\partial_\mu T^\mu{}_{\nu}=0$.

\section{Centered Reduction for a First-order Theory with Second-order Jet Group
Symmetry}
\label{sec:higher-order-gauge}

The group of second-order jets is a basic geometric example of a centered
semidirect product: its two actions are forced by the two terms in the
second-order chain rule.  It was the main geometric motivation in
\cite{CoJa2015}, while the underlying jet geometry is developed systematically
in \cite{KMS1993,Saunders1989}.  Three notions must be kept separate: the
second-order jet used as an internal structure group, the first jet $j^1s$ of
the dynamical field, and a higher-order variational Lagrangian on $J^kP$.
The construction below uses the first two and remains entirely within the
first-order theory of Section~\ref{sec:covariantreduction}; the third belongs to the higher-order
Euler-Poincar\'e theory of \cite{CastrillonRodriguezHigherOrder2026}.

\subsection{Geometric Setting}\label{sec5.1}

Let $E=\mathbb R^r$ and set $V_2:=S^2E^*\otimes E$.
Thus an element $T\in V_2$ is a symmetric bilinear map
$T:E\times E\to E$.  Denote by $G_r^2$ the group of second-order jets at the
origin of local diffeomorphisms $\varphi:E\to E$ satisfying $\varphi(0)=0$.
The map
\begin{equation}
G_r^2\simeq\operatorname{GL}(E)\times V_2,\qquad j_0^2\varphi\longmapsto
\big(g,T\big):=\big(D\varphi(0),D^2\varphi(0)\big)
\label{eq:jet2-coordinate-identification}
\end{equation}
defines a diffeomorphism of smooth manifolds.
For $g\in\operatorname{GL}(E)$ and $T\in V_2$, define
\begin{equation}
(g\cdot T)(u,v):=g\big(T(u,v)\big),
\qquad
(T\cdot g)(u,v):=T(gu,gv),
\label{eq:jet2-left-right-actions}
\end{equation}
for each $(u,v)\in E\times E$. Clearly, the previous expressions define a left and a right action, respectively, that commute.  Moreover, the chain rule gives
\begin{align}
D^2(\varphi\circ\psi)(0)(u,v)
={}&D\varphi(0)D^2\psi(0)(u,v)+D^2\varphi(0)
\big(D\psi(0)u,D\psi(0)v\big),
\label{eq:second-order-chain-rule}
\end{align}
or, equivalently,
\begin{equation}
(g_1,T_1)(g_2,T_2)
=\big(g_1g_2,~g_1\cdot T_2+T_1\cdot g_2\big),
\label{eq:jet2-centered-product}
\end{equation}
where we denoted $g_1=D\varphi(0)$, $g_2=D\psi(0)$, $T_1=D^2\varphi(0)$ and $T_2=D^2\psi(0)$. In addition, concatenation denotes composition, i.e., $g_1g_2=g_1\circ g_2$.
It follows that~\eqref{eq:jet2-coordinate-identification} defines an isomorphism of Lie groups,
\begin{equation}
G_r^2\simeq\operatorname{GL}(E)\Join V_2,
\label{eq:jet2-centered-identification}
\end{equation}
where the multiplication on $G_r^2$ is given by (the jet lift of) the composition of diffeomorphisms.

\begin{remark}
Following Proposition~\ref{prop:centered-standard-equivalence}, the associated ordinary semidirect-product representation is the natural
tensor representation
\begin{equation}
\rho(g)B:=(g\cdot B)\cdot g^{-1},
\qquad
(\rho(g)B)(u,v)=gB(g^{-1}u,g^{-1}v),
\label{eq:jet2-effective-representation}
\end{equation}
for each $(g,T)\in\operatorname{GL}(E)\Join V_2$ and $(u,v)\in E\times E$, where we denote $B:=T\cdot g^{-1}\in V_2$. The centered-to-semidirect isomorphism is thus
$(g,T)\mapsto\Phi(g,T)=(g,B)$.
\end{remark}\medskip

At the infinitesimal level, for $a\in\mathfrak{gl}(E)$ and $S\in V_2$,
\begin{equation}
a\cdot S=a\circ S,
\qquad
(S\cdot a)(u,v)=S(au,v)+S(u,av),
\qquad
\rho_*(a)S=a\cdot S-S\cdot a,
\label{eq:jet2-infinitesimal-actions}
\end{equation}
for each $(u,v)\in E\times E$.

Given a basis $\{e_i\mid 1\leq i\leq r\}$ of $E$, with dual basis $\{e^i\mid 1\leq i\leq r\}$ of $E^*$, the infinitesimal actions read
\begin{equation}
(a\cdot S)^i{}_{jk}=a^i{}_{\ell}S^{\ell}{}_{jk},
\qquad
(S\cdot a)^i{}_{jk}
=S^i{}_{\ell k}a^{\ell}{}_{j}
+S^i{}_{j\ell}a^{\ell}{}_{k}.
\label{eq:jet2-infinitesimal-actions-indices}
\end{equation}
where we write $a=a^i{}_j\,e^j\otimes e_i\in\mathfrak{gl}(E)$ and $S=S^i{}_{jk}\,e^j\otimes e^k\otimes e_i\in V_2$. Recall that we use the Einstein summation convention. Moreover, we retain the canonical dualities rather than identifying the vector spaces with their corresponding dual spaces. Hence, the dual pairings read
\[
\langle\Pi,S\rangle=\Pi_i{}^{jk}S^i{}_{jk},
\qquad
\langle P,a\rangle=P_i{}^j a^i{}_j,
\]
for each $\Pi=\Pi_i{}^{jk}\,e_j\otimes e_k\otimes e^i\in V_2^*= S^2 E\otimes E^*$ and $P=P_i{}^j e_j\otimes e^i\in\mathfrak{gl}(E)^*$. 
Using \eqref{eq:jet2-infinitesimal-actions-indices}, the heart~\eqref{eq:heart-def} and diamond~\eqref{eq:diamond-def} operators are given by
\begin{align}
& (a\heartsuit\Pi)_i{}^{jk}=a^{\ell}{}_{i}\Pi_{\ell}{}^{jk}
-\Pi_i{}^{\ell k}a^j{}_{\ell}
-\Pi_i{}^{j\ell}a^k{}_{\ell},
\label{eq:jet2-heart-indices}\\
& (S\diamondsuit\Pi)_a{}^b=2\Pi_i{}^{bk}S^i{}_{ak}
-\Pi_a{}^{jk}S^b{}_{jk}.
\label{eq:jet2-diamond-indices}
\end{align}
The factor two in \eqref{eq:jet2-diamond-indices} is a consequence of the
symmetry of the two covariant indices. Note that these equations do not require the choice of a metric on $E$.

Let $(X,\gamma)$ be an oriented, $(n+1)$-dimensional, semi-Riemannian manifold with volume form $\operatorname{vol}_X$, and consider the trivial principal $G_r^2$-bundle  $P=X\times G_r^2\to X$. A field is thus a smooth map $s=(g,T):X\to G_r^2\simeq\mathrm{GL}(E)\Join V_2$. 
We choose the zero reference connection on the trivial principal
bundle. The reduced field is then the right logarithmic derivative
$\Xi=(A,C):=(ds)s^{-1}\in
\Omega^1(X,\mathfrak{gl}(E)\Join V_2)$.

Given local coordinates $(x^\mu)_{1\leq\mu\leq n+1}$ on $X$, the components of the reduced section,
\begin{equation*}
\Xi_\mu=(A_\mu,C_\mu):=(\partial_\mu s)s^{-1}:X\to\mathfrak{gl}(E)\Join V_2,  
\end{equation*}
are given by
\begin{equation}
\begin{cases}
A_\mu
&=(\partial_\mu g)g^{-1},
\nonumber\\
C_\mu
&=(\partial_\mu T)\cdot g^{-1}-A_\mu\cdot B
=\partial_\mu B+B\cdot A_\mu-A_\mu\cdot B
=\partial_\mu B-\rho_*(A_\mu)B,    
\end{cases}
\label{eq:jet2-right-reduced-field}
\end{equation}
where we denoted $B:=T\cdot g^{-1}$. Thus $A_\mu$ transports the contravariant index of the quadratic jet and both
covariant indices enter through the right action in $C_\mu$.  This is the
field-theoretic form of the same chain-rule coupling as
\eqref{eq:second-order-chain-rule}.

\subsection{Euler-Poincar\'e Field Equations on the Second-order Jet Group}

Let $\{e_i\mid1\leq i\leq r\}$ be a basis of $E$ and $(x^\mu)_{1\leq\mu\leq n+1}$ be local coordinates on $X$, as above. Let us consider a reduced Lagrangian density,
\begin{equation}
\ell_{\rm jet}(A,C)
=\frac12\gamma^{\mu\nu}
\left(
\langle\mathbb I_1A_\mu,A_\nu\rangle
+\langle\mathbb I_2C_\mu,C_\nu\rangle
\right),\qquad(A,C):X\to T^*X\otimes(\mathfrak{gl}(E)\Join V_2),
\label{eq:jet2-reduced-lagrangian}
\end{equation}
where $\mathbb I_1:\mathfrak{gl}(E)\longrightarrow\mathfrak{gl}(E)^*$ and $\mathbb I_2:V_2\longrightarrow V_2^*$ are symmetric and nondegenerate linear maps. This defines a right
$G_r^2$-invariant first-order Lagrangian density as $\mathcal L(j^1s)=\ell_{\rm jet}((\mathrm ds)s^{-1})\,{\rm vol}_X$.

\begin{remark}
On a nontrivial principal bundle, a global quadratic constitutive
law is specified by a symmetric nondegenerate bundle morphism
$\mathbb I:\operatorname{Ad}(P)\to\operatorname{Ad}(P)^*$.
A local splitting identifies its components with the linear-jet and
quadratic-jet sectors, but mixed constitutive blocks may be present.
The block-diagonal choice $\mathbb I_1\oplus\mathbb I_2$ is an
additional constitutive assumption in the chosen splitting.
\end{remark}\medskip

Momenta $P=(P_A,P_C):X\to TX\otimes(\mathfrak{gl}(E)^*\Join V_2^*)$ locally read $P_A=P_A^\mu\partial_\mu$ and
$P_C=P_C^\mu\partial_\mu$ with
\begin{equation}
P_A^\mu=\gamma^{\mu\nu}\mathbb I_1A_\nu:X\to\mathfrak{gl}(E)^*,
\qquad
P_C^\mu=\gamma^{\mu\nu}\mathbb I_2C_\nu:X\to V_2^*.
\label{eq:jet2-momenta}
\end{equation}

Since $\gamma$, $\mathbb I_1$, and $\mathbb I_2$ are nondegenerate,
the reduced Legendre map is invertible, with
$A_\mu=\gamma_{\mu\nu}\mathbb I_1^{-1}P_A^\nu$ and
$C_\mu=\gamma_{\mu\nu}\mathbb I_2^{-1}P_C^\nu$.
The reduced Hamiltonian is
$h_{\mathrm{jet}}=\tfrac12\gamma_{\mu\nu}
(\langle P_A^\mu,\mathbb I_1^{-1}P_A^\nu\rangle
+\langle P_C^\mu,\mathbb I_2^{-1}P_C^\nu\rangle)$.
Thus the Hamiltonian equations of Section~\ref{sec:LiePoisson} give exactly the
Euler-Poincar\'e equations below, subject to the same reconstruction
conditions.

\begin{proposition}
\label{prop:jet2-field-equations}
The covariant Euler-Poincar\'e equations for the reduced Lagrangian density
\eqref{eq:jet2-reduced-lagrangian} are
\begin{align}
& \operatorname{div}(P_A)=\operatorname{ad}^*_{A}(P_A)
-C\diamondsuit P_C,
\label{eq:jet2-EP-linear}\\
& \operatorname{div}(P_C)=-A\heartsuit P_C.
\label{eq:jet2-EP-quadratic}
\end{align}
Locally, they read
\begin{align}
&(\operatorname{div}P_A)_a{}^b=(P_A^\mu)_a{}^j(A_\mu)^b{}_j
-(A_\mu)^i{}_a(P_A^\mu)_i{}^b-2(P_C^\mu)_i{}^{bk}(C_\mu)^i{}_{ak}
+(P_C^\mu)_a{}^{jk}(C_\mu)^b{}_{jk},
\label{eq:jet2-EP-linear-indices}\\
&(\operatorname{div}P_C)_i{}^{jk}=-(A_\mu)^{\ell}{}_{i}(P_C^\mu)_{\ell}{}^{jk}+(P_C^\mu)_i{}^{\ell k}(A_\mu)^j{}_{\ell}
+(P_C^\mu)_i{}^{j\ell}(A_\mu)^k{}_{\ell}.
\label{eq:jet2-EP-quadratic-indices}    
\end{align}
\end{proposition}

\begin{proof}
Equations \eqref{eq:jet2-EP-linear}-
\eqref{eq:jet2-EP-quadratic} are
\eqref{eq:centeredeqsg}-\eqref{eq:centeredeqsV} with
$G=\operatorname{GL}(E)$, $V=V_2$, and $\mathcal A=0$. The local expressions \eqref{eq:jet2-EP-linear-indices}-\eqref{eq:jet2-EP-quadratic-indices} readily follow
from \eqref{eq:jet2-heart-indices}-\eqref{eq:jet2-diamond-indices}, as well as the local expression for the coadjoint representation, $(\operatorname{ad}^*_A P)_a{}^b
=P_a{}^jA^b{}_j-A^i{}_aP_i{}^b$.
\end{proof}

As in the preceding example, the Euler-Poincar\'e equations must be accompanied
by the reconstruction condition~\eqref{eq:reconstructioncompatibility}. Following Remark~\ref{rem:reference-connection-sign}, locally it is given by~\eqref{eq:right-MC-sign-convention}. In coordinates, we can write
\begin{align}
& \mathcal F^A_{\mu\nu}:=\partial_\mu A_\nu-\partial_\nu A_\mu-[A_\mu,A_\nu]=0,
\label{eq:jet2-reconstruction-linear}\\
& \mathcal F^C_{\mu\nu}:=\partial_\mu C_\nu-\partial_\nu C_\mu
-\rho_*(A_\mu)C_\nu+\rho_*(A_\nu)C_\mu=0,
\label{eq:jet2-reconstruction-quadratic}
\end{align}
for each $1\leq\mu,\nu\leq n+1$. From~\eqref{eq:jet2-infinitesimal-actions} and~\eqref{eq:jet2-infinitesimal-actions-indices}, the second equation may be equivalently written as
\begin{align*}
\mathcal F^C_{\mu\nu}
={}&\partial_\mu C_\nu-\partial_\nu C_\mu
-A_\mu\cdot C_\nu-C_\mu\cdot A_\nu+A_\nu\cdot C_\mu+C_\nu\cdot A_\mu=0.
\end{align*}

Consequently, on a simply connected domain, a solution of
\eqref{eq:jet2-EP-linear}-\eqref{eq:jet2-EP-quadratic} satisfying
\eqref{eq:jet2-reconstruction-linear}-
\eqref{eq:jet2-reconstruction-quadratic} reconstructs an unreduced
$G_r^2$-valued field, uniquely after fixing one value of $s$.  This explicitly
realizes, for the second-order jet group, the reconstruction statement following
\eqref{eq:intrinsic-EP}.

\begin{remark}
\label{rem:jet-versus-variational-order}
The adjective ``second-order'' refers to the internal structure group $G_r^2$,
not to the variational order.  In particular,
\eqref{eq:jet2-reduced-lagrangian} is first order in $s$ and lies in the setting
of Theorem~\ref{thm:covariant-reduction-centered}.  A Lagrangian on $J^kP$,
$k\geq2$, has a different reduced configuration space and is governed by the
higher-order equations of \cite{CastrillonRodriguezHigherOrder2026}.
\end{remark}\medskip

\begin{remark}
\label{rem:jet2-not-yang-mills}
The reduced variable $\Xi=(A,C)$ should not be confused with an arbitrary
$G_r^2$-gauge potential.  In the present first-order reduction it is induced
by a section of the principal bundle and must satisfy
\eqref{eq:jet2-reconstruction-linear}-
\eqref{eq:jet2-reconstruction-quadratic}.  A Yang-Mills density depending on
the curvature of an independent connection is a Lagrangian on a jet of the
connection bundle; expressed in terms of $s$, it is of higher variational
order.  Moreover, a curvature-squared density vanishes identically on the
reconstructible locus \eqref{eq:jet2-reconstruction-linear}-
\eqref{eq:jet2-reconstruction-quadratic}.  It is therefore not an example of
the first-order theory established in this paper.
\end{remark}\medskip

The relevance of $G_r^2$ beyond this intrinsic example is well established:
jet data underlie locally affine descriptions in diffeomorphic image
registration and the hierarchy of jetlet models for regularized ideal fluids
\cite{SommerNielsenDarknerPennec2013,
CotterElderingHolmJacobsMeier2016}.  Passing from the internal jet-group field
considered here to a spatially holonomic second gradient of a deformation
would require additional contact or Spencer compatibility conditions; it
should not be inferred from the group-valued field alone.  The present model
isolates exactly the contribution supplied by centered covariant reduction:
the chain-rule group law, its heart and diamond couplings, the corresponding
Euler-Poincar\'e equations, and their multisymplectic conservation law.

\subsection{A Nonlinear $G_1^2$ Wave Reduction and an Exact Conversion Front}
\label{subsec:jet2-conversion-wave}

We now exhibit a nontrivial dynamical consequence of the centered coupling.
The internal jet dimension is reduced to one, while the base is the
$1+1$-dimensional Minkowski space
$X=\mathbb R^{1,1}$ with coordinates $(t,x)$ and metric
$\gamma=\operatorname{diag}(-1,1)$.  In this case
\[
G_1^2\simeq \mathbb R^\times\Join\mathbb R,
\qquad
g\cdot S=gS,
\qquad
S\cdot g=g^2S.
\]
Hence, for $a,S,\Pi\in\mathbb R$,
\begin{equation}
\rho_*(a)S=-aS,
\qquad
a\heartsuit\Pi=-a\Pi,
\qquad
S\diamondsuit\Pi=S\Pi.
\label{eq:jet1-actions}
\end{equation}
Write the reduced one-form as
\[
\Xi=A+C,
\qquad
A=A_t\,{\rm d}t+A_x\,{\rm d}x,
\qquad
C=C_t\,{\rm d}t+C_x\,{\rm d}x,
\]
where $A$ and $C$ take values in the linear- and quadratic-jet sectors,
respectively.  For constants $\kappa_1,\kappa_2>0$, consider
\begin{equation}
\ell_{1+1}(A,C)
=\frac{\kappa_1}{2}\big(-A_t^2+A_x^2\big)
+\frac{\kappa_2}{2}\big(-C_t^2+C_x^2\big).
\label{eq:jet1-lagrangian}
\end{equation}
Equations \eqref{eq:jet2-EP-linear}-
\eqref{eq:jet2-reconstruction-quadratic} reduce to
\begin{align}
\partial_tA_x-\partial_xA_t&=0,
\label{eq:jet1-MC-A}\\
\kappa_1\big(\partial_tA_t-\partial_xA_x\big)
&=-\kappa_2\big(C_t^2-C_x^2\big),
\label{eq:jet1-EP-A}\\
\partial_tC_x-\partial_xC_t
&=-A_tC_x+A_xC_t,
\label{eq:jet1-MC-C}\\
\partial_tC_t-\partial_xC_x
&=A_tC_t-A_xC_x.
\label{eq:jet1-EP-C}
\end{align}
The principal part consists of two first-order wave systems with
characteristic speeds $\pm1$; in particular,
\eqref{eq:jet1-MC-A}-\eqref{eq:jet1-EP-C} is a semilinear
symmetric-hyperbolic system.  Notice especially the sign on the right-hand
side of \eqref{eq:jet1-EP-C}, which is fixed by
$a\heartsuit\Pi=-a\Pi$ and by the Lorentzian raising of the spacetime index.

On the zero-integration-constant branch specified below, the
traveling-wave system admits two first integrals.

\begin{proposition}
\label{prop:jet1-exact-conversion-wave}
Let $|c|<1$ and set $z=x-ct$.  On the zero-integration-constant branch
$A_t=-cA_x$, define
\begin{equation}
a(z):=A_x(z),
\qquad
U(z):=C_t(z)+cC_x(z),
\qquad
V(z):=cC_t(z)+C_x(z),
\qquad
K:=\frac{\kappa_2}{\kappa_1(1-c^2)^2}.
\label{eq:jet1-traveling-variables}
\end{equation}
Then \eqref{eq:jet1-MC-A}-\eqref{eq:jet1-EP-C} is equivalent to
\begin{equation}
a'=K(U^2-V^2),
\qquad
U'=-aU,
\qquad
V'=aV.
\label{eq:jet1-traveling-system}
\end{equation}
This system has the first integrals
\begin{equation}
UV=q_*,
\qquad
a^2+K(U^2+V^2)=h,
\label{eq:jet1-traveling-integrals}
\end{equation}
where $q_*$ and $h$ are constants.  On the invariant branch $U=0$, for every
$k>0$ and $z_0\in\mathbb R$ it admits the exact solution
\begin{align}
A_x(z)&=-k\tanh\theta,
&A_t(z)&=ck\tanh\theta,
\label{eq:jet1-exact-A}\\
C_x(z)&=k\sqrt{\frac{\kappa_1}{\kappa_2}}\,\operatorname{sech}\theta,
&C_t(z)&=-ck\sqrt{\frac{\kappa_1}{\kappa_2}}\,
\operatorname{sech}\theta,
\label{eq:jet1-exact-C}
\end{align}
where $\theta=k(z-z_0)$.  It is the reduced jet of an explicit
$G_1^2$-valued traveling field.  On the component $g>0$, one such
reconstruction is
\begin{equation}
g(z)=g_0\operatorname{sech}\theta,
\qquad
B(z)=\sqrt{\frac{\kappa_1}{\kappa_2}}\sinh\theta,
\qquad
T(z)=g(z)^2B(z)
=g_0^2\sqrt{\frac{\kappa_1}{\kappa_2}}
\operatorname{sech}\theta\tanh\theta,
\label{eq:jet1-exact-reconstruction}
\end{equation}
where $g_0>0$ and $B=T\cdot g^{-1}=Tg^{-2}$.
\end{proposition}

\begin{proof}
Equation \eqref{eq:jet1-MC-A} gives
$(A_t+cA_x)'=0$, and we take the constant to vanish.  Substitution into
\eqref{eq:jet1-MC-C} and \eqref{eq:jet1-EP-C} gives, respectively,
$U'=-aU$ and $V'=aV$.  Moreover,
\[
C_t^2-C_x^2=\frac{U^2-V^2}{1-c^2},
\]
so \eqref{eq:jet1-EP-A} yields the first equation in
\eqref{eq:jet1-traveling-system}.  Direct differentiation gives
$(UV)'=0$ and
\[
\frac{{\rm d}}{{\rm d}z}
\big(a^2+K(U^2+V^2)\big)=0,
\]
which proves \eqref{eq:jet1-traveling-integrals}.  On $U=0$, equations
\eqref{eq:jet1-exact-A}-\eqref{eq:jet1-exact-C} correspond to
$a=-k\tanh\theta$ and
$V=kK^{-1/2}\operatorname{sech}\theta$, and hence solve
\eqref{eq:jet1-traveling-system}.  Finally, $A_x=g'g^{-1}$, and the reconstruction formula for $C_\mu$
in Subsection~\ref{sec5.1}, together with $\rho_*(a)B=-aB$, gives
\[
C_x=B'+A_xB.
\]
The functions in \eqref{eq:jet1-exact-reconstruction} give precisely
\eqref{eq:jet1-exact-A}-\eqref{eq:jet1-exact-C}; their $t$-components are
$-c$ times their $x$-components because the reconstructed field depends only
on $z$.
\end{proof}

The solution also makes the sector-to-sector exchange quantitative.
We use the stress-energy convention
$T_{\mu\nu}=2\,\partial\ell_{1+1}/\partial\gamma^{\mu\nu}
-\gamma_{\mu\nu}\ell_{1+1}$.
This is the tensor obtained by lowering the first index of the
energy-momentum tensor in Subsection~\ref{sec:noether}.
Its linear-jet and quadratic-jet contributions are
\begin{equation}
\mathcal T^A_{\mu\nu}
=\kappa_1\left(A_\mu A_\nu
-\frac12\gamma_{\mu\nu}A_\alpha A^\alpha\right),
\qquad
\mathcal T^C_{\mu\nu}
=\kappa_2\left(C_\mu C_\nu
-\frac12\gamma_{\mu\nu}C_\alpha C^\alpha\right).
\label{eq:jet1-split-stress-energy}
\end{equation}
On every solution one has
\begin{equation}
\partial^\mu\mathcal T^A_{\mu\nu}=\mathcal Q_\nu,
\qquad
\partial^\mu\mathcal T^C_{\mu\nu}=-\mathcal Q_\nu,
\qquad
\mathcal Q_\nu
=\kappa_2\big(C_t^2-C_x^2\big)A_\nu.
\label{eq:jet1-sector-exchange}
\end{equation}
Indeed, equations \eqref{eq:jet1-MC-A}-\eqref{eq:jet1-EP-C}
are equivalently
\[
\partial_\mu A^\mu
=\frac{\kappa_2}{\kappa_1}(C_t^2-C_x^2),
\qquad
\partial_\mu C^\mu=A_\mu C^\mu,
\qquad
\partial_\mu C_\nu-\partial_\nu C_\mu
=-A_\mu C_\nu+A_\nu C_\mu,
\]
and direct differentiation of \eqref{eq:jet1-split-stress-energy} gives
\eqref{eq:jet1-sector-exchange}.
Thus the two sectors generally exchange energy and momentum,
while their sum satisfies a conservation law.  For
\eqref{eq:jet1-exact-A}-\eqref{eq:jet1-exact-C}, the sector energy densities
and off-diagonal stresses are
\begin{align}
e_A&=\frac{\kappa_1k^2}{2}(1+c^2)\tanh^2\theta,
&e_C&=\frac{\kappa_1k^2}{2}(1+c^2)\operatorname{sech}^2\theta,
\label{eq:jet1-sector-energies}\\
(\mathcal T^A)_{tx}
&=-c\kappa_1k^2\tanh^2\theta,
&(\mathcal T^C)_{tx}
&=-c\kappa_1k^2\operatorname{sech}^2\theta.
\label{eq:jet1-sector-momenta}
\end{align}
Consequently,
\begin{equation}
e_A+e_C=\frac{\kappa_1k^2}{2}(1+c^2),
\qquad
\mathcal T_{tx}=-c\kappa_1k^2,
\label{eq:jet1-total-stress-energy}
\end{equation}
pointwise along the wave.  As $|z-z_0|\to\infty$, the quadratic-jet energy tends to zero
and the linear-jet energy approaches the total energy density,
whereas at $z=z_0$ the energy lies entirely in the quadratic-jet
sector. The profile is therefore an exact traveling conversion front: a
heteroclinic $\tanh$ profile in the linear jet is accompanied by a localized
$\operatorname{sech}$ pulse in the quadratic jet, while total energy and
momentum remain conserved.  We deliberately do not call it a soliton, since
no orbital-stability or scattering property is asserted.  Moreover, the
solution has nonzero asymptotic background energy and is therefore not a
finite-energy excitation of the vacuum on the whole line; its content is the
localized redistribution of a fixed total stress-energy between the two jet
sectors.

The sectorial decomposition is tied to the chosen linear-jet and
quadratic-jet variables and to the constitutive quadratic forms.
For $c\ne0$, a fixed spatial point experiences the reversible
sequence of energy redistribution from the linear sector to the
quadratic sector and back as the profile passes.
Both spatial asymptotic states have the same sectorial energies;
there is no permanent net conversion left behind the front.
For $c=0$, the same redistribution is stationary.

\begin{remark}
\label{rem:jet1-hyperbolic-plane}
On the component $g>0$, introduce $\phi:=\log g$ and $q:=gB=\frac{T}{g}$. Then $A=\mathrm d\phi$ and
$C=e^{-\phi}\mathrm dq$, so the unreduced density corresponding to
\eqref{eq:jet1-lagrangian} is the wave-map density for the target metric
\begin{equation}
h=\kappa_1\,{\rm d}\phi^2
+\kappa_2 e^{-2\phi}\,{\rm d}q^2,
\label{eq:jet1-hyperbolic-metric}
\end{equation}
which has constant Gaussian curvature $-1/\kappa_1$.  For
\eqref{eq:jet1-exact-reconstruction}, the rescaled coordinate
$\widetilde q:=\sqrt{\kappa_2/\kappa_1}\,q$ satisfies
\[
e^{2\phi}+\widetilde q^{\,2}=g_0^2.
\]
Thus the reconstructed profile traces a hyperbolic geodesic semicircle.  The
traveling solution is therefore a geodesic wave in intrinsic target
coordinates; the centered splitting resolves its constant total
stress-energy into the complementary linear- and quadratic-jet contributions
displayed in \eqref{eq:jet1-sector-energies}-
\eqref{eq:jet1-total-stress-energy}.
\end{remark}\medskip

The internal Noether currents can also be written explicitly.
In the coordinates $(\phi,q)$, the right action of $(a,b)$
in the identity component of the jet group is
$(\phi,q)\mapsto(\phi+\log a,\,aq+b/a)$, with $a>0$.
Its infinitesimal generators may be chosen as
$K_1=\partial_q$ and $K_2=\partial_\phi+q\partial_q$.
The construction in Subsection~\ref{sec:noether} therefore gives
\begin{equation}
\begin{aligned}
J_1^\mu&=\kappa_2 e^{-2\phi}\partial^\mu q,\\
J_2^\mu&=\kappa_1\partial^\mu\phi
+\kappa_2q e^{-2\phi}\partial^\mu q,
\qquad \partial_\mu J_a^\mu=0,\quad a=1,2.
\end{aligned}
\label{eq:jet-internal-noether-currents}
\end{equation}
Here $\partial^\mu=\gamma^{\mu\nu}\partial_\nu$.
Conservation of the corresponding integrated quantities on a finite
spatial interval includes the boundary flux; it must not be inferred
without checking the boundary conditions.

The reduced multisymplectic current is obtained from
\eqref{eq:centered-MS-current} by the substitutions
\[
(\mathsf P_{\mathfrak g},\mathsf P_V)=(P_A,P_C),
\qquad
(\eta^{\mathfrak g},\eta^V)=(\Sigma,Z).
\]
The terms specific to the quadratic-jet sector are
\begin{equation}
-\langle\Sigma_{\mathsf U}\heartsuit P_C^\mu,
Z_{\mathsf V}\rangle
-\langle Z_{\mathsf U}\diamondsuit P_C^\mu,
\Sigma_{\mathsf V}\rangle,
\label{eq:jet2-multisymplectic-mixed-terms}
\end{equation}
whose index structure is given by
\eqref{eq:jet2-heart-indices}-\eqref{eq:jet2-diamond-indices}.  Hence the
two ways in which a linear jet acts on a quadratic jet--on its output and on
its two inputs--appear simultaneously in both the field equations and the
conserved multisymplectic flux.  For first variations satisfying the linearized equations,
the current $\omega_{\mathrm{jet}}
=\omega_{\mathrm{jet}}^\mu d^n x_\mu$ is closed.
On the Minkowski base considered here, this gives the local
conservation law $\partial_\mu\omega_{\mathrm{jet}}^\mu=0$.

For $r=1$, the complete current can be written explicitly.
For two first variations $U$ and $W$, with reduced parameters
$(\Sigma_U,Z_U)$ and $(\Sigma_W,Z_W)$, it is
\begin{equation*}
\omega_{\mathrm{jet}}^\mu(U,W)
={}\kappa_1\bigl(
(\delta_W A^\mu)\Sigma_U
-(\delta_U A^\mu)\Sigma_W
\bigr)+\kappa_2\bigl(
(\delta_W C^\mu)Z_U
-(\delta_U C^\mu)Z_W
\bigr)+\kappa_2 C^\mu
\bigl(\Sigma_U Z_W-Z_U\Sigma_W\bigr).
\end{equation*}
Here the reduced variations satisfy
$\delta_U A_\mu=\partial_\mu\Sigma_U$ and
$\delta_U C_\mu=\partial_\mu Z_U+A_\mu Z_U-C_\mu\Sigma_U$,
with analogous formulas for $W$.
This expression follows from the reduced multisymplectic
current in Subsection~\ref{sec:multisymplecticformula}.
Its conservation concerns pairs of solutions of the
linearized equations and is distinct from the Noether
and stress-energy conservation laws.

\paragraph{Relation with geodesic wave maps.}
The wave-map interpretation places the reconstructed front
within the geometric framework of \cite{ShatahStruwe1998}.
Related stability results for totally geodesic wave maps
were obtained by Abbrescia and Chen~\cite{AbbresciaChen2021} under different
dimensional and structural assumptions; those results
are not invoked here.
In the present example, the wave-map description provides
coordinates adapted to the reconstructed front, whereas
the centered reduction identifies its linear- and
quadratic-jet contributions and their exchange.
These adapted coordinates allow us to describe the
linearized perturbations and prescribe compatible
initial data for the reduced evolution considered below.

Let $(s,r)$ be Fermi coordinates about the complete geodesic
traced by the reconstructed front. In the remainder of this subsection, $s$ denotes the longitudinal
Fermi coordinate, rather than the group-valued section. With
$\widetilde q=\sqrt{\kappa_2/\kappa_1}\,q$, the metric
\eqref{eq:jet1-hyperbolic-metric} takes the form
\begin{equation}
h=\kappa_1\bigl({\rm d}r^2+\cosh^2r\,{\rm d}s^2\bigr),
\label{eq:jet1-fermi-metric}
\end{equation}
and the exact conversion front is
\begin{equation}
s_0(t,x)=k(x-ct-z_0),
\qquad
r_0(t,x)=0.
\label{eq:jet1-fermi-background}
\end{equation}

A concrete choice of the normal orientation in these Fermi coordinates
is obtained by setting
$\mathcal D=\cosh r\cosh s-\sinh r$ and
\begin{equation}
e^\phi=\frac{g_0}{\mathcal D},\qquad
\widetilde q=\frac{g_0\cosh r\sinh s}{\mathcal D},\qquad
q=\sqrt{\frac{\kappa_1}{\kappa_2}}\,\widetilde q.
\label{eq:fermi-to-jet-coordinates}
\end{equation}
For $r=0$, these expressions give
$e^\phi=g_0\operatorname{sech}s$ and
$\widetilde q=g_0\tanh s$, consistently with the exact reconstruction.
They transform the upper-half-plane metric
$\kappa_1\bigl((d\widetilde q)^2+(d(e^\phi))^2\bigr)/e^{2\phi}$
into the Fermi metric above.
The corresponding group variables are $g=e^\phi$ and $T=gq$,
and their reduced fields are
$A_\mu=\partial_\mu\phi$ and
$C_\mu=\sqrt{\kappa_1/\kappa_2}\,e^{-\phi}
\partial_\mu\widetilde q$.
The opposite orientation of the normal coordinate is obtained by
replacing $r$ with $-r$.

The full wave-map equations in these coordinates are
\begin{align}
r_{tt}-r_{xx}
+\sinh r\cosh r\bigl(s_x^2-s_t^2\bigr)&=0,
\label{eq:jet1-fermi-r}\\
\partial_t\bigl(\cosh^2r\,s_t\bigr)
-\partial_x\bigl(\cosh^2r\,s_x\bigr)&=0.
\label{eq:jet1-fermi-s}
\end{align}

Let $|c|<1$ and $k>0$. Linearize
\eqref{eq:jet1-fermi-r}-\eqref{eq:jet1-fermi-s} about
\eqref{eq:jet1-fermi-background} by writing $s=s_0+\varepsilon\sigma$ and 
$r=\varepsilon\rho$. Then the tangential perturbation $\sigma$ and the normal
perturbation $\rho$ satisfy, respectively, the wave and
Klein-Gordon equations
\begin{equation}
\sigma_{tt}-\sigma_{xx}=0,
\qquad
\rho_{tt}-\rho_{xx}+m^2\rho=0,
\qquad
m^2:=k^2(1-c^2)>0.
\label{eq:jet1-linearized-wave-KG}
\end{equation}

The associated Cauchy problem consists of solving these
equations on $\mathbb R_t\times\mathbb R_x$ with prescribed
initial displacements and velocities
\[
\sigma(0,x)=\sigma_0(x),\qquad
\sigma_t(0,x)=\sigma_1(x),\qquad
\rho(0,x)=\rho_0(x),\qquad
\rho_t(0,x)=\rho_1(x).
\]
Its natural energy space is $\mathcal E
=
\bigl(\dot H^1(\mathbb R)\times L^2(\mathbb R)\bigr)
\times
\bigl(H^1(\mathbb R)\times L^2(\mathbb R)\bigr)$. Here
$\dot H^1(\mathbb R)
=
\{u\in H^1_{\mathrm{loc}}(\mathbb R):
\partial_xu\in L^2(\mathbb R)\}/\mathbb R$ is the homogeneous Sobolev space, in which functions differing
by an additive constant are identified, with norm
$\|[u]\|_{\dot H^1}=\|\partial_xu\|_{L^2}$.

\begin{proposition}
\label{prop:jet1-linear-stability} For every initial state
$(\sigma_0,\sigma_1,\rho_0,\rho_1)\in\mathcal E$,
the Cauchy problem has a unique global energy solution,
with $\sigma$ understood modulo constants, and the solution
depends continuously on the initial state. Moreover,
\begin{equation}
E_{\rm lin}(t)
=\frac{\kappa_1}{2}\int_{\mathbb R}
\left(
\sigma_t^2+\sigma_x^2
+\rho_t^2+\rho_x^2+m^2\rho^2
\right){\rm d}x
=E_{\rm lin}(0).
\label{eq:jet1-linearized-energy}
\end{equation}
Thus the linearized flow preserves the norm
$\|(\sigma,\sigma_t,\rho,\rho_t)\|_{\mathcal E,m}
=(2E_{\rm lin}/\kappa_1)^{1/2}$.
In particular, it controls the tangential derivatives and,
since $m>0$, the full $H^1\times L^2$ norm of the normal
perturbation.

The conversion front has no exponentially growing linear
Fourier modes\footnote{These conclusions concern the linearized
equations and do not assert nonlinear orbital or
asymptotic stability.}.
\end{proposition}

\begin{proof}
The wave-map density associated with
\eqref{eq:jet1-fermi-metric} is
\[
\frac{\kappa_1}{2}
\left[-r_t^2+r_x^2
+\cosh^2r\bigl(-s_t^2+s_x^2\bigr)\right],
\]
whose Euler--Lagrange equations are
\eqref{eq:jet1-fermi-r}-\eqref{eq:jet1-fermi-s}.

For the background solution,
$s_{0,x}=k$ and $s_{0,t}=-kc$, so
$s_{0,x}^2-s_{0,t}^2=k^2(1-c^2)=m^2$.
Using
\[
\sinh(\varepsilon\rho)\cosh(\varepsilon\rho)
=\varepsilon\rho+O(\varepsilon^3),
\qquad
\cosh^2(\varepsilon\rho)=1+O(\varepsilon^2),
\]
the normal equation becomes
\[
\varepsilon\bigl(\rho_{tt}-\rho_{xx}+m^2\rho\bigr)
+O(\varepsilon^2)=0.
\]
In the tangential equation, the factor $\cosh^2r$
has no term of first order in $\varepsilon$.
Since $s_0$ solves the wave equation, its linearization is
$\sigma_{tt}-\sigma_{xx}=0$.
This proves \eqref{eq:jet1-linearized-wave-KG}.

For smooth compactly supported initial data, multiply the
two linearized equations by $\kappa_1\sigma_t$ and
$\kappa_1\rho_t$, respectively, and add them. This gives
the local energy identity
\[
\partial_t\left[
\frac{\kappa_1}{2}
\left(\sigma_t^2+\sigma_x^2
+\rho_t^2+\rho_x^2+m^2\rho^2\right)
\right]
=
\kappa_1\partial_x
\left(\sigma_t\sigma_x+\rho_t\rho_x\right).
\]
Finite propagation speed ensures that the boundary terms
vanish upon integration over $\mathbb R$, yielding
\eqref{eq:jet1-linearized-energy}.

To describe the Fourier modes explicitly, consider
complex plane waves
\[
\sigma(t,x)=a\,e^{i(\xi x-\omega t)},
\qquad\text{or}\qquad
\rho(t,x)=b\,e^{i(\xi x-\omega t)},
\]
where $\xi\in\mathbb R$ is the spatial wave number.
Substitution into the respective equations gives
\[
\omega^2=\xi^2
\quad\text{and}\quad
\omega^2=\xi^2+m^2.
\]
Both dispersion relations have real frequencies, so
no such mode grows exponentially in time.
At $\xi=0$, the tangential equation admits constant
and affine functions of time; this is a zero-frequency
mode, not an exponential instability. Spatially constant
nonzero velocities are excluded by the $L^2(\mathbb R)$
condition on $\sigma_t$, and constant displacements are
identified in $\dot H^1(\mathbb R)$.

Finally, the Fourier solution formulas for the wave and
Klein-Gordon equations give global solutions for smooth
data. The conserved energy estimate extends these
solutions by density to all initial states in $\mathcal E$,
and gives uniqueness and continuous dependence in the
energy norm.
\end{proof}

This linearized description provides a reference for interpreting
perturbations of the reconstructed front. In particular, initially
vanishing tangential data remain zero at linear order, so any
tangential response generated from purely normal initial data
reflects nonlinear coupling. In the experiment below, Fermi
coordinates are used to specify the initial field and its velocity;
the subsequent evolution is computed in the reduced variables,
from which the group-valued field is reconstructed.

\paragraph{A nonlinear perturbation experiment.}
We now illustrate the reduction--reconstruction procedure
beyond the exact traveling-wave branch. We evolve the reduced
jet variables directly and recover the group-valued field by
integrating the reconstruction equations. Fermi coordinates
are used to prescribe compatible initial data and to compute
an independent wave-map solution for comparison.

Writing $s=s_0+\sigma$, we prescribe
\begin{equation}
r(0,x)=\varepsilon\exp\!\left[-(x/w)^2\right],
\qquad
r_t(0,x)=\sigma(0,x)=\sigma_t(0,x)=0.
\label{eq:jet1-numerical-data}
\end{equation}
The map \eqref{eq:fermi-to-jet-coordinates} determines the
initial group variables and their velocities. Their reduced
jets are calculated analytically from
$A_\mu=\partial_\mu\phi$ and
$C_\mu=e^{-\phi}\partial_\mu q$.
Thus the initial reduced data arise from a prescribed field
and its velocity, rather than from independent perturbations
of the four jet components.

The numerical evolution uses the closed first-order system
\begin{equation}
\begin{aligned}
\partial_t A_t
 &=\partial_x A_x
   -\frac{\kappa_2}{\kappa_1}(C_t^2-C_x^2),\\
\partial_t A_x
 &=\partial_x A_t,\\
\partial_t C_t
 &=\partial_x C_x+A_tC_t-A_xC_x,\\
\partial_t C_x
 &=\partial_x C_t-A_tC_x+A_xC_t.
\end{aligned}
\label{eq:jet1-numerical-reduced-system}
\end{equation}
The first and third equations are the reduced dynamical
equations; the second and fourth are the reconstruction
compatibility conditions. In $1+1$ dimensions, these conditions
supply the evolution equations for the spatial jet components.

From this reduced evolution, we recover the field by integrating
\begin{equation}
\partial_t\phi=A_t,
\qquad
\partial_t q=e^\phi C_t,
\qquad
g=e^\phi,\quad T=gq,
\label{eq:jet1-numerical-reconstruction}
\end{equation}
with the prescribed initial values.
These reconstruction variables do not feed back into
\eqref{eq:jet1-numerical-reduced-system}.
To check compatibility numerically, we also reconstruct at
each time by integrating
$\partial_x\phi=A_x$ and $\partial_xq=e^\phi C_x$,
starting from the temporally reconstructed value at $x=0$.
The two procedures correspond to different integration paths
from the same initial reference point. Their agreement in the
continuum follows from the compatibility equations; their
difference at finite resolution measures the combined
evolution and reconstruction errors.

We use $k=\kappa_1=\kappa_2=g_0=1$, $c=1/2$,
$\varepsilon=0.08$, and $w=2$ on $[-60,60]$ up to $t=18$.
Fourth-order centered spatial differences and classical
fourth-order Runge--Kutta time stepping are applied with
$N=2400$ cells, hence $2401$ nodes,
$\Delta x=0.05$, and $\Delta t=0.01$.
The complete reduced fields are not treated as periodic:
ghost values at both ends are supplied by the exact
unperturbed front at each Runge--Kutta stage.
Repeating the computation on $[-80,80]$ with the same mesh
spacing changes the final reduced and reconstructed
variables by less than $4\times10^{-14}$ on $[-25,25]$.
Spatial reconstruction uses cubic-spline quadrature.

Figure~\ref{fig:jet1-conversion-front-simulation} shows the
reduced jet profiles, the reconstructed group variables,
and their sector energy densities.
The perturbed evolution produces additional spatial structure
in the quadratic-jet component while retaining a propagating
front over the reported time interval.
The reconstructed variables provide the corresponding
group-valued realization of this reduced dynamics.

\begin{figure}[htbp]
\centering
\includegraphics[width=\textwidth]{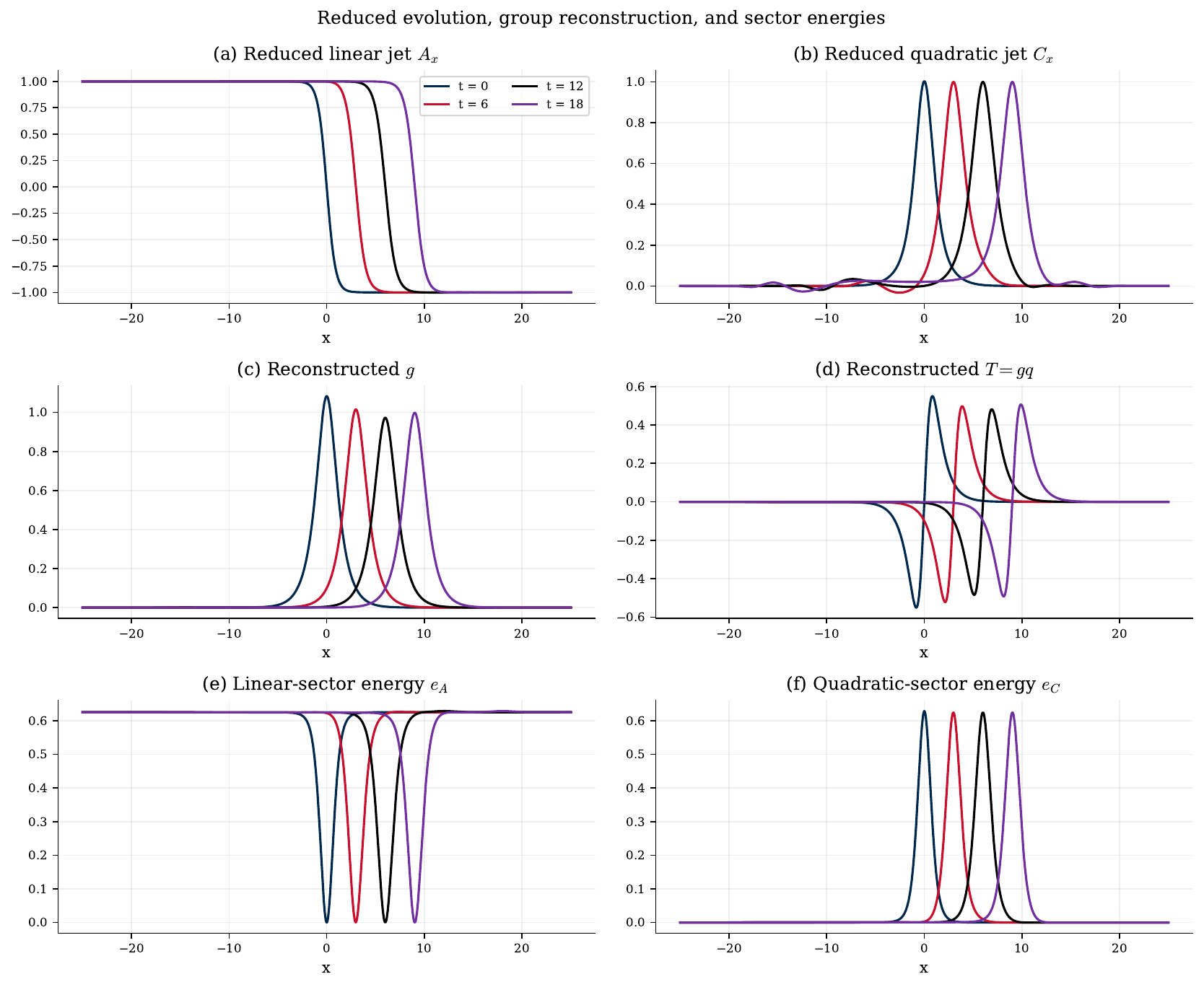}
\caption{Reduced evolution and reconstruction of a perturbed
$G_1^2$ conversion front at $t=0,6,12,18$.
(a,b) Spatial components $A_x$ and $C_x$, obtained by evolving
\eqref{eq:jet1-numerical-reduced-system}.
(c,d) Group variables $g$ and $T$, obtained through
\eqref{eq:jet1-numerical-reconstruction}.
(e,f) Linear- and quadratic-jet energy densities
$e_A=\kappa_1(A_t^2+A_x^2)/2$ and
$e_C=\kappa_2(C_t^2+C_x^2)/2$.
All panels describe the same nonlinear solution;
colors distinguish the four displayed times.
The parameters and numerical resolution are given in the text.}
\label{fig:jet1-conversion-front-simulation}
\end{figure}

The total energy and momentum densities are evaluated
directly in reduced variables:
\begin{equation}
e=\frac{\kappa_1}{2}(A_t^2+A_x^2)
  +\frac{\kappa_2}{2}(C_t^2+C_x^2),
\qquad
p=\kappa_1A_tA_x+\kappa_2C_tC_x.
\label{eq:jet1-numerical-invariants}
\end{equation}
Their constant background values are
$e_0=\kappa_1k^2(1+c^2)/2$ and $p_0=-\kappa_1ck^2$.
Over the $181$ stored times, the maximum relative drifts
of the integrals of $e-e_0$ and $p-p_0$ are
$1.24\times10^{-9}$ and $1.44\times10^{-9}$, respectively,
normalized by the absolute initial value of each integral.
The sector balances distinguish exchange from transport:
\[
\partial_t e_A-\partial_x p_A=S_A,
\qquad
\partial_t e_C-\partial_x p_C=-S_A,
\]
where
$p_A=\kappa_1A_tA_x$,
$p_C=\kappa_2C_tC_x$, and
$S_A=-\kappa_2A_t(C_t^2-C_x^2)$.
Thus changes in each sector include both boundary flux
and internal exchange, whereas the exchange terms cancel
in the total energy balance.
These jet sectors are distinct from the normal and tangential
directions of the Fermi parametrization.

We assess numerical consistency on $J=[-25,25]$ using
$N=600,1200,2400$ and $\Delta t=0.2\Delta x$.
For the full reduced state $(A_t,A_x,C_t,C_x)$, successive-grid
differences at $t=18$, evaluated in the discrete $L^2(J)$
norm on common nodes, are $1.823\times10^{-3}$ and
$1.157\times10^{-4}$, giving an observed order of $3.978$.
The final $L^2(J)$ differences between the two reconstructions
of $(\phi,q)$ are, respectively,
$1.747\times10^{-3}$, $1.110\times10^{-4}$, and
$6.982\times10^{-6}$.
Their approximately sixteenfold decrease under mesh
halving is consistent with fourth-order accuracy.

As an independent comparison, we solve the wave-map equations
in Fermi coordinates with the same initial data and transform
that solution to $(\phi,q)$.
Its final $L^2(J)$ difference from the temporally reconstructed
field is $6.45\times10^{-6}$ on the finest grid, again
decreasing approximately with fourth order under refinement.
These are coordinate norms of $(\phi,q)$, rather than
intrinsic target distances.
Together, the comparisons support the consistency of the
reduced evolution with the reconstructed field and its
wave-map description.


\section{Conclusions and Future Work}\label{sec:conclusions}

We have formulated the covariant reduction of first-order field theories
in centered semidirect-product variables, making explicit the reduced
equations, reconstruction conditions, and covariant currents associated
with the two-sided actions. Starting from the standard reduction
of first order field theories on principal bundles, we have shown that the commuting
left and right actions defining the centered product give rise, in split
variables, to the heart and diamond couplings introduced in
finite-dimensional centered mechanics.  The resulting formulation
includes the covariant
Euler-Poincar\'e field equations together with the reconstruction condition, the reduced
multisymplectic form formula, the associated Noether and
energy-momentum balance laws, and the covariant Lie-Poisson
description. For hyperregular theories, the Lagrangian and Hamiltonian
formulations are related by the covariant Legendre transformation.

The applications also delimit the scope of the centered construction.
For $\mathrm{SO}(3)$, the standard director representation on $\mathbb{R}^3$
does not admit two independent nontrivial commuting actions.  Passing
instead to $\mathrm{SO}(3)\Join\operatorname{Mat}(3)$ yields an effective
conjugation representation on the semidirect variable
$B=AR^{-1}$, together with the invariant decomposition
$\operatorname{Mat}(3)
=
\mathbb{R}I
\oplus
\mathfrak{so}(3)
\oplus
\operatorname{Sym}_0(3)$.
Its three components provide scalar, vector, and alignment-tensor
channels and hence a precise point of contact with microstructural and
liquid-crystal variables. Equivariant linear constitutive maps preserve these three channels;
for self-adjoint equivariant choices, the contracted internal torque
vanishes. Nonzero torque and off-diagonal constitutive couplings
therefore require the corresponding additional constitutive structure. This comparison also shows what the centered
theory does not supply by itself: unit or projective director
constraints, $Q$-tensor bulk potentials, and the broken-symmetry data
of standard nematodynamics require the corresponding advected-parameter
or affine-reduction frameworks.

The last application returns to the geometric origin of centered
semidirect products: the second-order jet group
$G_r^2
\simeq
GL(r)\Join
\bigl(S^2(\mathbb{R}^r)^*\otimes\mathbb{R}^r\bigr)$
inherits its two actions directly from the two contributions in the
second-order chain rule.  Using $G_r^2$ as an internal structure group
therefore provides a first-order example of the theory:
the same output and input actions that define jet composition appear,
after reduction, in the heart and diamond terms of the reduced field equations,
the Lie-Poisson structure, and the multisymplectic current.  This
construction must be distinguished from both higher-order variational
theories and gauge theories of independent connections; the
second-order character belongs to the internal symmetry group rather
than to the variational order.

For $G_1^2$ on a $(1+1)$-dimensional Minkowski base,
the reduced equations admit first integrals on the specified
traveling-wave branch and an explicit conversion front.
Its reconstruction identifies a geodesic wave in a hyperbolic
target, while the centered splitting resolves the constant
total stress-energy into complementary linear- and
quadratic-jet contributions.
The linearized dynamics separate into a tangential wave
equation and a normal Klein-Gordon equation with positive
squared mass, yielding a conserved positive energy in the
stated perturbation space.

The nonlinear experiment implements the same
reduction-reconstruction procedure beyond the exact
traveling-wave branch: it evolves the reduced dynamical
and compatibility equations and recovers the group-valued
field by integration.
Mesh refinement, comparison of two reconstruction paths,
and agreement with an independent wave-map computation
support the numerical consistency of this correspondence.
These finite-time results do not establish nonlinear
stability or exact discrete multisymplectic preservation.

\medskip

Several directions remain open:

\begin{itemize}

\item
\emph{Structure-preserving discretization.}
The reduced multisymplectic form formula provides a natural starting
point for discrete covariant variational principles adapted to centered
semidirect-product symmetry.  It would be particularly interesting to
determine whether discrete heart and diamond couplings can be retained
exactly while preserving the corresponding discrete momentum and
multisymplectic balance laws. The discrete variational and multisymplectic constructions provide
natural reference points for this extension
\cite{MaPeShWe2001, MaWe2001}.

\item
\emph{Higher jet-group internal symmetries.}
The appearance of $G_r^2$ suggests studying the algebraic structure of
higher jet groups $G_r^k$ and determining when their successive
chain-rule couplings admit centered, iterated, or more general
semidirect-product descriptions (cf.~\cite{CastrillonRodriguezHigher2023}).  Such an extension should be kept
distinct from higher-order field theories on $J^kP$: relating an
internal jet-group field to spatially holonomic higher derivatives of a
deformation requires additional contact or Spencer-type compatibility
conditions.

\item
\emph{Lie group-bundle and gauge extensions.}
The present reduction concerns the action of a fixed structure group on
a principal bundle.  A gauge theory with $x$-dependent
fiberwise symmetries requires instead the Lie group-bundle framework and
generalized principal connections.  Extending centered
semidirect-product reduction, its Lie-Poisson formulation, and its
multisymplectic currents to that setting would provide the appropriate
route toward gauge theories with two-sided internal symmetry.

\end{itemize}

\subsection*{Acknowledgments}

L. Colombo would like to thank Darryl D. Holm for introducing him and Henry O. Jacobs to the idea of centered semidirect products and for encouraging its further development, despite the years that have passed since its first appearance, as well as for highlighting its potential applications in geometric mechanics and mathematical physics. He also thanks Henry O. Jacobs and Joris Vankerschaver for many valuable discussions on centered semidirect products and their possible applications.

\subsection*{Funding}

The authors acknowledge financial support from grant PID2022-137909NB-C21 funded by MCIN/AEI/ 10.13039/501100011033. \'A. Rodr\'iguez Abella and M. \'A. Berbel were partially supported by grant PID2024-156578NB-I00 funded by MICIU/AEI/10.13039/501100011033/FEDER, EU.

\subsection*{Use of artificial intelligence}
Generative AI tools were used to review and improve the writing of the manuscript and to assist with code for the numerical experiments. The authors take full responsibility for the content of the manuscript, including all mathematical statements, proofs, and numerical results.


\begin{thebibliography}{99}

\bibitem{AbbresciaChen2021}
L.~E. Abbrescia and Y.~Chen.
\newblock Global stability of some totally geodesic wave maps.
\newblock {\em Journal of Differential Equations},
284:219--252, 2021.


\bibitem{BerbelCastrillon2023}
M.~\'A. Berbel and M.~Castrill\'on L\'opez.
\newblock Poisson-Poincar\'e reduction for field theories.
\newblock {\em Journal of Geometry and Physics},
  191:104879, 2023.

  \bibitem{BerbelCastrillon2024}
M.~\'A. Berbel and M.~Castrill\'on L\'opez.
\newblock Lagrangian reduction by stages in field theory.
\newblock {\em Annales Henri Poincar\'e},
  25:1511-1543, 2024.

  \bibitem{BerbelCastrillonAffine2025}
M.~\'A. Berbel and M.~Castrill\'on L\'opez.
\newblock Hamiltonian reduction in affine principal bundles.
\newblock {\em Geometric Mechanics}, 2(4):519-533, 2025.

\bibitem{BerbelCastrillonSubgroup2026}
M.~\'A. Berbel and M.~Castrill\'on L\'opez.
\newblock Lie--Poisson reduction in principal bundles by a subgroup
  of the structure group.
\newblock arXiv:2604.08288, 2026.



\bibitem{Capriz1989}
G.~Capriz.
\newblock {\em Continua with Microstructure}.
\newblock Springer Tracts in Natural Philosophy, vol.~35.
\newblock Springer-Verlag, New York, 1989.

\bibitem{CaChGa2013}
M.~Castrill\'on L\'opez, P.~M. Chac\'on, and P.~L. Garc\'ia.
\newblock Lagrange--Poincar\'e reduction in affine principal bundles.
\newblock {\em Journal of Geometric Mechanics}, 5(4):399-414, 2013.



\bibitem{CaMa2008}
M.~Castrill\'on L\'opez and J.~E. Marsden.
\newblock Covariant and dynamical reduction for principal bundle field
  theories.
\newblock {\em Annals of Global Analysis and Geometry},
  34(3):263-285, 2008.

\bibitem{CaMa2003}
M.~Castrill\'on L\'opez and J.~E. Marsden.
\newblock Some remarks on Lagrangian and Poisson reduction for field theories.
\newblock {\em Journal of Geometry and Physics}, 48(1):52-83, 2003.

\bibitem{CaGaRa2001}
M.~Castrill\'on L\'opez, P.~L. Garc\'ia P\'erez, and T.~S. Ratiu.
\newblock Euler--Poincar\'e reduction on principal bundles.
\newblock {\em Letters in Mathematical Physics}, 58(2):167-180, 2001.

\bibitem{CaRa2003}
M.~Castrill\'on L\'opez and T.~S. Ratiu.
\newblock Reduction in principal bundles: Covariant Lagrange--Poincar\'e
  equations.
\newblock {\em Communications in Mathematical Physics},
  236(2):223-250, 2003.

\bibitem{CaRaSh2000}
M.~Castrill\'on L\'opez, T.~S. Ratiu, and S.~Shkoller.
\newblock Reduction in principal fiber bundles:
  Covariant Euler--Poincar\'e equations.
\newblock {\em Proceedings of the American Mathematical Society},
  128(7):2155-2164, 2000.

\bibitem{CastrillonRodriguezGauge2024}
M.~Castrill\'on L\'opez and \'A.~Rodr\'iguez Abella.
\newblock Gauge reduction in covariant field theory.
\newblock {\em Journal of Physics A: Mathematical and Theoretical},
  57(30):305204, 2024.

\bibitem{CastrillonRodriguezGeneralized2023}
M.~Castrill\'on L\'opez and \'A.~Rodr\'iguez Abella.
\newblock Principal bundles and connections modeled by Lie group bundles.
\newblock {\em  Geometriae Dedicata},
  217:32, 2023.

\bibitem{CastrillonRodriguezHigher2023}
M.~Castrill\'on L\'opez and \'A.~Rodr\'iguez Abella.
\newblock Higher Order Jet Bundles of Lie Group-Valued Functions.
\newblock {\em  Journal of Lie Theory},
  33(3):831-844, 2023.


\bibitem{CastrillonRodriguezHigherOrder2026}
M.~Castrill\'on L\'opez and \'A.~Rodr\'iguez Abella.
\newblock Higher-order Euler--Poincar\'e field equations.
\newblock {\em Journal of Geometry and Physics},
  219:105693, 2026.


\bibitem{CoJa2015}
L.~Colombo and H.~O. Jacobs.
\newblock Lagrangian mechanics on centered semidirect products.
\newblock In D.~E. Chang, D.~D. Holm, G.~W. Patrick, and T.~S. Ratiu,
  editors, {\em Geometry, Mechanics, and Dynamics: The Legacy of Jerry
  Marsden}, Fields Institute Communications, vol.~73, pages 167-184.
\newblock Springer, New York, 2015.

\bibitem{CotterElderingHolmJacobsMeier2016}
 C.~J. Cotter, J.~Eldering, D.~D. Holm, H.~O. Jacobs, and D.~M. Meier.
 \newblock Weak dual pairs and jetlet methods for ideal incompressible fluid
 models in $n\geq2$ dimensions.
 \newblock {\em Journal of Nonlinear Science}, 26(6):1723--1765, 2016.

\bibitem{CoHo2013}
C.~J.~Cotter and D.~D.~Holm.
\newblock On Noether's theorem for the Euler--Poincar\'e equation on the
diffeomorphism group with advected quantities.
\newblock {\em Foundations of Computational Mathematics},
13(4):457--477, 2013.

\bibitem{GayBalmazRatiu2010}
F.~Gay--Balmaz and T.~S. Ratiu.
\newblock A new Lagrangian dynamic reduction in field theory.
\newblock {\em Annales de l'Institut Fourier},
  60(3):1125-1160, 2010.

\bibitem{deGennesProst1993}
P.~G. de~Gennes and J.~Prost.
\newblock {\em The Physics of Liquid Crystals}.
\newblock Oxford University Press, Oxford, second edition, 1993.

\bibitem{DeLeonMartinDiegoSantamaria2004}
M.~de Le\'on, D.~Mart\'in de Diego, and A.~Santamar\'ia-Merino.
\newblock Symmetries in classical field theory.
\newblock {\em International Journal of Geometric Methods in Modern Physics},
  1(5):651-710, 2004.


\bibitem{EllisGayBalmazHolmRatiu2011}
D.~C.~P. Ellis, F.~Gay--Balmaz, D.~D. Holm, and T.~S. Ratiu.
\newblock Lagrange--Poincar\'e field equations.
\newblock {\em Journal of Geometry and Physics},
  61(11):2120-2146, 2011.


\bibitem{Eringen1999}
A.~C. Eringen.
\newblock {\em Microcontinuum Field Theories I: Foundations and Solids}.
\newblock Springer, New York, 1999.

\bibitem{FoRo2005}
M.~Forger and S.~V. Romero.
\newblock Covariant Poisson brackets in geometric field theory.
\newblock {\em Communications in Mathematical Physics},
  256(2):375-410, 2005.

\bibitem{GBRa2009}
F.~Gay--Balmaz and T.~S. Ratiu.
\newblock The geometric structure of complex fluids.
\newblock {\em Advances in Applied Mathematics}, 42(2):176-275, 2009.

\bibitem{GBRa2012}
F.~Gay--Balmaz, T.~S. Ratiu, and C.~Tronci.
\newblock Euler--Poincar\'e approaches to nematodynamics.
\newblock {\em Acta Applicandae Mathematicae}, 120(1):127-151, 2012.

\bibitem{GayBalmazRatiuTronci2013}
 F.~Gay--Balmaz, T.~S. Ratiu, and C.~Tronci.
 \newblock Equivalent theories of liquid crystal dynamics.
 \newblock {\em Archive for Rational Mechanics and Analysis}, 210(3):773--811, 2013.



\bibitem{GayBalmazTronci2010}
 F.~Gay--Balmaz and C.~Tronci.
 \newblock Reduction theory for symmetry breaking with applications to nematic systems.
 \newblock {\em Physica D: Nonlinear Phenomena},
 239(20--22):1929--1947, 2010.

\bibitem{GoIsMa1998}
M.~J. Gotay, J.~Isenberg, J.~E. Marsden, and R.~Montgomery.
\newblock Momentum maps and classical relativistic fields.
  Part~I: Covariant field theory.
\newblock {\em arXiv:physics/9801019}, 1998.


\bibitem{GoIsMaMo2004}
M.~J. Gotay, J.~Isenberg, and J.~E. Marsden.
\newblock Momentum maps and classical relativistic fields.
  Part~II: Canonical analysis of field theories.
\newblock {\em arXiv:math-ph/0411032}, 2004.

\bibitem{Holm2002ComplexFluids}
D.~D.~Holm,
Euler--Poincar\'e dynamics of perfect complex fluids,
in P.~Newton, P.~Holmes, and A.~Weinstein (eds.),
\emph{Geometry, Mechanics, and Dynamics},
Springer, New York, 2002, pp.~169-180.



\bibitem{Holm2025}
D.~D.~Holm,
\emph{Geometric Mechanics, Part III:
Broken Symmetry and Composition of Maps},
World Scientific Publishing, London, 2025.


\bibitem{HolmHuStreet2023}
D.~D.~Holm, R.~Hu, and O.~D.~Street,
Lagrangian reduction and wave mean flow interaction,
\emph{Physica D: Nonlinear Phenomena}
\textbf{454} (2023), 133847.


\bibitem{HoMaRa1998}
D.~D. Holm, J.~E. Marsden, and T.~S. Ratiu.
\newblock The Euler--Poincar\'e equations and semidirect products with
  applications to continuum theories.
\newblock {\em Advances in Mathematics}, 137(1):1-81, 1998.



\bibitem{KiTu1979}
J.~Kijowski and W.~M. Tulczyjew.
\newblock {\em A Symplectic Framework for Field Theories}.
\newblock Lecture Notes in Physics, vol.~107.
\newblock Springer-Verlag, Berlin, 1979.

\bibitem{KMS1993}
I.~Kol\'a\v{r}, P.~W. Michor, and J.~Slov\'ak.
\newblock {\em Natural Operations in Differential Geometry}.
\newblock Springer-Verlag, Berlin, 1993.



\bibitem{MaPaSh1998}
J.~E. Marsden, G.~W. Patrick, and S.~Shkoller.
\newblock Multisymplectic geometry, variational integrators, and nonlinear
  {PDE}s.
\newblock {\em Communications in Mathematical Physics},
  199(2):351-395, 1998.


\bibitem{MaPeShWe2001}
J.~E. Marsden, S.~Pekarsky, S.~Shkoller, and M.~West.
\newblock Variational methods, multisymplectic geometry, and continuum
  mechanics.
\newblock {\em Journal of Geometry and Physics}, 38(3-4):253-284, 2001.

\bibitem{MaWe2001}
J.~E. Marsden and M.~West.
\newblock Discrete mechanics and variational integrators.
\newblock {\em Acta Numerica}, 10:357-514, 2001.

\bibitem{RomanRoy2009}
N.~Rom\'an-Roy.
\newblock Multisymplectic Lagrangian and Hamiltonian formalisms of classical
  field theories.
\newblock {\em SIGMA}, 5:100, 25 pages, 2009.

\bibitem{Saunders1989}
D.~J. Saunders.
\newblock {\em The Geometry of Jet Bundles}.
\newblock London Mathematical Society Lecture Note Series, vol.~142.
\newblock Cambridge University Press, Cambridge, 1989.

\bibitem{ShatahStruwe1998}
J.~Shatah and M.~Struwe.
\newblock {\em Geometric Wave Equations}.
\newblock Courant Lecture Notes in Mathematics, vol.~2.
New York University, Courant Institute of Mathematical Sciences,
New York, 1998.

\bibitem{Simo1984}
J.~C. Simo and J.~E. Marsden.
\newblock On the rotated stress tensor and the material version of the
  Doyle-Ericksen formula.
\newblock {\em Archive for Rational Mechanics and Analysis},
  86(3):213-231, 1984.

\bibitem{Simo1988}
J.~C. Simo, J.~E. Marsden, and P.~S. Krishnaprasad.
\newblock The Hamiltonian structure of nonlinear elasticity:
  The material and convective representations of solids, rods, and plates.
\newblock {\em Archive for Rational Mechanics and Analysis},
  104(2):125-183, 1988.

\bibitem{SommerNielsenDarknerPennec2013}
 S.~Sommer, M.~Nielsen, S.~Darkner, and X.~Pennec.
 \newblock Higher-order momentum distributions and locally affine
 {LDDMM} registration.
 \newblock {\em SIAM Journal on Imaging Sciences}, 6(1):341--367, 2013.


\bibitem{Virga1994}
E.~G. Virga.
\newblock {\em Variational Theories for Liquid Crystals}.
\newblock Applied Mathematics and Mathematical Computation, vol.~8.
\newblock Chapman \& Hall, London, 1994.

\end{thebibliography}
\end{document}